\documentclass[10pt]{article}

\usepackage{color}
\usepackage{amsmath}
\usepackage{graphicx}
\usepackage{footnote}
\usepackage{subscript}
\usepackage{amsfonts,color}
\usepackage{amsmath,amssymb}
\usepackage{epsfig}
\usepackage{amssymb} 
\usepackage{mathtools}
\usepackage{amsthm,wasysym}
\usepackage[english]{babel}
\usepackage{cancel}
\makeatletter
\usepackage{float}
\usepackage{wrapfig}
\restylefloat{figure}

\makeatletter
\let\@fnsymbol\@arabic
\makeatother

\usepackage{bbm}
\newcommand{\id}{{\boldsymbol{\mathbbm{1}}}}

\newcommand{\tr}{{\rm tr}}
\newcommand{\dev}{{\rm dev}}
\newcommand{\sym}{{\rm sym}}
\newcommand{\skw}{{\rm skew}}

\newcommand{\Curl}{{\rm Curl}}
\newcommand{\Mprod}[2]{ { \bigl\langle  #1 ,#2\bigr\rangle } }
\newcommand{\norm}[1]{\|#1\|}

\def\dd{\displaystyle}

\newtheorem{theorem}{Theorem}[section]
\newtheorem{lemma}[theorem]{Lemma}
\newtheorem{remark}[theorem]{Remark}
\newtheorem{proposition}[theorem]{Proposition}

\def\barr{\begin{array}}

	\def\sk{\text{skew}}
	
	\def\earr{\end{array}}
\def\bec#1{\begin{equation}\label{#1}}
\def\becn{\begin{equation*}}
\def\endec{\end{equation}}
\def\endecn{\end{equation*}}
\def\dd{\displaystyle}
\def\bfm#1{\mbox{\boldmat}}

\begin{document}
	\title{The isotropic  Cosserat  shell model  including terms up to  $O(h^5)$. \\ Part I: Derivation in matrix notation}
	\author{  Ionel-Dumitrel Ghiba\thanks{Corresponding author: Ionel-Dumitrel Ghiba, \ \  Department of Mathematics, Alexandru Ioan Cuza University of Ia\c si,  Blvd.
			Carol I, no. 11, 700506 Ia\c si,
			Romania; and  Octav Mayer Institute of Mathematics of the
			Romanian Academy, Ia\c si Branch,  700505 Ia\c si, email:  dumitrel.ghiba@uaic.ro} \quad and \quad Mircea B\^irsan\thanks{Mircea B\^irsan, \ \  Lehrstuhl f\"{u}r Nichtlineare Analysis und Modellierung, Fakult\"{a}t f\"{u}r Mathematik,
	Universit\"{a}t Duisburg-Essen, Thea-Leymann Str. 9, 45127 Essen, Germany; and Department of Mathematics, Alexandru Ioan Cuza University of Ia\c si,   Blvd.
	Carol I, no. 11, 700506 Ia\c si,
	Romania;  email: mircea.birsan@uni-due.de}   \quad 
		and  \quad     Peter Lewintan\!\,\thanks{Peter Lewintan,  \ \  Lehrstuhl f\"{u}r Nichtlineare Analysis und Modellierung, Fakult\"{a}t f\"{u}r
		Mathematik, Universit\"{a}t Duisburg-Essen,  Thea-Leymann Str. 9, 45127 Essen, Germany, email: peter.lewintan@uni-due.de}\\  
		and   \quad   Patrizio Neff\,\thanks{Patrizio Neff,  \ \ Head of Lehrstuhl f\"{u}r Nichtlineare Analysis und Modellierung, Fakult\"{a}t f\"{u}r
			Mathematik, Universit\"{a}t Duisburg-Essen,  Thea-Leymann Str. 9, 45127 Essen, Germany, email: patrizio.neff@uni-due.de}
	}

\maketitle
\begin{abstract}
We present a new geometrically nonlinear Cosserat shell model incorporating effects up to order $O(h^5)$ in the  shell thickness $h$. The method that we follow  is an educated 8-parameter ansatz for the three-dimensional elastic shell deformation with attendant analytical thickness integration, which leads us to obtain completely two-dimensional sets of equations in variational form. We give an explicit form of the curvature energy using the orthogonal Cartan-decomposition  of the wryness tensor. Moreover,  we consider the  matrix representation of all tensors in the  derivation of the variational formulation, because this is   convenient when the problem of existence is considered, and it is also preferential for numerical simulations. The step by step construction allows us  to give a transparent approximation of the three-dimensional parental problem. The resulting 6-parameter isotropic shell model  combines membrane, bending and curvature effects at the same time. The Cosserat shell model naturally includes a frame of orthogonal directors, the last of which does not necessarily coincide with the normal of the surface. This rotation-field is coupled to the shell-deformation and augments the well-known Reissner-Mindlin kinematics (one independent director) with so-called in-plane drill rotations, the inclusion of which is a decisive for subsequent numerical treatment and existence proofs.   As a major novelty, we determine the constitutive coefficients of the Cosserat shell model in dependence on the geometry of the shell which are otherwise difficult to guess.

\medskip

\end{abstract}

\tableofcontents

\section{Introduction}

The theory of shells is an important branch of the theory of deformable solids. Its importance resides in the multitude of applications that can be investigated using shell models. In general, the shell and plate theories are intended for the study of \textit{thin bodies}, i.e.,\ bodies in which the thickness in one direction is much smaller than the dimensions in the other two orthogonal directions. As typical examples for shells we mention: roofs of buildings in civil engineering, vehicle bodies in automotive industry, components of wings and propellers  in aerospace industry, cell walls and biological membranes.

The large variety of shell-type structures, such as laminated or functionally graded shells made of advanced materials, like polymer foams or cellular materials, as well as the need to fabricate three-dimensional micro- and nanostructures of various shapes leads to
the necessity of elaborating new adequate models to describe their mechanical behavior.
The process of development of various shell theories is far from  being finalized,
due to the continuous emergence of new technologies in connection with shell manufacturing. For instance, the emerging need to simulate the mechanical response of highly flexible ultra thin structures (allowing easily for finite rotations) and nano-scale thin structures excludes the use of classical infinitesimal--displacement models, either of Reissner--Mindlin or Kichhoff--Love type. In particular, graphene sheets consisting of monolayer atomic arrangement have non-vanishing  ``bending stiffness''. This is at odds with classical thin shell theory, in which, for $h\rightarrow 0$, the bending stiffness should be absent (since the bending terms scale with $h^3$ against $h$ for the membrane strains). Instead, in an extended continuum model like the \textit{Cosserat-type shell model}, there is a ``curvature stiffness'' surviving for $h\rightarrow 0$ related to a characteristic size ${L}_{\rm c}>0$ (internal length parameter, which is related to the microstructure and in principle independent of the thickness).

\subsection{Different approaches to shell modelling}

There are several alternative ways to describe the mechanical behavior of shells and to derive the two-dimensional field equations. The method used in our paper is the so-called \textit{derivation approach}. It starts from a given three-dimensional model of the body and reduces it via physically reasonable constitutive assumptions on the kinematics to a two-dimensional model (i.e.,\ \emph{dimensional reduction}). The philosophy behind the derivation approach is expressed clearly by the grandmaster W.T.\ Koiter as follows \cite[p.~93]{Koiter69}: ``Any two-dimensional theory of thin shells is necessarily of an
{\it approximate character}. An exact two-dimensional theory of shells cannot exist, because the actual body we have to deal with, thin as it may be, is always three-dimensional. [\ldots] Since the theory we have to deal with is approximate in character, we feel that extreme
rigour in its development is hardly desirable. [\ldots] {\it Flexible bodies} like thin shells require a {\it flexible approach}.'' We mention also that the  ``rationale'' of descend from three to two dimensions should always be complemented by an investigation of the intrinsic mathematical properties of the lower dimensional models.

This procedure is opposed to both the \textit{intrinsic approach} and the \textit{asymptotic methods}.  The intrinsic approach is  related to the \textit{derivation approach} but it takes the shell a priori to be a two-dimensional surface (appropriate for modelling graphene sheets, with virtually zero thickness) with additional extrinsic directors in the sense of a deformable Cosserat surface \cite{Cosserat08,Cosserat09,Altenbach-Erem12}. There, two-dimensional equilibrium in appropriate resultant stress and strain variables is postulated ab-initio more or less independently of three-dimensional considerations.
For instance, in the shell model elaborated by Zhilin and Altenbach \cite{Zhilin76,Altenbach04}  each point of the surface is endowed with a triad of orthonormal vectors (called \emph{directors}), which specify the orientation of the material points and describe the \emph{rotations} of shell filaments. Several interesting applications of this approach (also called \emph{directed surfaces}) have been presented in  \cite{Birsan-Alten-2011}. In this approach, the  constitutive parameters appearing in the a priori chosen two-dimensional energy are obtained by fitting the solutions of some specific problems with the solutions obtained by considering the shell  as a three-dimensional body. However, this fitting is usually done only for linear problems, since one reason of being of the nonlinear shell model is that even the classical three-dimensional problem is difficult to be solved for nonlinear problems.  Therefore, it is still missing a complete description of the dependence of the constitutive parameters of the reduced nonlinear two-dimensional problem on the constitutive parameters of the  parental nonlinear three-dimensional problem and on the  mean-curvature ${\rm H}$ and  Gau\ss-curvature ${\rm K}$  of the shell's initially stress-free midsurface. It is possible that using the intrinsic  approach, certain three-dimensional effects may be missing in the derived shell model. Another approach which is also related to the derivation approach is the {\it uniform-approximation technique}, mostly motivated by engineering intuition \cite{wang2019uniformly,taylor2019finite,bose2006material,kienzler2002consistent,kienzler2013theories,birsan2020derivation,shirani2019cosserat}. It   uses polynominal expansions in the
thickness direction both for the displacements and for the stresses  and then it truncates the series expansions.

The procedure of the \textit{asymptotic methods} is  to establish two-dimensional equations by formal expansion of the three-dimensional solutions in power series in terms of a small thickness parameter.
Using the asymptotic methods, a thorough mathematical analysis of linear, infinitesimal displacement shell theory  is presented in \cite{Ciarlet00}
and the extensive references therein. Properly invariant, geometrically nonlinear elastic shell theories are derived by formal asymptotic methods in \cite{Miara98a,Miara98b}, see also \cite{Chapelle-JE-14}.
 The various shell models based on linearized three-dimensional elasticity proposed in the literature have been rigorously justified in those cases, where some normality assumption is introduced, either a priori or as a result of an
asymptotic analysis, see notably the extensive work of Ciarlet and his collaborators \cite{Ciarlet00,Ciarlet96b}, see also \cite{Tambaca-14,Tambaca-16,Tambaca-19}. 
However, even  in the infinitesimal-displacement case it becomes apparent that a model involving membrane and bending simultaneously, cannot be obtained by formal asymptotic methods.  
 In some landmark contributions  \cite{Raoult95a,Mueller02}, see also \cite{Fox93,Marsden96} based on the $ \Gamma$-limit of the three-dimensional model for vanishing thickness, the ``membrane-dominated model" and the ``flexural-dominated model", respectively, are obtained. The difference between asymptotic methods and $\Gamma$-convergence is that the methods based on $\Gamma$-convergence \cite{Raoult95a} lead to another finite-strain membrane term, indicating a nonresistance of the membrane shell in compression. In some  examples of careful modelling, the (derivation) Koiter model \cite{Koiter60,Steigmann13},  is simply the sum of the correctly identified membrane and bending contributions, properly scaled with the thickness (the membrane terms scale with $h$ and bending terms with $h^3$). It is shown in \cite{Ciarlet96c} that the Koiter model is asymptotically at least as good as either the membrane model or the bending model in the respective deformation regimes.
Regarding the bending term, agreement has been reached that the term which is consistent with the three-dimensional  isotropic Saint-Venant--Kirchhoff energy
\begin{align}\label{SVKe}W_{\rm SVK}(F)=\frac{\mu}{4}\, \lVert U^2-\id_3\rVert^2+\frac{\lambda}{8}[\tr(U^2-\id_3)]^2=\frac{\mu}{4}\, \lVert C-\id_3\rVert^2+\frac{\lambda}{8}[\tr(C-\id_3)]^2,\end{align} 
where $C=U^2=F^TF$, is a quadratic expression containing  the second fundamental form ${\rm II}_{y_0}$ of the surface.
Nevertheless,  the model obtained by energy projection in \cite{Marsden96} differs from the results obtained by formal asymptotic analysis in \cite{Fox93}. 
For the finite-strain membrane model, no rigorous justification of the formal asymptotic approach has been given, because of the lack of a
theory which guarantees the well-posedness of the nonlinear three-dimensional problem based on \eqref{SVKe}. Since the membrane terms in a finite-strain properly invariant Kirchhoff-Love shell or finite-strain Reissner-Mindlin model are non-elliptic, the remaining minimization problem is not well-posed even if classical bending is present. By contrast, this is not the case, when a  nonlinear three-dimensional problem in the Cosserat theory is considered \cite{neff2014existence,Tambaca10,tambavca2010semicontinuity}. By ignoring the Cosserat effects, we will be forced to consider a polyconvex  energy \cite{Ball77} in the three-dimensional formulation of the initial problem. These energies do not allow an easy manipulation with respect to  the approaches described in this subsection. In this direction, an example is the article \cite{ciarlet2018existence}, see also \cite{ciarlet2013orientation,bunoiu2015existence}, where the Ciarlet-Geymonat energy \cite{ciarlet1982lois} is used. In these articles,  no reduced completely two-dimensional minimization problem is presented and {no through the thickness integration is performed analytically}. The obtained problems are ``two-dimensional" only in the sense that the final problem is to find three vector fields defined on a bounded open subset of $\mathbb{R}^2$, but all three-dimensional coordinates remain present in the minimization problem. Moreover, the obtained minimization problem is compared with the Koiter model only for small  strain-tensors, situation in which the considered energy is actually the isotropic Saint-Venant--Kirchhoff energy.

In applications there are usually regions of a shell where membrane effects dominate, while bending is dominant in others, but both have to be present in the general model. A fully three-dimensional resolution of a thin shell problem remains elusive at present, notwithstanding the increase in computing power. Hence, there is still a need to come
up with a sound finite-strain shell model, combining both effects of membrane and bending in one system of equations, as the Koiter model does successfully in the infinitesimal-displacement case.   

There are numerous proposals in the engineering literature for such a finite-strain, geometrically nonlinear shell formulation. In many cases, the need has been felt to devote attention to the rotation field $\overline{R} \in {\rm SO}(3)$, since rotations are locally the dominant deformation mode of a thin flexible structure. We also mention that considering the new unknown field $\overline{R}\in {\rm SO}(3)$, we are able to keep more three-dimensional effects into our two-dimensional variational formulation  (local independent rotations in various directions, i.e., each material point of the body has the degrees of freedom of a rigid body), in contrast to the shell models based on classical elasticity.  In fact, this is one of the first raison d'\^etre of the Cosserat theory \cite{Cosserat08,Cosserat09}. This has led to shell models which include the so-called \emph{drilling rotations} \cite{wisniewski2010finite}, meaning that in-plane rotations about the shell filament are also taken into account.
In \cite{Hughes92}   it is shown  that the inclusion of drilling rotations in the model has a beneficial influence on the numerical implementation. However, a mathematical analysis for such a family of finite-strain curved shell models is, as of yet, still missing.

One of the most general and effective approaches to shells is the so called geometrically nonlinear 6-parameter  resultant shell theory, which was originally proposed by Reissner \cite{Reissner74} and has subsequently been extended considerably. An account of these developments and main achievements have been presented in the works of Libai and Simmonds \cite{Libai98}
and Pietraszkiewicz et al.\ \cite{Pietraszkiewicz14,Eremeyev06}. This model involves two independent kinematic fields:
the translation vector field and the rotation tensor field (in total six independent scalar kinematic variables). The two-dimensional equilibrium equations and static boundary conditions of the shell are derived exactly by direct through-the-thickness integration of the stresses in the three-dimensional balance laws of linear and angular momentum. The kinematic fields are then constructed directly on the two-dimensional level using the integral identity of the virtual work principle. Following this procedure, the two-dimensional model is expressed in terms of stress resultants and work-averaged deformation fields defined on the shell base surface. The kinematical structure of this 6-parameter model (involving the translation vector and rotation tensor) is identical to the kinematical structure of Cosserat shells (defined as material surfaces endowed with a triad of rigid directors describing the orientation of points). Several developments of this model and applications to complex shell problems, including phase transition and multifold shells, together with the finite element implementation, have been presented by Pietraszkiewicz, Eremeyev and Chr\'o\'scielewski with their co-workers in a number of papers \cite{Eremeyev06,Eremeyev11,Chroscielewski11,Pietraszkiewicz14,Chroscielewski15}, see also \cite{pimenta2010fully}.
\medskip

\subsection{The new shell model presented in this article}

In this paper, we extend the Cosserat plate model established by Neff in his habilitation thesis in 2003  \cite{Neff_plate04_cmt,Neff_plate07_m3as,Neff_membrane_Weinberg07,Neff_membrane_plate03,Neff_membrane_existence03} to the general case of curved initial shell configurations. The results have been previously announced (using a succinct tensor notation) in \cite{birsan2019refined,neff2019higher} and the aim of the current article is to explain the derivation of the model in more details, with added transparency, and using only the matrix representation. Moreover, all our calculations do not use curvilinear coordinates and do not explicitly use an a priori parametrization of the mid-surface, since at the very beginning we are starting by considering the general form of the three-dimensional deformation energy of a fully three-dimensional body, without involving the informations about its mid-surface in the variational formulation. To be more precise, our initial constitutive  assumptions (the form of the energy) do not, of course, depend on the shape  of the three-dimensional body.

We start with a suitable bulk three-dimensional isotropic Cosserat model written in Cartesian coordinates, we  make an appropriate ansatz for the deformation and rotation functions and we perform the integration over the thickness.   The form we take for the three-dimensional Cosserat energy is already a strong point of our new approach, since it will help us to give an explicit analytic form of the entire energy of the shell model, and therefore it will be useful 
in analytical and numerical studies.  More precisely, we  give an explicit form of the curvature energy using the orthogonal Cartan-decomposition  of the wryness tensor (the used curvature tensor of the Cosserat bulk model). In the modelling process we follow the derivation approach as described for planar configurations in \cite{Neff_plate04_cmt}, but we need to additionally incorporate the curvature effects by using known results from the differential geometry of surfaces in $\mathbb{R}^3$.
Thus, we  obtain a geometrically nonlinear formulation for Cosserat-type shells with 6 independent kinematical variables: 3 for displacements and 3 for finite rotations (including drilling rotations).

The new model will resolve some shortcomings of classical approaches, which we have mentioned in the previous subsections. In particular, it  satisfies the following requirements which we deem necessary for an effective general shell model\vspace{-5pt}:
\begin{itemize}
\item a geometrically nonlinear formulation that allows for finite rotations.\vspace{-5pt}
\item the description of transverse shear, drilling rotations, thickness stretch and asymmetric shift of the midsurface.\vspace{-5pt}
\item a hyperelastic, variational formulation with second-order Euler-Lagrange equations in view of an efficient finite element implementation.\vspace{-5pt}
\item a dimensionally reduced energy density which is entirely defined in terms of two-dimensional quantities with a clear physical meaning and by a step by step construction.\vspace{-5pt}
\item well-posedness: existence of solutions, but not unqualified uniqueness in order to be able to describe buckling due to membrane forces (e.g., under lateral compression).\vspace{-5pt}
\item the consistency with classical shell models for infinitesimal deformations.\vspace{-5pt}
\item the incorporation of non-classical size effects, such that graphene sheets have bending/curvature resistance.
\end{itemize}

Since we begin with a Cosserat bulk model which already contains in its formulation the so-called \emph{Cosserat couple modulus} $\,\mu_{\rm c}\ge 0$\, and an internal length $\,{L}_{\rm c}>0\,$ (which is characteristic for the material, e.g.,\ related to the grain size in a polycrystal), the reduced energy density for shells will also include the material parameters $\,\mu_{\rm c}$\, and  $\,{L}_{\rm c}\,$, in conjunction with specific terms having a clear physical meaning, expressed as functions of two-dimensional quantities. The internal length parameter $\,{L}_{\rm c}\,$ allows for the incorporation of non-classical size effects in the shell model (in the sense that smaller samples are relatively stiffer than larger samples \cite{Cirak03}).

For very irregular and curved initial shell configurations it is not at all clear which terms get small in a thin shell approximation. Moreover, when another dimension (beside the thickness) of the parental three-dimensional body is very small or  when the deformations are large compared to the thickness, terms of order $O(h^3)$ may not be  sufficient to capture important three-dimensional behaviour. Therefore, we aim to elaborate a complete and consistent model up to the order $O(h^5)$, i.e.,\ we determine all the terms up to the order $O(h^5)$ in both the  membrane part   and the bending-curvature part of the energy density. Indeed, in a shell model $h$ is very small, but the reason of being of a shell model is to obtain an approximation of the deformation of a three-dimensional body. By considering  terms up to the order $O(h^5)$, some additional three-dimensional effects are not omitted in the obtained two-dimensional model. The used method  allows this construction in a very transparent way without  considering at the very beginning a two-dimensional problem, approach in which terms of order  $O(h^5)$ cannot be  guessed a priori. Moreover, the coefficients of the terms in the energy density  depend on the mean-curvature ${\rm H}$ and  Gau\ss-curvature ${\rm K}$ of the shell's midsurface, the calculations showing us that these coefficients have unforeseeable expressions.  Thus, we come up with an improved model which should generalize most of the known variants of shell theory (since they consider only terms of order $O(h^3)$, see \cite{Steigmann13}). In this respect we will deliver more accurate qualitative and numerical results in forthcoming papers.

 We regard also other shell models from the literature as particular cases of our formulation. For instance, the 6-parameter resultant shell theory \cite{Eremeyev06} can be viewed as a special case, since  it is a theory of order $O(h^3)$, it omits all mixed terms, it is not elaborated starting from a three-dimensional parental problem, and their constitutive coefficients are not expressed in terms of the mean-curvature and Gau\ss-curvature and not in terms of the constitutive coefficients of the  three-dimensional  internal energy.

The present paper is completely self-contained and can be read also by researchers not yet accustomed to the specific notation and usages of shell-theory. We arrive at this point at the expense of working, as far as possible, with concepts from 3D-elasticity theory as well as consequently  utilizing ``reconstructed'' 3D-matrix  objects. Thus the paper is ideally suited for researchers in need of quickly understanding the basic   ingredients  of a geometrically nonlinear shell theory. 
In  forthcoming papers we will compare a suitable restriction of our modellling framework with the geometrically nonlinear $O(h^3)$-Koiter model. Preliminary observations suggests that our model (restricted to the same order $O(h^3)$) includes terms not present in the standard Koiter model (isotropic Kirchhoff-Love shell). For developable surfaces (Gau\ss-curvature ${\rm K}=0$) and after linearisation, both approaches seem to coincide. We will also investigate a corresponding $\Gamma$-convergence result, similar to \cite{Neff_Danzig09,Neff_Hong_Reissner08,Neff_Chelminski_ifb07}.

The considered  matrix representation, in the entire derivation of the model,   is more convenient when the problem of existence is considered, it is also preferred for numerical simulations in the engineering community and it offers some details about how the  elastic shell model obtained in the present article may be extended to an elasto-plastic shell model  \cite{HutterSFB02,Neff_Habil04} or to a model which includes  residual (initial) stresses \cite{marohnic2015model}. These subjects will be considered in future works based on our model. {For instance, in Part II \cite{GhibaNeffPartII} we show that the expression of the energy allows us to have a decent control on each term of the energy density, in order to show the coercivity and the convexity of the energy, and finally to show the existence of minimisers.} These types of thin bodies are of great importance nowadays, in view of the new shell manufacturing procedures and we believe that the terms of order $O(h^5)$ included here  will play  important roles  in increasing the accuracy of analytical and numerical predictions  
in these industrial processes. In forthcoming papers  we will prove the existence of the solution of the obtained minimization problem \cite{GhibaNeffPartII} (at order $O(h^3)$ and $O(h^5)$), we will offer some numerical simulations similar to \cite{Sander-Neff-Birsan-15} in order to compare our model with some other previous models from literature. Moreover, the pure elastic nonlinear Cosserat shell model will also be extended to viscoelasticity and multiplicative plasticity \cite{Neff_membrane_plate03,sauer2019multiplicative,roychowdhury2019growth} and it will allow us to discuss residual stress effects in applications to design-control problems of nano-three-dimensional objects \cite{danescu2020shell}, situations in which a model up to order $O(h^5)$ is useful since  another dimension (beside thickness) may be very small or the deformations are large compared to thickness. Therefore, a model up to order $O(h^5)$ may be very useful in the study of thin bodies with a relative not ``so small'' thickness  compared to the other two dimensions, e.g., in the study of nano-structures.
\section{The three-dimensional formulation }\setcounter{equation}{0}
In \cite{Neff_plate04_cmt} a physically linear, fully frame-invariant isotropic Cosserat model is introduced. The problem has been posed
in a variational setting. We consider a three-dimensional {\it shell-like thin domain} $\Omega_\xi\subset\mathbb{R}^3.$ A generic point of $\Omega_\xi$ will be denoted by $(\xi_1,\xi_2,\xi_3)$ in a fixed  standard base vector  $e_1, e_2, e_3$   of $
\mathbb{R}^3$. The elastic material constituting the shell is assumed to be homogeneous and isotropic and the reference configuration $\Omega_\xi$ is assumed to be a natural (stress-free) state. 
In the rest of the present article we use the notation given in Appendix \ref{Sectnotation}.

The deformation of the body occupying the domain $\Omega_\xi$ is described by a map $\varphi_\xi$ (\textit{called deformation}) and by a \textit{microrotation} $\overline{R}_\xi$,
\begin{equation}
\varphi_\xi:\Omega_\xi\subset\mathbb{R}^3\rightarrow\mathbb{R}^3, \qquad  \ \overline{R}_\xi:\Omega_\xi\subset\mathbb{R}^3\rightarrow {\rm SO}(3)\, .
\end{equation}
We denote the current configuration by $\Omega_c:=\varphi_\xi(\Omega_\xi)\subset\mathbb{R}^3$. 
The deformation and the microrotation is solution of the following \textit{geometrically nonlinear minimization problem} posed on $\Omega_\xi$:
\begin{equation}\label{minprob}
I(\varphi_\xi,F_\xi,\overline{R}_\xi, \alpha_\xi)=\dd\int_{\Omega_\xi}\left[W_{\rm{mp}}(\overline U _\xi)+
W_{\rm{curv}}(\alpha_\xi)\right]dV(\xi)
- \Pi(\varphi_\xi,\overline{R}_\xi)\quad 
{\to}
\textrm{\ \ min.} \quad  {\rm   w.r.t. }\quad (\varphi_\xi,\overline{R}_\xi)\, ,
\end{equation}
where
\begin{align}
 F_\xi:\,=\,&\nabla_\xi\varphi_\xi\in\mathbb{R}^{3\times3}\, \qquad \qquad \qquad \qquad \qquad \textrm{(the deformation gradient)},  \notag\\
\overline U _\xi:\,=\,&\dd\overline{R}^T_\xi F_\xi\in\mathbb{R}^{3\times3} \ \ \  \ \, \, \  \qquad \qquad \qquad \qquad \textrm{(the non-symmetric Biot-type stretch tensor)},  \notag\\
 \alpha_\xi:\,=\,&\overline{R}_\xi^T\, \Curl_\xi \,\overline{R}_\xi\in\mathbb{R}^{3\times 3} \,  \qquad \qquad\qquad  \quad  \textrm{(the second order  dislocation density tensor \cite{Neff_curl08})}\, ,  \\
\dd W_{\rm{mp}}(\overline U _\xi):\,=\,&\dd\mu\,\lVert \dev\,\text{sym}(\overline U _\xi-\id_3)\rVert^2+\mu_{\rm c}\,\lVert \text{skew}(\overline U _\xi-\id_3)\rVert^2+
\dd\frac{\kappa}{2}\,[{\rm tr}(\text{sym}(\overline U _\xi-\id_3))]^2\ \ \, \textrm{(physically linear)}\, ,  \notag\\
\dd W_{\rm{curv}}( \alpha_\xi):\,=\,&\mu\,{L}_{\rm c}^2\left( b_1\,\lVert \dev\,\text{sym}\, \alpha_\xi\rVert^2+b_2\,\lVert \text{skew}\, \alpha_\xi\rVert^2+  b_3\,
[{\rm tr}(\alpha_\xi)]^2\right)\qquad \qquad\qquad\qquad\quad \textrm{(curvature energy)}, \notag
\end{align}
and $dV(\xi)$ denotes the  volume element in the $\Omega_\xi$-configuration. 

The total elastically stored energy $W=W_{\rm mp}+W_{\rm curv}$ depends on
the deformation gradient $F_\xi$ and microrotations $\overline{R}_\xi$ together with their spatial
derivatives. In general, the \textit{{Biot-type stretch tensor}} $\overline U _\xi$ is not symmetric (the first Cosserat deformation tensor \cite{Cosserat09}). The parameters $\mu$ and $\lambda$ are the \textit{Lam\'e constants}
of classical isotropic elasticity, $\kappa=\frac{2\mu+3\lambda}{3}$ is the \textit{infinitesimal bulk modulus}, $b_1, b_2, b_3$ are \textit{non-dimensional constitutive curvature coefficients (weights)}, $\mu_{\rm c}\geq 0$ is called the \textit{{Cosserat couple modulus}} and ${L}_{\rm c}>0$ introduces an \textit{{internal length} } which is {characteristic} for the material, e.g., related to the grain size in a polycrystal. The
internal length ${L}_{\rm c}>0$ is responsible for \textit{size effects} in the sense that smaller samples are relatively stiffer than
larger samples. If not stated otherwise, we assume that $\mu>0$, $\kappa>0$, $\mu_{\rm c}>0$, $b_1>0$, $b_2>0$, $b_3> 0$.  

The form of the curvature energy  $W_{\rm curv}$ is not that originally considered in \cite{Neff_plate04_cmt}.  Indeed, Neff \cite{Neff_plate04_cmt} uses a curvature energy expressed in terms of the \textit{third order curvature tensor} $\mathfrak{K}=(\overline{R}_\xi^T\nabla (\overline{R}_\xi. e_1)\,|\,\overline{R}_\xi^T\nabla (\overline{R}_\xi. e_2)\,|\,\overline{R}_\xi^T\nabla (\overline{R}_\xi. e_3))$. As we will remark in Section 3, the new form of the energy based on the \textit{second order dislocation density tensor} $\alpha_\xi$ simplifies considerably
the representation by admitting to use the orthogonal decomposition
\begin{align}\overline{R}_\xi^T\, \Curl_\xi \,\overline{R}_\xi=\alpha_\xi=\dev\, \sym \,\alpha_\xi+\skw\, \alpha_\xi+ \frac{1}{3}\,\tr(\alpha_\xi)\id_3.\end{align} Moreover, it yields an equivalent control of spatial derivatives of rotations \cite{Neff_curl08} and allows us to write the curvature energy   in a fictitious Cartesian configuration in terms of the so-called \textit{wryness tensor}.  This fact has some further implications, e.g., the coupling between the membrane part, the membrane-bending part, the bending-curvature part and the curvature part of the energy of the shell model is transparent and will coincide with shell-bending curvature tensors elsewhere considered \cite{Eremeyev06}.

In \eqref{minprob}, $ \Pi(\varphi_\xi,\overline{R}_\xi) $ is the external loading potential, which admits the following additive decomposition:
\begin{equation}\label{loadpot1}
\Pi(\varphi_\xi,\overline{R}_\xi)
= \Pi_f(\varphi_\xi)+ \Pi_t(\varphi_\xi)+ \Pi_\Omega(\overline{R}_\xi)+ \Pi_{\partial\Omega_t}(\overline{R}_\xi)
\, ,
\end{equation}
where 
\begin{align}
\Pi_f(\varphi_\xi):\,=\,&\dd\int_{\Omega_\xi} \bigl\langle  f,  u \bigr\rangle   \, dV(\xi)\,= \textrm{potential of external applied body forces $ f $},  \notag\\
\Pi_t(\varphi_\xi):\,=\,&\dd\int_{\partial\Omega_t} \bigl\langle  t,  u \bigr\rangle   \, dS(\xi)\,= \textrm{potential of external applied boundary forces $ t $}\, ,  \label{loadpot2}\\
\Pi_{\Omega}(\overline{R}_\xi):\,=\,&\textrm{ potential of external applied body couples}\, ,  \notag\\
\Pi_{\partial\Omega_t}(\overline{R}_\xi):\,=\,&\textrm{ potential of external applied boundary couples}\, ,  \notag
\end{align}
and $ u= \varphi_\xi - \xi$ is the displacement vector, $ \partial\Omega_t $ is a subset of the boundary of $ \Omega_\xi $\,, and $ dS(\xi) $ is the area element.

\subsection{Relation to the Biot nonlinear elasticity model}
The used three-dimensional Cosserat model can be seen as an extension of the geometrically nonlinear isotropic Biot-model. Indeed, letting formally\footnote{In order to arrive at the limit Biot model for $\lambda=0$, it is sufficient to consider $L_{\rm c}\to 0$ and $\mu_{\rm c}\geq \mu$, see \cite{fischle2017geometricallyI,fischle2017geometricallyII,neff2019explicit,borisov2019optimality}\,.} $\mu_{\rm c}\to +\infty$ and $L_{\rm c}\to 0$, the independent rotation field $\overline{R}\to {\rm polar}(F)$ must coincide with the continuum rotation in the polar decomposition of 
$
	F=R\, U={\rm polar}(F)\, \sqrt{F^T F}.
$
Since for $L_{\rm c}\to 0$  curvature is absent, the resulting minimization problem is
\begin{equation}
\dd\int_{\Omega_\xi}\left[W_{\rm{Biot}}(F)- \bigl\langle  f, \varphi\bigr\rangle  \right]dV(\xi)\quad 
{\to}
\textrm{\ \ min.} \quad  {\rm   w.r.t. }\quad (\varphi_\xi)\, ,
\end{equation}
where
\begin{align}\label{Bioten}
W_{\rm{Biot}}(F)\,=\, \mu\,\lVert U-\id_3\rVert^2+\frac{\lambda}{2}\,[\tr(U-\id_3)]^2=\dd\mu\,\lVert \dev\,\text{sym}( U -\id_3)\rVert^2+
\dd\frac{\kappa}{2}\,[{\rm tr}(\text{sym}( U-\id_3))]^2.
\end{align}
Recall that typically, the Koiter shell-model is obtained based on the dimension reduction from the isotropic Saint-Venant--Kirchhoff energy
\begin{align}\label{SVKen}W_{\rm SVK}(F)=\frac{\mu}{4}\, \lVert U^2-\id_3\rVert^2+\frac{\lambda}{8}[\tr(U^2-\id_3)]^2=\frac{\mu}{4}\, \lVert C-\id_3\rVert^2+\frac{\lambda}{8}[\tr(C-\id_3)]^2,\end{align} 
where $C=U^2=F^TF$. Both energies \eqref{Bioten} and \eqref{SVKen} are linearisation-equivalent and meant to well-capture the small strain regime expected for the response of a thin shell. However, $W_{\rm SVK}(F)$ introduces physically unacceptable behaviour under the slightest compression (compression would be softer than tension). Since $W_{\rm Biot}$ does not have this feature, we believe that arriving with our model at $W_{\rm Biot}$ is advantageous. 

Preliminary calculations show us that, in some particular cases, the total energy of the Cosserat-shell model constructed by using the Biot energy reduces to quadratic and bilinear forms in terms of the   difference of the squares of the first fundamental forms (of the initial configuration and of the current configuration) and/or in terms of the   difference of the second fundamental forms. This is consistent with  new estimates of the distance between two surfaces obtained in 
\cite{ciarlet2015nonlinear,ciarlet2019new} which anticipate the need for  this type of energies. 
\section{Transformed variational problem in   the fictitious  configuration $\Omega_h$}\setcounter{equation}{0}

In what follows, we assume that the parameter domain $\Omega_h\subset\mathbb{R}^3$ is a right cylinder of the form
$$\Omega_h=\left\{ (x_1,x_2,x_3) \,\Big| \, (x_1,x_2)\in\omega,  \, -\dfrac{h}{2}\,< x_3<\, \dfrac{h}{2}\, \right\} = \,\dd\omega\,\times\left(-\frac{h}{2}, \frac{h}{2}\right),$$
where  $\omega\subset\mathbb{R}^2$ is a bounded domain with Lipschitz boundary
$\partial \omega$ and the constant length $h>0$ is the \textit{thickness of the shell}.
For shell--like bodies we consider   the  domain $\Omega_h $ to be {thin}, i.e., the thickness $h$ is {small}. Thus, the domain $\Omega_h $ can be viewed as a \textit{fictitious Cartesian configuration} of the body.

We assume furthermore that there exists a given $C^1$-diffeomorphism $\Theta:\mathbb{R}^3\rightarrow\mathbb{R}^3$, which maps the fictitious Cartesian parameter space $\Omega_h$ with coordinates $(x_1,x_2,x_3)\in\mathbb{R}^3$ onto $\Theta(x_1,x_2,x_3)=(\xi_1,\xi_2,\xi_3)$ such that  the initially curved reference configuration of the shell is $\Theta(\Omega_h)=\Omega_\xi$ (see Figure \ref{Fig1}).
\begin{figure}
	\begin{center}
	\includegraphics[scale=1.6]{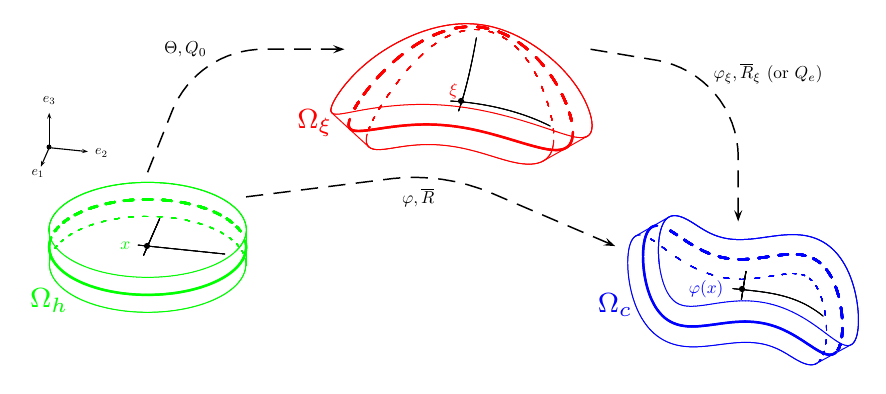}
 		\caption{\footnotesize The shell in its   initial configuration $\Omega_\xi$, the shell in the deformed configuration $\Omega_c$, and the fictitious planar cartesian reference configuration $\Omega_h$. Here,  $\overline{R}_{\xi}  $ is the elastic rotation field, ${Q}_{0}$ is the  initial rotation from the fictitious planar cartesian reference configuration to the initial  configuration $\Omega_\xi$, and $\overline{R}$ is the total rotation field from the fictitious planar cartesian reference configuration to the deformed configuration $\Omega_c$.}
		\label{Fig1}       
	\end{center}
\end{figure}

Now, let us  define the map
\begin{equation}
\varphi:\Omega_h\rightarrow \Omega_c,\ \ \qquad\quad  \varphi(x_1,x_2,x_3)=\varphi_\xi( \Theta(x_1,x_2,x_3)).
\end{equation}
We view $\varphi$ as a function which maps the fictitious  planar reference configuration $\Omega_h$ into the deformed (current) configuration $\Omega_c$. Hence, the guiding question is: how  can we construct the map $\varphi$ and a total rotation tensor $\overline{R}$ in order to reduce suitably the three-dimensional problem to a two-dimensional problem? To answer  this question (see Figure \ref{Fig1}) we reformulate the minimization principle  (\ref{minprob}) in the fictitious, Cartesian  configuration $\Omega_h$.
If we construct such  mappings, since the diffeomorphism $\Theta$ is considered known, then we also know the map which describes the deformation of the initial curved reference configuration $\Omega_\xi$ into the current configuration $\Omega_c$ of the body.

Assume an underlying three-dimensional deformation of the shell-like body $\varphi_\xi$ is known and differentiable. Consider a point $\beta=(x_1,x_2,0)\in \omega\times\{0\}$ and $\Theta(\beta)$. For the moment, we do not assume that $\Theta(\beta)$ is mapped to the midsurface of $\Omega_\xi=\Theta(\Omega_h)$. Consider also the point $\beta_{x_3}=(x_1,x_2,x_3)$, i.e., the line $\beta\beta_{x_3}$ is normal to $\omega$. Then, we have the expansion
\begin{align}
\varphi_\xi(\Theta(\beta_{x_3}))&=\varphi_\xi(\Theta(\beta))+x_3 \nabla_\xi \varphi_\xi (\Theta(\beta))\nabla_x\Theta(\beta)\, e_3+o(x_3)\, \qquad\qquad  \text{or} \\\nonumber
\varphi_\xi(\Theta(x_1,x_2,x_3))&=\underbrace{\varphi_\xi(\Theta(x_1,x_2,0))}_{:=\,m_\xi}+x_3 \nabla_\xi \varphi_\xi (\Theta(x_1,x_2,0))\nabla_x\Theta(x_1,x_2,0)+o(x_3)\, ,
\end{align}
where $\Theta(x_1,x_2,0)$ does not belong to  the midsurface of $\Omega_\xi$, but to  the transformed midsurface $ \omega_\xi= \Theta(\omega\times\{0\})$, as long as $\Theta$ has not a specific expression.

\subsection{Transformation of the minimization problem}

Consider the \textit{elastic microrotation}
\begin{equation}
\overline{Q}_e:\Omega_h\rightarrow{\rm SO}(3),\ \ \ \ \ \overline{Q}_e(x_1,x_2,x_3):=\overline{R}_\xi(\Theta(x_1,x_2,x_3))\, 
\end{equation}
and \textit{the elastic (non-symmetric) Biot-type stretch tensor} (the elastic first Cosserat deformation tensor)
\begin{equation}
\overline U _e:\Omega_h\rightarrow{\rm Sym}(3),\ \ \ \ \ \overline U _e(x_1,x_2,x_3):=\overline U _\xi(\Theta(x_1,x_2,x_3))\,. 
\end{equation}
We use the polar decomposition \cite{neff2013grioli} of $\nabla_x \Theta$  and write 
\begin{equation}\label{dec}
\nabla_x \Theta\,=\,{Q}_0\, U_0\, ,\qquad 
{Q}_0\,=\,{\rm polar}{(\nabla_x \Theta)}\,=\,{\rm polar}{([\nabla_x \Theta]^{-T})}\in {\rm SO}(3 ),\qquad   U _0\in \rm{Sym}^+(3).
\end{equation}
Corresponding to the elastic  deformation process, we have the total microrotation
\begin{align}
\overline{R}&:\Omega_h\rightarrow{\rm SO}(3),\qquad \overline{R}(x_1,x_2,x_3)\,=\,\overline{Q}_e(x_1,x_2,x_3)\,{Q}_0(x_1,x_2,x_3).
\end{align}
Obviously,  if we know the total microrotation $\overline{R}$, then we know  the microrotation $\overline{R}_\xi$.
Using  the chain rule
\begin{align}
\partial_{x_k}\varphi\,=\,\sum_{i=1}^3 \partial _{\xi_i}\varphi_\xi\,\partial_{x_k}\xi_i,
\qquad \qquad 
\nabla_x\varphi(x_1,x_2,x_3)&\,=\,\nabla_\xi\varphi_\xi (\Theta(x_1,x_2,x_3))\nabla_x \Theta(x_1,x_2,x_3)\, ,
\end{align}
we deduce (the multiplicative decomposition)
\begin{align}
F(x_1,x_2,x_3)&\,=\,F_\xi (\Theta(x_1,x_2,x_3))\,\nabla_x \Theta(x_1,x_2,x_3),
\qquad \text{where}\quad 
F\,=\,\nabla_x \varphi,\\
F_\xi (\Theta(x_1,x_2,x_3))&\,=\,F(x_1,x_2,x_3)\,[\nabla_x \Theta(x_1,x_2,x_3)]^{-1}.\notag
\end{align}
Therefore, the elastic non-symmetric stretch tensor is given by
\begin{align}
\overline U _e\,=\,\overline{Q}_e^T\, F\,[\nabla_x\Theta]^{-1}\,=\,{Q}_0\,\overline{R}^T\, F\,[\nabla_x\Theta]^{-1}.
\end{align}

As a  Lagrangian strain measure for curvature (orientation change) one can also employ the
so-called \textit{{wryness tensor}} (second order tensor) \cite{Neff_curl08,Pietraszkiewicz04}
\begin{align}
\Gamma_\xi&:= \Big(\mathrm{axl}(\overline{R}_\xi^T\,\partial_{\xi_1} \overline{R}_\xi)\,|\, \mathrm{axl}(\overline{R}_\xi^T\,\partial_{\xi_2} \overline{R}_\xi)\,|\,\mathrm{axl}(\overline{R}_\xi^T\,\partial_{\xi_3} \overline{R}_\xi)\,\Big)\in \mathbb{R}^{3\times 3},
\end{align}
since (see \cite{Neff_curl08})  the following close relationship between the wryness tensor
and the dislocation density tensor holds
\begin{align}\label{Nye1}
\alpha_\xi\,=\,-\Gamma_\xi^T+\tr(\Gamma_\xi)\, \id_3\,, \qquad\textrm{or equivalently},\qquad \Gamma_\xi\,=\,-\alpha_\xi^T+\frac{1}{2}\tr(\alpha_
\xi)\, \id_3\,.
\end{align}
For infinitesimal strains this formula is well-known under
the name Nye's formula, and $-\Gamma$ is also called Nye's curvature tensor. We will use this terminology further on \cite{Neff_curl08}.

Note that 
$
\partial_{x_k}\overline{Q}_e=\sum_{i=1}^3 \partial_{\xi_i}\overline{R}_\xi\,\partial_{x_k}\xi_i,$  $\partial_{\xi_k}\overline{R}_\xi=\sum_{i=1}^3 \partial_{x_i}\overline{Q}_e\,\partial_{\xi_k}x_i
$
and
\begin{align}
\overline{R}_\xi^T\,\partial_{\xi_k} \overline{R}_\xi\,=\,&\sum_{i=1}^3(\overline{Q}_e^T \partial_{x_i}\overline{Q}_e)\,\partial_{\xi_k}x_i=\sum_{i=1}^3\
\big(\overline{Q}_e^T\,\partial_{x_i} \overline{Q}_e\big)([\nabla_x \Theta]^{-1})_{ik}\\
\mathrm{axl}\big(\overline{R}_\xi^T\,\partial_{\xi_k} \overline{R}_\xi\big)\,=\,&\sum_{i=1}^3\
\mathrm{axl}\big(\overline{Q}_e^T\,\partial_{x_i} \overline{Q}_e\big)([\nabla_x \Theta]^{-1})_{ik}.\notag
\end{align}
Thus, we have the chain rule
\begin{align}\label{crg}
\Gamma_\xi\,=\,&
\Big( \sum_{i=1}^3{\rm axl}\
\Big(\overline{Q}_e^T\,\partial_{x_i} \overline{Q}_e\Big)([\nabla_x \Theta]^{-1})_{i1}\,\Big|\,  \sum_{i=1}^3{\rm axl}\
\big(\overline{Q}_e^T\,\partial_{x_i} \overline{Q}_e\big)([\nabla_x \Theta]^{-1})_{i2}\,\Big|\, \sum_{i=1}^3{\rm axl}\
\big(\overline{Q}_e^T\,\partial_{x_i} \overline{Q}_e\big)([\nabla_x \Theta]^{-1})_{i3}\Big)\vspace{1.5mm}\notag\\
\,=\,&\Big(\mathrm{axl}(\overline{Q}_e^T\,\partial_{x_1} \overline{Q}_e)\,|\, \mathrm{axl}(\overline{Q}_e^T\,\partial_{x_2} \overline{Q}_e)\,|\,\mathrm{axl}(\overline{Q}_e^T\,\partial_{x_3} \overline{Q}_e)\,\Big)[\nabla_x \Theta]^{-1}.
\end{align}
Define 
$
\Gamma_e:=\Big(\mathrm{axl}(\overline{Q}_e^T\,\partial_{x_1} \overline{Q}_e)\,|\, \mathrm{axl}(\overline{Q}_e^T\,\partial_{x_2} \overline{Q}_e)\,|\,\mathrm{axl}(\overline{Q}_e^T\,\partial_{x_3} \overline{Q}_e)\,\Big), \ 
\alpha_e:=\overline{Q}_e^T\Curl_x\,\overline{Q}_e.\notag
$
Using   Nye's formula for $\alpha_e$ and $\Gamma_e$, we have the correspondence
\begin{align}
\alpha_e\,=\,-\Gamma_e^T+\tr(\Gamma_e)\, \id_3, \qquad\textrm{or equivalently,}\qquad \Gamma_e\,=\,-\alpha_e^T+\frac{1}{2}\tr(\alpha_
e)\, \id_3.
\end{align}
In view of \eqref{Nye1} and \eqref{crg}, we can write
\begin{align}
\alpha_\xi\,=\,&-\Gamma_\xi^T+\tr(\Gamma_\xi)\, \id_3\,=\,-(\Gamma_e\, [\nabla_x \Theta]^{-1})^T+\tr(\Gamma_e\, [\nabla_x \Theta]^{-1})\, \id_3\notag\vspace{1.5mm}\\
\,=\,&
-[\nabla_x \Theta]^{-T}\Gamma_e^T+\tr(\Gamma_e)\, [\nabla_x \Theta]^{-T}-\tr(\Gamma_e)\, [\nabla_x \Theta]^{-T}+\tr(\Gamma_e\, [\nabla_x \Theta]^{-1})\, \id_3
\notag\vspace{1.5mm}\\
\,=\,&\,[\nabla_x \Theta]^{-T}
\alpha_e-\tr(\Gamma_e)\, [\nabla_x \Theta]^{-T}+\tr(\Gamma_e\, [\nabla_x \Theta]^{-1})\, \id_3.
\end{align}
Moreover
\begin{align}
\tr(\Gamma_e)&\,=\,-\tr(\alpha_e)+\frac{3}{2}\tr(\alpha_
e)\,=\,\frac{1}{2}\tr(\alpha_
e),\vspace{1.5mm}\notag\\
\Gamma_e\, [\nabla_x \Theta]^{-1}&\,=\,-\alpha_e^T\, [\nabla_x \Theta]^{-1}+\frac{1}{2}\tr(\alpha_
e)\, [\nabla_x \Theta]^{-1},\vspace{1.5mm}\\
\tr(\Gamma_e\, [\nabla_x \Theta]^{-1})&\,=\,-\tr(\alpha_e^T\, [\nabla_x \Theta]^{-1})+\frac{1}{2}\tr(\alpha_
e)\, \tr([\nabla_x \Theta]^{-1})\notag.
\end{align}
We deduce
\begin{align}
\alpha_\xi\,=\,&\,[\nabla_x \Theta]^{-T}
\alpha_e-\frac{1}{2}\tr(\alpha_
e)\, [\nabla_x \Theta]^{-T}-\tr(\, [\nabla_x \Theta]^{-T}\alpha_e)\, \id_3+\frac{1}{2}\tr(\alpha_
e)\, \tr([\nabla_x \Theta]^{-1})\, \id_3\vspace{1.5mm}\notag\\
\,=\,&\,[\nabla_x \Theta]^{-T}
\alpha_e-\tr(\alpha_e^T\,[\nabla_x \Theta]^{-1})\, \id_3-\frac{1}{2}\tr(\alpha_
e)\, \Big\{[\nabla_x \Theta]^{-T}-\, \tr([\nabla_x \Theta]^{-1})\, \id_3\Big\}.
\end{align}
However, we will not use this formula to rewrite the curvature energy in the fictitious Cartesian configuration $\Omega_h$, since it is easier to use (from \eqref{Nye1})
\begin{align}
\sym \,\alpha_\xi\,=\,&-\sym\,\Gamma_\xi+ \tr(\Gamma_\xi)\, \id_3\,=\,-\sym(\Gamma_e\,[\nabla_x \Theta]^{-1})+ \tr(\Gamma_e\,[\nabla_x \Theta]^{-1})\, \id_3, \vspace{1.5mm}\notag\\
{\rm dev}\,\sym \,\alpha_\xi\,=\,&-{\rm dev}\,\sym\, \Gamma_\xi\,=\,-{\rm dev}\,\sym (\Gamma_e\,[\nabla_x \Theta]^{-1}),\vspace{1.5mm}\\
{\skw}
\,\alpha_\xi\,=\,&-\skw\,\Gamma_\xi\,=\,-\skw(\Gamma_e\,[\nabla_x \Theta]^{-1}),\vspace{1.5mm}\notag\\
\tr(\alpha_\xi)\,=\,&-\tr(\Gamma_\xi)+ 3\,\tr(\Gamma_\xi)\,=\,2\,\tr(\Gamma_\xi)\,=\,2\,\tr(\Gamma_e\,[\nabla_x \Theta]^{-1}),\notag
\end{align}
and to express the curvature energy in terms of $\Gamma_e\,[\nabla_x \Theta]^{-1}$ as
\begin{align}
W_{\rm{curv}}( \alpha_\xi)\,=\,&\mu\, {L}_{\rm c}^2\left( b_1\,\lVert \dev\,\text{sym} (\Gamma_e\,[\nabla_x \Theta]^{-1})\rVert^2+b_2\,\lVert \text{skew}(\Gamma_e\,[\nabla_x \Theta]^{-1})\rVert^2+4\,b_3\,
[{\rm tr}(\Gamma_e\,[\nabla_x \Theta]^{-1})]^2\right).
\end{align}
Note that using \
\begin{align}\label{curvfuraxl}
\overline{Q}_e^T\,\partial_{x_i} \overline{Q}_e\,=\,Q_0\,\overline{R}^T\,\partial_{x_i} (\overline{R}\,Q_0^T)\,=\,Q_0 (\overline{R}^T\,\partial_{x_i} \overline{R})\,Q_0^T-Q_0(Q_0^T\partial_{x_i} Q_0)\,Q_0^T,\quad i\,=\,1,2,3
\end{align}
and the identity
\begin{align}
 \mathrm{axl}(Q\, A\, Q^T)\,=\,Q\,\mathrm{axl}( A)\qquad  \forall\, Q\in {\rm SO}(3) \quad \text{and}\quad \forall\, A\in \mathfrak{so}(3),
\end{align}
we obtain  the following form of the {wryness tensor} 
\begin{align}\label{qiese}
\Gamma(x_1,x_2,x_3) :\,=\,&\,\Gamma_\xi(\Theta(x_1,x_2,x_3))\,=\,\Gamma_e\,[\nabla_x \Theta]^{-1}\notag\\\,=\,&\,Q_0\Big[ \Big(\mathrm{axl}(\overline{R}^T\,\partial_{x_1} \overline{R})\,|\, \mathrm{axl}(\overline{R}^T\,\partial_{x_2} \overline{R})\,|\,\mathrm{axl}(\overline{R}^T\,\partial_{x_3} \overline{R})\,\Big)\\&\ \ \quad - \Big(\mathrm{axl}(Q_0^T\,\partial_{x_1} Q_0)\,|\, \mathrm{axl}(Q_0^T\,\partial_{x_2} Q_0)\,|\,\mathrm{axl}(Q_0^T\,\partial_{x_3} Q_0)\,\Big)\Big] \,[\nabla_x \Theta]^{-1}.\notag
\end{align}

Applying the change of variables formula we obtain now a new form of the energy functional $I$ which suggests  to seek the unknown functions $\varphi$ and $\overline{R}$ as solutions of the following minimization problem
\begin{equation}\label{minprmod}
\begin{array}{crl}
{I}\,=\,\dd\int_{\Omega_h}\bigg[\widetilde{W}_{\rm{mp}}(F,\overline{R})+\widetilde{W}_{\rm{curv}}(\Gamma)\bigg]\, {\rm det}{(\nabla_x\Theta)} \,dV  - \widetilde{\Pi}(\varphi,\overline{R})\quad
{\to}
\quad\text{ min.\quad w.r.t.  }\quad (\varphi,\overline{R})\ ,
\end{array}
\end{equation}
where $dV$ denotes the volume element $dx_1dx_2dx_3$ and
\begin{align}
\widetilde{W}_{\rm{mp}}(F,\overline{R})&\,=\,W_{\rm{mp}}(\overline U _e)\notag\,=\,\mu\,\lVert \text{sym}(\overline U _e-\id_3)\rVert^2+\mu_{\rm c}\,\lVert \text{skew}(\overline U _e-\id_3)\rVert^2+
\dd\frac{\lambda}{2}\,[{\rm tr}(\text{sym}(\overline U _e-\id_3))]^2\\
&\,=\,\mu\,\lVert   {\rm dev}  \,\text{sym}(\overline U _e-\id_3)\rVert^2+\mu_{\rm c}\,\lVert \text{skew}(\overline U _e-\id_3)\rVert^2+
\dd\frac{\kappa}{2}\,[{\rm tr}(\text{sym}(\overline U _e-\id_3))]^2,\nonumber
 \\ 
\widetilde{W}_{\rm{curv}}(\Gamma)&\,=\,\mu\, {L}_{\rm c}^2 \left( b_1\,\lVert \dev \,\text{sym} \,\Gamma\rVert^2+b_2\,\lVert \text{skew} \,\Gamma\rVert^2+4\,b_3\,
[{\rm tr}(\Gamma)]^2\right).\notag
\end{align}
The external loading potential can be written as
\begin{equation}\label{loadpot3}
\widetilde{\Pi}(\varphi,\overline{R})
\,=\, \widetilde{\Pi}_f(\varphi) + \widetilde{\Pi}_t(\varphi) + \widetilde{\Pi}_{\Omega_h}(\overline{R}) + \widetilde{\Pi}_{\Gamma_t}(\overline{R}) 
\, ,
\end{equation}
with
\begin{align}\label{loadpot4}
\widetilde{\Pi}_f(\varphi):= & \;\Pi_f(\varphi_\xi) \,=\,\dd\int_{\Omega_\xi} \bigl\langle  f,  u\bigr\rangle \, dV(\xi)\, \,=\,\dd\int_{\Omega_h} \bigl\langle  \tilde{f}, \tilde{v} \bigr\rangle   \, dV\,,  \\
\widetilde{\Pi}_t(\varphi) :\,=\,&\;  \Pi_t(\varphi_\xi)= \dd\int_{\partial\Omega_t} \bigl\langle  t, u \bigr\rangle \, dS(\xi)\,= \int_{\Gamma_t} \bigl\langle \tilde t, \tilde{v} \bigr\rangle   \, dS\, , \qquad 
\widetilde{\Pi}_{\Omega_h}(\overline{R}) :\,=\; \Pi_{\Omega}(\overline{R}_\xi)\, , \qquad
\widetilde{\Pi}_{\Gamma_t}(\overline{R}):=\Pi_{\partial\Omega_t}(\overline{R}_\xi)\, ,  \notag
\end{align}
where $ \tilde{v} (x_i)= \varphi(x_i) - \Theta(x_i) $ is the displacement vector and the vector fields $ \tilde f $ and $ \tilde t $ can be determined in terms of $   f $ and $   t $, respectively, for instance
{$ \;(\tilde f(x))_i=(f(\Theta(x)))_i\det(\nabla_x \Theta).$}
 Here, $ \Gamma_t $ and $ \Gamma_d $ are nonempty subsets of the boundary of $ \Omega_h $ such that $   \Gamma_t \cup \Gamma_d= \partial\Omega_h $ and $ \Gamma_t \cap \Gamma_d= \emptyset $\,. On $ \Gamma_t $ we consider traction boundary conditions, while on $ \Gamma_d $ we have Dirichlet-type boundary conditions (i.e., $ \varphi $ and $ \overline{R} $ are prescribed on $ \Gamma_d $). We assume that  $ \Gamma_d  $ has the form $ \Gamma_d = \gamma_d\times (-\frac{h}{2} , \frac{h}{2})$, where the curve $\gamma_d $ is a subset of $ \partial \omega $ with length $(\gamma_d)>0 $. Accordingly, the boundary subset $ \Gamma_t $ has the form  $ \Gamma_t = \Big(\gamma_t\times (-\frac{h}{2} , \frac{h}{2})\Big) \cup \Big(\omega\times \Big\{\frac{h}{2}\Big\} \Big) \cup \Big(\omega\times \Big\{-\frac{h}{2}\Big\} \Big) $
and $ \Theta(\Gamma_t) = \partial\Omega_t\, $.

\subsection{Useful tensors defined through the diffeomorphism $\Theta$}

For our purpose, the diffeomorphism $\Theta:\mathbb{R}^3\rightarrow\mathbb{R}^3$ describing the reference configuration (i.e., the curved surface of the shell),  will be chosen in the specific form
\begin{equation}\label{defTheta}
\Theta(x_1,x_2,x_3)\,=\,y_0(x_1,x_2)+x_3\ n_0(x_1,x_2), \ \ \ \ \ \ \ \ \quad  n_0\,=\,\dd\frac{\partial_{x_1}y_0\times \partial_{x_2}y_0}{\lVert \partial_{x_1}y_0\times \partial_{x_2}y_0\rVert}\, ,
\end{equation}
where $y_0:\omega\to \mathbb{R}^3$ is a function of class $C^2(\omega)$. This specific form of the diffeomorphism 
$\Theta$ maps the midsurface $\omega$ of the fictitious Cartesian configuration  parameter space $\Omega_h$ onto the midsurface $\omega_\xi=y_0(\omega)$ of $\Omega_\xi$ and $n_0$ is the unit normal vector to $\omega_\xi$. For simplicity and  where no confusions may arise, further on we will omit  to write explicitly   the arguments $(x_1,x_2, x_3)$ of the diffeomorphism  $\Theta$ or we will specify only its dependence  on $x_3$.
 Remark that
\begin{align}
\nabla_x \Theta(x_3)&\,=\,(\nabla y_0|n_0)+x_3(\nabla n_0|0) \, \  \forall\, x_3\in \left(-\frac{h}{2},\frac{h}{2}\right),
\ \ 
\nabla_x \Theta(0)\,=\,(\nabla y_0|\,n_0),\ \ [\nabla_x \Theta(0)]^{-T}\, e_3\,=n_0,
\end{align}
and  $\det\nabla_x \Theta(0)=\det(\nabla y_0|n_0)=\sqrt{\det[ (\nabla y_0)^T\nabla y_0]}$ represents the surface element.

In the following we {identify} the  {\it Weingarten map (or shape operator)} {on $y_0(\omega)$} {with its associated matrix} ${\rm L}_{y_0}\in\mathbb{R}^{2\times 2}$ defined in Appendix \ref{sectGeo} by 
$
{\rm L}_{y_0}\,=\, {\rm I}_{y_0}^{-1} {\rm II}_{y_0},
$
where ${\rm I}_{y_0}$ and ${\rm II}_{y_0}$ are  the matrix representations of the {\it first fundamental form (metric)} and the  {\it  second fundamental form} {on $y_0(\omega)$}, respectively. 
Then, the {\it Gau{\ss} curvature} ${\rm K}$ of the surface {$y_0(\omega)$} is determined by
$
{\rm K} \,=\,{\rm det}{({\rm L}_{y_0})}\, 
$
and the {\it mean curvature} $\,{\rm H}\,$ through
$
2\,{\rm H}\, :={\rm tr}({{\rm L}_{y_0}}) \, .
$ We denote  the principal curvatures of the surface by  ${\kappa_1}$ and ${\kappa_2}$.

For our purpose we will write the expressions of $\nabla_x \Theta$, ${\rm det} (\nabla_x \Theta)$, $[\nabla_x \Theta]^{-1}$ 
corresponding to the special form of the map $\Theta$ given by \eqref{defTheta}, as well as some other of its  properties, see  Appendix \ref{proofLemmaMircea}.
We  have
$
[	\nabla_x \Theta(x_3)]\, e_3\,= \, n_0\,.
$
Let us recall that $X\in \mathrm{GL}^+(3)$ satisfies the \textit{\textbf{G}eneralized \textbf{K}irchhoff \textbf{C}onstraint} ({\rm GKC}) \cite{Neff_Habil04} if 
$
X\in {\rm GKC }:=\{X\in \mathrm{GL}^+(3)\ |\ X^T \, X\, e_3\,=\,\varrho^2 e_3,\ \varrho\in \mathbb{R}^+\}\, .
$
For all $X\in {\rm GKC }$ with the polar decomposition $X=R\, U _0$, if follows that  $ U _0\in  {\rm GKC }$.
In view of this property and $\nabla \Theta(x_3)=Q_0(x_3)U_0(x_3)$, it follows\footnote{In the rest of the paper $*$ denotes quantities having expressions which are not relevant for our calculations.}
$
U _0(x_3)\,=\,\begin{footnotesize}
\begin{pmatrix}
* &* &0  \\
* &* &0 \\
0 &0 &1
\end{pmatrix}
\end{footnotesize}\,, 
$
which implies  that 
\begin{align}
\label{Q0n}
d_3^0(x_3):\,=\,Q_0(x_3).\, e_3\,=\,Q_0(x_3)\,U_0(x_3).\, e_3\,=\,\nabla_x \Theta(
x_3).\, e_3\,=\,n_0.
\end{align}
This means that the initial director $ d_3^0$ is chosen along the normal to the reference midsurface (the ``\textit{material filament}'' of the shell, see Figure \ref{Fig2}), while $ d_\alpha^0:=Q_0.\, e_\alpha${, for $\alpha=1,2$,}
is an orthonormal basis in the tangent plane of $\omega_\xi\,$. 
In  the current configuration $\Omega_c$ the director $ d_3:=\overline{Q}_e.\, d_3^0$ is no longer orthogonal to the deformed surface $\omega_c:= \varphi_\xi(\omega_\xi)$ and the directors  $d_\alpha:= \overline{Q}_e.\, d_\alpha^0$, for $\alpha=1,2$, are not tangent to this surface. The deviation of the director $ d_3$ from the normal vector to $\omega_c$ describes the \textit{{transverse shear deformation}} of shells. Moreover, the rotations of $ d_1,\, d_2\,$ about the  director $ d_3$ describe the so-called \textit{{drilling rotations}} in shells (see \cite{Birsan-Neff-L54-2014,wisniewski2010finite}).

Let us  introduce the  tensors\footnote{These tensors are usually called the first fundamental form and  the second fundamental form, respectively. However, we will not use this terminology since it may lead to some confusions. The relation between these tensors are explained in Proposition \ref{propAB}.} defined by:
\begin{align}\label{AB}
{\rm A}_{y_0}&:=(\nabla y_0|0) \,[\nabla_x \Theta(0)]^{-1}\in\mathbb{R}^{3\times 3}, \qquad \qquad 
{\rm B}_{y_0}:=-(\nabla n_0|0) \,[\nabla_x \Theta(0)]^{-1}\in\mathbb{R}^{3\times 3},
\end{align}
and the so-called \textit{{alternator tensor}} ${\rm C}_{y_0}$ of the surface \cite{Zhilin06}
\begin{align}
{\rm C}_{y_0}:=\det(\nabla_x \Theta(0))\, [	\nabla_x \Theta(0)]^{-T}\,\begin{footnotesize}\begin{pmatrix}
0&1&0 \\
-1&0&0 \\
0&0&0
\end{pmatrix}\end{footnotesize}\,  [	\nabla_x \Theta(0)]^{-1}.
\end{align}

The introduced tensors have the  properties given by Proposition \ref{propAB} from Appendix \ref{AppendixpropAB}, which are essential in the derivation of the model entirely in matrix representation. 

\subsection{Plane stress conditions in the curved (reference) configuration}

The \textit{first Piola-Kirchhoff stress tensor} in the reference (curved) configuration $\Omega_\xi$ is given by
$
S_1(F_\xi,\overline{R}_\xi) =D_{F_\xi}\widetilde{W}_{\rm{mp}}(F_\xi,\overline{R}_\xi).
$ 
We also consider the \textit{Biot-type stress tensor}  from   classical elasticity theory
$
{T}_{\rm Biot}\big(\overline U _\xi\big):=D_{\overline U _\xi}{W}_{\rm mp}(\overline U _\xi)\, .
$
Since
$
D_{F
{_\xi}
}
{\overline U _\xi}.X\,=\,\overline{R}_\xi^T \, X 
$ \text{
	and} $
 \bigl\langle  D_{F
{_\xi}
}\widetilde{W}_{\rm{mp}}(F_\xi,\overline{R}_\xi),X\bigr\rangle \,=\, \bigl\langle  D_{\overline U _\xi}{W}_{\rm{mp}}({\overline U _\xi}),D_{F
{_\xi}
}{\overline U _\xi}.X\bigr\rangle ,
$ 
for all $X\in \mathbb{R}^{3\times3}$,
we note that
$
D_{F_\xi}\widetilde{W}_{\rm{mp}}(F_\xi,\overline{R}_\xi)\,=\,\overline{R}_\xi D_{\overline U _\xi}{W}_{\rm mp}(\overline U _\xi)\, .
$
Hence, we have
\begin{align}\label{pk0}
S_1(F_\xi,\overline{R}_\xi)& \,=\,\overline{R}_\xi\,{T}_{\rm Biot}( \overline U _\xi),\qquad\qquad {T}_{\rm Biot}( \overline U _\xi)\,=\,\overline{R}_\xi^T\,S_1(F_\xi,\overline{R}_\xi)\,.
\end{align}
The Biot-type stress  tensor ${T}_{\rm Biot}$ is given by
\begin{align}\label{pk2}
{T}_{\rm Biot}(\overline U _\xi)\,=\,&
2\,\mu \,\sym(\overline U _\xi-\id_3)+2\,\mu_{\rm c} \, \ {\rm skew} (\overline U _\xi-\id_3)
+\dd{\lambda}\,{\rm tr}(\sym(\overline U _\xi-\id_3))\id_3\, ,
\end{align}
while, using \eqref{pk0}, we obtain that the first Piola-Kirchhoff tensor $S_1$ has the following form
\begin{align}\label{pk}
S_1(F_\xi,\overline{R}_\xi)\,=\,\overline{R}_\xi\,\bigg[&
2\,\mu\, \sym(\overline{R}_\xi^T\, F_\xi-\id_3)+2\,\mu_{\rm c}  \ {\rm skew} (\overline{R}_\xi^T\, F_\xi- \id_3) +\dd{\lambda}\,{\rm tr}(\sym(\overline{R}_\xi^T\, F_\xi-\id_3))\id_3\bigg].
\end{align}

As usual in  the development of shell theories, we assume that the normal stress (Piola-Kirchhoff stress tensor in the normal direction $n_0$) on the transverse boundaries (\textit{upper and lower faces} $\omega_\xi^+$ and $\omega_\xi^-$, respectively, of the curved reference configuration $\Omega_\xi$) are vanishing, i.e.,
\begin{equation}\label{bc001}
S_1( F_\xi,\overline{R}_\xi)\Big|_{\omega_\xi^\pm}.\,(\pm n_0)\,=\,0 \qquad \text{``zero normal stresses'' on upper and lower faces}.
\end{equation}

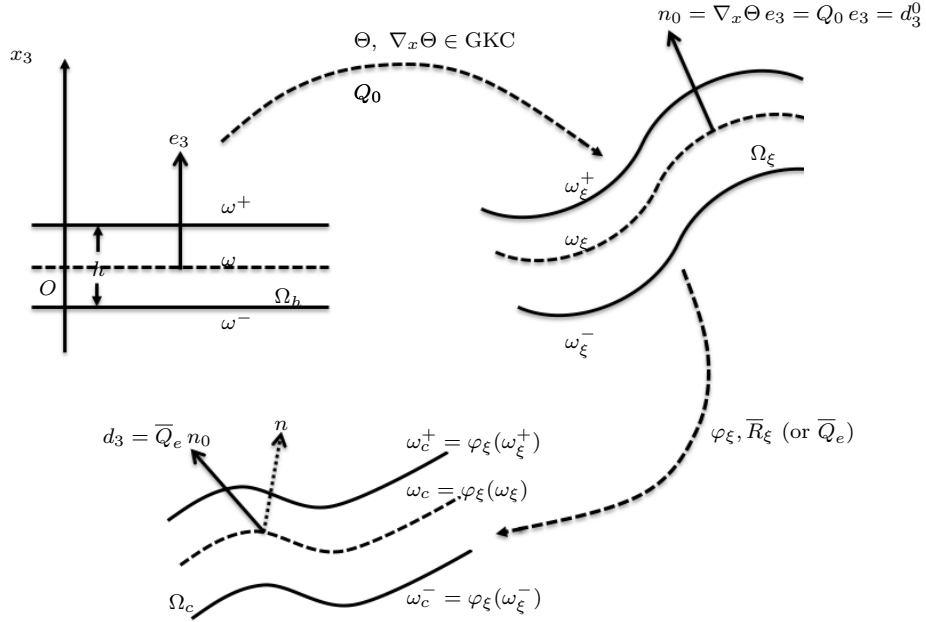
\begin{figure}
	\begin{center}	\includegraphics[scale=0.43]{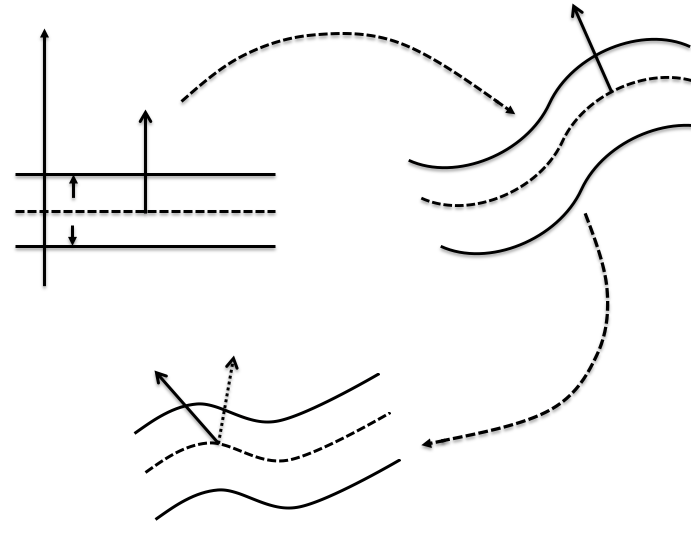} 
		\put(-240,188){\footnotesize $e_3$}
		\put(-200,128){\footnotesize $\Omega_h$}
		\put(-220,142){\footnotesize $\omega$}
		\put(-289,131){\footnotesize $O$}
		\put(-269,138){\footnotesize $h$}
		\put(-300,220){\footnotesize $x_3$}
		\put(-220,162){\footnotesize $\omega^+$}
		\put(-220,118){\footnotesize $\omega^-$}
		\put(-55,235){\footnotesize $n_0=\nabla_x\Theta\, e_3=Q_0\, e_3=d_3^0$}
		\put(-20,182){\footnotesize $\Omega_\xi$}
		\put(-90,150){\footnotesize $\omega_\xi$}
		\put(-90,170){\footnotesize $\omega_\xi^+$}
		\put(-90,111){\footnotesize $\omega_\xi^-$}
		\put(-170,225){\footnotesize $\Theta,\ \nabla_x\Theta\in {\rm GKC}$}
		\put(-170,205){\footnotesize $Q_0$}
		\put(-265,75){\footnotesize $d_3=\overline{Q}_e\, n_0$}
		\put(-200,80){\footnotesize $n$}
		\put(-240,12){\footnotesize $\Omega_c$}
		\put(-150,55 ){\footnotesize $\omega_c=\varphi_\xi(\omega_\xi)$}
		\put(-150,73){\footnotesize $\omega_c^+=\varphi_\xi(\omega_\xi^+)$}
		\put(-150,14){\footnotesize $\omega_c^-=\varphi_\xi(\omega_\xi^-)$}
		\put(-35,77){\footnotesize $\varphi_\xi, \overline{R}_\xi$  (or $\overline{Q}_e$)}
		\put(-170,205){\footnotesize $Q_0$}
		\caption{\footnotesize Transverse section in the shell. The shell is stress free at the upper and lower surface in the current configuration $\Omega_c$. With regard to the first Piola-Kirchhoff tensor $S_1$ this is equivalent to condition \eqref{bc001}. Transverse shear is automatically included in the model allowing the unit vector $d_3=\overline{Q}_e. n_0$ not to coincide with unit vector $n$, the normal to the deformed midsurface.}
		\label{Fig2}       
	\end{center}
\end{figure}

In the limit case $h\to 0$, these conditions imply $
S_1( F_\xi,\overline{R}_\xi)\Big|_{\omega_\xi}.\,(\pm n_0)\,=\,0$, i.e.,``zero normal stresses'' on the midsurface $\omega_\xi=\Theta(\omega\times \{0\})$, but the reverse of this implication is not valid. 

In fact, \eqref{bc001} is  equivalent to the assumption that the  Biot-stress tensor  in the normal direction $n_0$ is vanishing,  since
$
 {T}_{\rm Biot}( \overline U _\xi)\Big|_{\omega_\xi^\pm}.\,(\pm n_0)=\,\overline{R}_\xi^T\Big|_{\omega_\xi^\pm}\,S_1( F_\xi,\overline{R}_\xi))\Big|_{\omega_\xi^\pm}.(\pm n_0)\,=\,0,
$ and this implies, after scalar multiplication with $n_0$
\begin{equation}\label{bcn1nou0}
\bigl\langle  {T}_{\rm Biot}( \overline U _\xi)\Big|_{\omega_\xi^\pm}.\, n_0,n_0\bigr\rangle\,=\,0\, \qquad \text{``zero normal tractions'' on upper and lower faces} .
\end{equation}

\subsection{Neumann boundary conditions in the fictitious Cartesian configuration}

Using the coordinates of the fictitious Cartesian configuration, the plane stress conditions \eqref{bcn1nou0} are written in the form
\begin{equation}\label{bcn1nou}
\bigl\langle  {T}_{\rm Biot}\bigg( \overline U _e\big(x_1,x_2,\pm \frac{h}{2}\big)\bigg)\, n_0,n_0\bigr\rangle\,=\,0\, .
\end{equation}

A simplified approximated form of \eqref{bcn1nou} can be written in the limit case $h\to 0$ as in the following. Let us define the function
\begin{align}
f(x_3):=\bigl\langle  {T}_{\rm Biot}(\overline U _e(x_1,x_2,x_3))\, n_0,n_0\bigr\rangle,\,\qquad \forall \, x_3\in \left[-\frac{h}{2},\frac{h}{2}\right].
\end{align}
The Taylor expansion of $f(x_3)$ in $x_3=0$ leads to
\begin{align}\label{taylorS2}
f\left(\frac{h}{2}\right)+f\left(-\frac{h}{2}\right)\,=\,2\, f(0)+O(h^2),\qquad f\left(\frac{h}{2}\right)-f\left(-\frac{h}{2}\right)\,=\,h\, f'(0)+O(h^3),
	\end{align}
where
\begin{align}
f'(0)\,=\,\bigl\langle {\partial_{x_3}}{T}_{\rm Biot}(\overline U _e(x_1,x_2,x_3))\Big|_{x_3=0}\, n_0,n_0\bigr\rangle.
\end{align}
In view of the boundary conditions \eqref{bcn1nou} we have $f\left(\frac{h}{2}\right)=0=f\left(-\frac{h}{2}\right)$ and the relations \eqref{taylorS2} yield $f(0)=O(h^2)$ and $f'(0)=O(h^2)$. In the limit case $h\to 0$ one obtains the following approximated form of the conditions \eqref{bcn1nou}
\begin{align}\label{conditii}
\bigl\langle {T}_{\rm Biot}(\overline U _e(x_1,x_2,0))\, n_0, n_0\bigr\rangle&=\,0\  \ \ \text{``zero normal tractions'' on the midsurface} \ \omega_\xi=\Theta(\omega),\\
\bigl\langle \! {\partial_{x_3}}{T}_{\rm Biot}(\overline U _e(x_1,x_2,x_3))\Big|_{x_3=0}\, n_0,n_0\bigr\rangle\,&=\,0\ \ \  \text{``zero variations of normal tractions'' on the midsurface} \ \omega_\xi=\Theta(\omega).\notag
\end{align}
These relations represent a first approximation  in the dimensional reduction procedure and they will be used further on. In addition, in Appendix \ref{AppendixNeumann} we prove that, for our method, this first approximation leads to the same results (but in a simpler way) as is the case when the complete Neumann condition \eqref{bc001} are used, and then the limit $h\to 0$ is considered.

\section{The two-dimensional approximation}\setcounter{equation}{0}
\subsection{The 8-parameter ansatz for the two-dimensional approximation}

In the following, we want to find a \textit{reasonable approximation} of $(\varphi,\overline{R})$ involving only two-dimensional quantities.
Following the formal dimensional reduction procedure for the Cosserat elastic plates given in \cite{Neff_plate04_cmt}, we consider that the  rotation $\overline{R}:\Omega_h\rightarrow \text{\rm{SO}}(3)$
in the thin shell does  not depend on the thickness variable $x_3$
\begin{equation}\label{ansatzR}
\overline{R}(x_1,x_2,x_3)\,=\,\overline{R}_s(x_1,x_2)\, ,
\end{equation}
in line with the assumed thinness and material homogeneity of the structure. 
Moreover, an approximation of the elastic rotation $\overline{Q}_{e}:\Omega_h\rightarrow \text{\rm{SO}}(3)$ will be given by $\overline{Q}_{e,s}$
\begin{equation}
\overline{Q}_{e,s}(x_1,x_2)\,=\,\overline{R}_s(x_1,x_2)\,Q_0^T(x_1,x_2,0)\, .
\end{equation}

Taking into account\footnote{The definition  of the set {\rm GKC} and its properties are presented in Appendix \ref{AppendixpropAB}} that $\nabla_x \Theta\in {\rm GKC}$, with $\varrho\,=\,1$, we have
$
[\nabla_x \Theta]^{-1}\,[\nabla_x \Theta]^{-T}\, e_3\,=\,e_3\, .
$ 
In view of the properties of GKC, it follows
\begin{align}\label{eq:id_Qthetae3}
\overline{R}_{s}(x_1,x_2)\, e_3\,=\,&\overline{R}_{s}(x_1,x_2)\,U_0(x_1,x_2,0)\, e_3\,=\,\overline{Q}_{e,s}(x_1,x_2)\,Q_0(x_1,x_2,0)\,U_0(x_1,x_2,0)\, e_3\, .
\notag \\\,=\,&\overline{Q}_{e,s}(x_1,x_2)\,\nabla_x \Theta(x_1,x_2,0)\, e_3\,=\,\overline{Q}_{e,s}(x_1,x_2)\,(\nabla y_0|n_0)\, e_3\,=\,\overline{Q}_{e,s}(x_1,x_2)\, n_0\, .
\end{align}
Since $\overline{Q}_{e,s}\,{Q}_0\in \rm{SO}(3 )$ and with \eqref{Q0n} we have
\begin{align}\label{idp}
{Q}_0(x_1,x_2,x_3)\, e_3\,=\,n_0\,=\,{Q}_0(x_1,x_2,0)\, e_3\, .
\end{align}

In the engineering shell community it is well known \cite{Chernykh80,Schmidt85,Pietraszkiewicz85} that
the ansatz for the deformation over the thickness should be at least  quadratic%
\index{hierarchic plate models}
in order to avoid the so called
{\it  Poisson thickness locking} %
and to fully capture the three-dimensional kinematics without artificial
modification of the material laws\footnote{Let us quote from \cite{Schmidt85}: ``Due to
bending this change of length is generally {asymmetric} about (the midsurface) and leads to a shift of the original midsurfaces.... This asymmetry requires at least a {quadratic} representation of the (deformation in thickness direction).''}, see the detailed discussion of this point in \cite{Ramm00}
and compare with \cite{Ramm92,Ramm94,Ramm96,Ramm97,Sansour98c}.

We consider therefore the following \textit{8-parameter quadratic ansatz} in the thickness direction for the reconstructed total deformation $\varphi_s:\Omega_h\subset \mathbb{R}^3\rightarrow \mathbb{R}^3$ of the shell-like structure
\begin{align}
\varphi_s(x_1,x_2,x_3)\,=\,&m(x_1,x_2)+\bigg(x_3\varrho_m(x_1,x_2)+\dd\frac{x_3^2}{2}\varrho_b(x_1,x_2)\bigg)\overline{Q}_{e,s}(x_1,x_2)\nabla_x\Theta(x_1,x_2,0)\, e_3\, ,
\end{align}
where $m: \omega\subset\mathbb{R}^2\rightarrow\mathbb{R}^3$ takes on the role of the deformation of the midsurface of
the shell viewed as a parametrized surface,   the yet indeterminate functions $\varrho_m,\,\varrho_b:\omega\subset\mathbb{R}^2\rightarrow  \mathbb{R}$ allow in principal for symmetric thickness stretch  ($\varrho_m\neq1$) and asymmetric thickness stretch ($\varrho_b\neq 0$) about the midsurface.

To construct the  new   geometrically nonlinear   Cosserat shell model we start with an 8-parameter ansatz  (8 `dof': 3 components of the membrane deformation $m$, 3 degrees of freedom for $\overline{R}\in\rm{SO}(3)$,  2 degrees of freedom over the thickness $\varrho_m$ and $\varrho_b$) but in the end we will arrive at a \textit{6-parameter model}.  This will be possible because in the isotropic case the two scalar parameters $\varrho_m$ and $\varrho_b$ (the degrees of freedom over the thickness) can  analytically  be condensed out.

In view of \eqref{idp}, the above 8-parameter quadratic ansatz in the thickness direction can be written as
\begin{align}\label{ansatz}
\varphi_s(x_1,x_2,x_3)
\,=\,&m(x_1,x_2)+\bigg(x_3\varrho_m(x_1,x_2)+\dd\frac{x_3^2}{2}\varrho_b(x_1,x_2)\bigg)\overline{Q}_{e,s}(x_1,x_2)\, n_0\, .
\end{align}

With regard to the total deformation, this is then a kind of plate  formulation since the midsurface of the fictitious Cartesian reference configuration $\omega\subset\mathbb{R}^2$ is assumed to lie in the plane. This implies for
the total (reconstructed) deformation gradient of the shell
\begin{align}
\label{quadratic_shell_gradient}
F_s\,=\,\nabla_x\varphi_s(x_1,x_2,x_3)\,=\,&(\nabla  m|\, \varrho_m\,\overline{Q}_{e,s}(x_1,x_2)\nabla_x\Theta(x_1,x_2,0)\, e_3) \\&+x_3\, (\nabla \left[\varrho_m\,\overline{Q}_{e,s}(x_1,x_2)\nabla_x\Theta(x_1,x_2,0)\, e_3\right]|\varrho_b\,
\overline{Q}_{e,s}(x_1,x_2)\nabla_x\Theta(x_1,x_2,0)\, e_3)\nonumber \\
&+\frac{x_3^2}{2}(\nabla \left[\varrho_b\,\overline{Q}_{e,s}(x_1,x_2)\nabla_x\Theta(x_1,x_2,0)\, e_3\right]|0) \, \nonumber
\end{align}
{and especially}
\begin{align}\label{quadratic_shell_gradient_3}
 F_s\, e_3&=\, \varrho_m\,\overline{Q}_{e,s}(x_1,x_2)\nabla_x\Theta(x_1,x_2,0)\, e_3 + x_3\, \varrho_b\,
\overline{Q}_{e,s}(x_1,x_2)\nabla_x\Theta(x_1,x_2,0)\, e_3 \notag\\
&\overset{\mathclap{\eqref{eq:id_Qthetae3}}}{=} (\varrho_m+ x_3\, \varrho_b)\,\overline{R}_{s}(x_1,x_2)\, e_3.
\end{align}
If
$\overline{Q}_e\,=\,\id_3,\  \,\varrho_m\,=\,1,\,\,\varrho_b\,=\,0,\,\,m\,=\,y_0$ (as in the reference configuration $\Omega_\xi$), then $F_s\,=\,\nabla_x \Theta.$ 

In the rest of the paper we do not write explicitly the dependence of these functions on $x_1$ and $x_2$.

\subsection{From an 8-parameter ansatz to a 6-parameter model via the fictitious boundary conditions}
We  intend to find $\varrho_m$ and $\varrho_b$ such that the  boundary conditions \eqref{conditii} in the fictitious configuration
\begin{align}
\bigl\langle {T}_{\rm Biot}(\overline{U}_{e,s})\, n_0,n_0\bigr\rangle&\,=\,0,
\qquad
\bigl\langle  {\partial_{x_3}}{T}_{\rm Biot}(\overline{U}_{e,s})\, n_0,n_0\bigr\rangle\,=\,0 \qquad \text{for}\qquad x_3\,=\,0\notag
\end{align}
are satisfied. These boundary conditions are equivalent to
\begin{align}\label{eq5.9}
\bigl\langle {T}_{\rm Biot}(\overline{U}_{e,s})\,Q_0\, e_3,Q_0\, e_3\bigr\rangle&\,=\,0 ,
\qquad
\bigl\langle  {\partial_{x_3}}{T}_{\rm Biot}(\overline{U}_{e,s})\,Q_0\, e_3,Q_0\, e_3\bigr\rangle\,=\,0 \qquad \text{for}\qquad x_3\,=\,0,
\end{align}
and further  to
\begin{align}\label{bcnb}
\bigl\langle  Q_0^T\,{T}_{\rm Biot}(\overline{U}_{e,s})\,Q_0\, e_3, e_3\bigr\rangle&\,=\,0,
\qquad
\bigl\langle  Q_0^T\,{\partial_{x_3}}{T}_{\rm Biot}(\overline{U}_{e,s})\,Q_0\, e_3, e_3\bigr\rangle\,=\,0 \qquad \text{for}\qquad x_3\,=\,0,
\end{align}
where
\begin{align}\label{noi}
\overline {U}_{e,s}&\,=\,\overline{Q}_{e,s}^T\, F_s\,[\nabla_x\Theta]^{-1}\,=\,{Q}_0\,\overline{R}^T\, F_s\,[\nabla_x\Theta]^{-1}, \\
{T}_{\rm Biot}(\overline {U}_{e,s})&\,=\,
2\,\mu \,\sym(\overline {U}_{e,s}-\id_3)+2\,\mu_{\rm c} \,  {\rm skew} (\overline {U}_{e,s}-\id_3)
+\dd{\lambda}\,{\rm tr}(\sym(\overline {U}_{e,s}-\id_3))\id_3\, .\notag
\end{align}
To this aim, we calculate
\begin{align}
2\,Q_0^T\,\sym(\overline {U}_{e,s}-\id_3)\,{Q}_0\,=\,&\,{Q}_0^T\Big(\overline{Q}_{e,s}^T F_s[\nabla_x \Theta]^{-1}+[\nabla_x \Theta]^{-T} F_s^T \overline{Q}_{e,s}-2\id_3\Big){Q}_0 \\\,=\,&\,{Q}_0^T\overline{Q}_{e,s}^T F_s  U _0^{-1}{Q}_0^{T}{Q}_0+{Q}_0^T{Q}_0 U _0^{-T}F_s^T \overline{Q}_{e,s}{Q}_0-2\id_3
\notag \\\,=\,&\,{Q}_0^T\overline{Q}_{e,s}^T F_s  U _0^{-1}+ U _0^{-T}F_s^T \overline{Q}_{e,s}{Q}_0-2\id_3\, ,\notag \\
2\,Q_0^T\,{\rm skew}(\overline {U}_{e,s}-\id_3)\,{Q}_0\,=\,&\,{Q}_0^T\Big(\overline{Q}_{e,s}^T F_s[\nabla_x \Theta]^{-1}-[\nabla_x \Theta]^{-T} F_s^T \overline{Q}_{e,s}\Big){Q}_0 \\\,=\,&\,{Q}_0^T\overline{Q}_{e,s}^T F_s  U _0^{-1}{Q}_0^{T}{Q}_0-{Q}_0^T{Q}_0 U _0^{-T}F_s^T \overline{Q}_{e,s}{Q}_0
\,=\,{Q}_0^T\overline{Q}_{e,s}^T F_s  U _0^{-1}- U _0^{-T}F_s^T \overline{Q}_{e,s}{Q}_0\, ,\notag
\end{align}
and
\begin{align}
 \bigl\langle  Q_0^T\,\sym(\overline {U}_{e,s}-\id_3)&\,{Q}_0\, e_3, e_3\bigr\rangle \,=\, \bigl\langle 
({Q}_0^T\overline{Q}_{e,s}^T F_s  U _0^{-1}+ U _0^{-T}F_s^T \overline{Q}_{e,s}{Q}_0-2\id_3)\, e_3, e_3\bigr\rangle \notag \\&{=\,}\ 
2\, \bigl\langle 
( U _0^{-1}F_s^T \overline{Q}_{e,s}{Q}_0-\id_3)\, e_3, e_3\bigr\rangle =\,2\, \bigl\langle 
U _0^{-1}F_s^T \overline{Q}_{e,s}{Q}_0\, e_3, e_3\bigr\rangle -2,\notag \\
&{=\,2\,\bigl\langle 
F_s^T \overline{Q}_{e,s}{Q}_0\, e_3, U _0^{-T}\, e_3\bigr\rangle -2
=\,2\,\bigl\langle 
F_s^T \overline{Q}_{e,s}{Q}_0\, e_3, e_3\bigr\rangle -2
}\notag\\
&{\overset{\mathclap{\eqref{eq:id_Qthetae3}}}{=}\,2\,\bigl\langle 
F_s^T \overline{R}_{s}\, e_3, e_3\bigr\rangle -2 = 2\,\bigl\langle 
\overline{R}_{s}\, e_3, F_s\, e_3\bigr\rangle -2} \notag
\\
& {\overset{\mathclap{\eqref{quadratic_shell_gradient_3}}}{=} \,2\,\bigl\langle 
\overline{R}_{s}\, e_3,(\varrho_m+ x_3\, \varrho_b)\,\overline{R}_{s}\, e_3\bigr\rangle -2 = 2\, (\varrho_m+ x_3\, \varrho_b-1)
} 
\end{align}
{where we have used the special structure of $U_0=\begin{pmatrix} * & * & 0 \\ * & * & 0 \\0 & 0 & 1\end{pmatrix} $ and $\overline{R}_{s}\in{\rm SO}(3)$. Furthermore, we have}

\begin{align}
 \bigl\langle  Q_0^T\,{\rm skew}(\overline {U}_{e,s}-\id_3)&\,{Q}_0\, e_3, e_3\bigr\rangle \,=\, \bigl\langle (
{Q}_0^T\overline{Q}_{e,s}^T F_s  U _0^{-1}- U _0^{-T}F_s^T \overline{Q}_{e,s}{Q}_0)\, e_3, e_3\bigr\rangle \,=\,0, 
\end{align}
and
\begin{align}
 \bigl\langle  Q_0^T\,{\rm tr}(\sym(\overline {U}_{e,s}-\id_3))&\,\id_3 \,{Q}_0\, e_3, e_3\bigr\rangle \,=\,{\rm tr}(\sym(\overline {U}_{e,s}-\id_3)) \bigl\langle  Q_0^T\,{Q}_0\, e_3, e_3\bigr\rangle \,=\,{\rm tr}(\sym(\overline {U}_{e,s}-\id_3))\notag \\
\,=\,& \bigl\langle \sym(\overline {U}_{e,s}-\id_3),\id_3\bigr\rangle \,=\, \bigl\langle {Q}_0^T\overline{Q}_{e,s}^T F_s  U _0^{-1}+ U _0^{-T}F_s^T \overline{Q}_{e,s}{Q}_0-2\id_3,\id_3\bigr\rangle \notag \\
 \,=\, &2\, \bigl\langle  U _0^{-1}F_s^T \overline{Q}_{e,s}{Q}_0-\id_3,\id_3\bigr\rangle \,=\,
2\, \bigl\langle  F_s^T \overline{Q}_{e,s}{Q}_0, U _0^{-1}\bigr\rangle  -6\notag \\
\,=\,&2[ \bigl\langle (\nabla  m|0)^T\overline{Q}_{e,s} {Q}_0(x_3), U _0^{-1}\bigr\rangle +\varrho_m +x_3\varrho_m \bigl\langle  (\nabla (\,\overline{Q}_{e,s} {Q}_0(x_3)\, e_3)|0)^T\overline{Q}_{e,s} {Q}_0(x_3), U _0^{-1}\bigr\rangle \nonumber \\&\ +x_3\varrho_b +\frac{x_3^2}{2}\varrho_b \bigl\langle  (\nabla (\,\overline{Q}_{e,s} {Q}_0(x_3)\, e_3)|0)^T\overline{Q}_{e,s} {Q}_0(x_3), U _0^{-1}\bigr\rangle -3] \nonumber \\
\,=\,&2[ \bigl\langle (\nabla  m|0)^T\overline{Q}_{e,s} ,[\nabla_x\Theta(x_3)]^{-1}\bigr\rangle +\varrho_m +x_3\varrho_m \bigl\langle  (\nabla (\,\overline{Q}_{e,s} {Q}_0(x_3)\, e_3)|0)^T\overline{Q}_{e,s} ,[\nabla_x\Theta(x_3)]^{-1}\bigr\rangle \nonumber \\&\  +x_3\varrho_b +\frac{x_3^2}{2}\varrho_b \bigl\langle  (\nabla (\,\overline{Q}_{e,s} {Q}_0(x_3)\, e_3)|0)^T\overline{Q}_{e,s} ,[\nabla_x\Theta(x_3)]^{-1}\bigr\rangle -3]\, .
\end{align}

We deduce
\begin{align}
\bigl\langle  Q_0^T\,{T}_{\rm Biot}(\overline{U}_{e,s})\,Q_0\, e_3, e_3\bigr\rangle\,=\,&\,\mu[2(\varrho_m-1)+2x_3\varrho_b]+\lambda\Big[ \bigl\langle (\nabla  m|0)^T\overline{Q}_{e,s} ,[\nabla_x\Theta(x_3)]^{-1}\bigr\rangle +\varrho_m\nonumber \\
&\ \ +x_3\varrho_m \bigl\langle  (\nabla (\,\overline{Q}_{e,s} {Q}_0(x_3)\, e_3)|0)^T\overline{Q}_{e,s} ,[\nabla_x\Theta(x_3)]^{-1}\bigr\rangle \nonumber \\&\  \ +x_3\varrho_b +\frac{x_3^2}{2}\varrho_b \bigl\langle  \overline{Q}_{e,s} ^T(\nabla (\,\overline{Q}_{e,s} {Q}_0(x_3)\, e_3)|0),[\nabla_x\Theta(x_3)]^{-T}\bigr\rangle -3\Big].
\end{align}
The requirement (\ref{bcnb})${}_1$ leads to the ``plane stress" requirement for $x_3\,=\,0$\,(zero normal tractions on the upper and lower surface)
\begin{align}
2\,\mu\,(\varrho_m-1)+\lambda\,\Big[ \bigl\langle (\nabla  m|0)^T\overline{Q}_{e,s} ,[\nabla_x\Theta(0)]^{-1}\bigr\rangle +\varrho_m-3\Big]\,=\,0,
\end{align}
which, considering
$
\nabla_x \Theta\,=\,(\nabla y_0|n_0)+x_3(\nabla n_0|0)\, ,
$
is equivalent to
\begin{align}
2\,\mu\,&(\varrho_m-1)+\lambda\,[ \bigl\langle \overline{Q}_{e,s} ^T(\nabla  m|0),(\nabla y_0|n_0)^{-T}\bigr\rangle +\varrho_m-3]\,=\,0
\end{align}
and we obtain 
\begin{equation}\label{finalrhom}
\varrho_m\,=\,1-\frac{\lambda}{\lambda+2\mu}[ \bigl\langle  \overline{Q}_{e,s}^T(\nabla m|0)[\nabla_x\Theta(0)]^{-1},\id_3\bigr\rangle -2].
\end{equation}
Now, let us consider the boundary conditions (\ref{bcnb})${}_2$ and 
 observe 
\begin{align}\label{bcnb1}
\bigl\langle  Q_0^T\,\partial_{ x_3}{T}_{\rm Biot}(\overline{U}_{e,s})\,Q_0\, e_3, e_3\bigr\rangle&\,=\,\bigl\langle \partial_{ x_3}{Y}_{\rm Biot}(\overline{U}_{e,s})\,Q_0\, e_3, Q_0\, e_3\bigr\rangle\,=\,
\partial_{ x_3}\bigl\langle  Q_0^T{T}_{\rm Biot}(\overline{U}_{e,s})\,Q_0\, e_3, e_3\bigr\rangle,
\end{align}
since $Q_0\, e_3\,=\,n_0$ and therefore it is independent of $x_3$. We deduce
\begin{align}
\partial_{ x_3}\bigl\langle  Q_0^T{T}_{\rm Biot}(\overline{U}_{e,s})\,Q_0\, e_3, e_3\bigr\rangle\,=\,&
\,2\,\mu\,\varrho_b+\lambda[ \bigl\langle \overline{Q}_{e,s} ^T(\nabla  m|0),\partial_{ x_3}[\nabla_x\Theta(x_3)]^{-T}\bigr\rangle \nonumber \\
&\ +\varrho_m \bigl\langle  \overline{Q}_{e,s} ^T(\nabla (\,\overline{Q}_{e,s} \nabla_x\Theta(0)\, e_3)|0),[\nabla_x\Theta(x_3)]^{-T}\bigr\rangle \notag \\&\ +
x_3\varrho_m \bigl\langle  \overline{Q}_{e,s} ^T(\nabla (\,\overline{Q}_{e,s} \nabla_x\Theta(0)\, e_3)|0),\partial_{ x_3}[\nabla_x\Theta(x_3)]^{-T}\bigr\rangle  \\&\ +\varrho_b +{x_3}\varrho_b \bigl\langle  \overline{Q}_{e,s} ^T(\nabla (\,\overline{Q}_{e,s} \nabla_x\Theta(0)\, e_3)|0),[\nabla_x\Theta(x_3)]^{-T}\bigr\rangle -3]\notag \\
&\ +\frac{x_3^2}{2}\varrho_b \bigl\langle  \overline{Q}_{e,s} ^T(\nabla (\,\overline{Q}_{e,s} \nabla_x\Theta(0)\, e_3)|0),\partial_{ x_3}[\nabla_x\Theta(x_3)]^{-T}\bigr\rangle ]\notag
\end{align}
and the boundary condition (\ref{bcnb1})${}_2$ becomes
\begin{align}\hspace{-0.25cm}
2\,\mu\,\varrho_b+\lambda[ \bigl\langle \overline{Q}_{e,s} ^T(\nabla  m|0),\partial_{ x_3}[\nabla_x\Theta(x_3)]^{-T}\bigr\rangle \Big|_{x_3\,=\,0}\!\!+\varrho_m \bigl\langle  \overline{Q}_{e,s} ^T(\nabla (\,\overline{Q}_{e,s} \nabla_x\Theta(0)\, e_3)|0),[\nabla_x\Theta(0)]^{-T}\bigr\rangle +\varrho_b]\,=\,0.
\end{align}
Here we have used that $\partial_{x_3}[\nabla_x\Theta]^{-1}\Big|_{x_3\,=\,0}$ is finite, since $\det(\nabla_x\Theta(x_3))$ has a third order polynomial expression, see  Proposition \ref{propnablatheta},  and  $\det(\nabla y_0|n_0)\neq 0$.
We also remark that
\begin{align}\partial_{x_3}[\nabla_x\Theta(x_3)]^{-1}\Big|_{x_3\,=\,0}&\,=\,-[\nabla_x\Theta(0)]^{-1}\partial_{x_3}[\nabla_x\Theta(x_3)]\Big|_{x_3\,=\,0}\, [\nabla_x\Theta(0)]^{-1}\,=\,
-(\nabla y_0|n_0)^{-1}\,(\nabla n_0|0)\, (\nabla y_0|n_0)^{-1}.\notag
\end{align}
Therefore,  equation (\ref{bcnb1})${}_2$  is equivalent to
\begin{align}
2\,\mu\,\varrho_b+\lambda[&- \bigl\langle \overline{Q}_{e,s} ^T(\nabla  m|0)\,[\nabla_x\Theta(0)]^{-1}\,(\nabla n_0|0)\, [\nabla_x\Theta(0)]^{-1},\id_3\bigr\rangle  \\&\qquad \qquad \qquad +\varrho_m \bigl\langle  \overline{Q}_{e,s} ^T(\nabla (\,\overline{Q}_{e,s} \nabla_x\Theta(0)\, e_3)|0),[\nabla_x\Theta(0)]^{-T}\bigr\rangle +\varrho_b]\,=\,0,\notag
\end{align}
which yields
\begin{align}
\varrho_b\,=\,& \;\frac{\lambda}{\lambda+2\,\mu}\, \bigl\langle \overline{Q}_{e,s} ^T(\nabla  m|0)\,[\nabla_x\Theta(0)]^{-1}\,(\nabla n_0|0)\, [\nabla_x\Theta(0)]^{-1},\id_3\bigr\rangle \notag \\&-\varrho_m\,\frac{\lambda}{\lambda+2\,\mu}\, \bigl\langle  \overline{Q}_{e,s} ^T(\nabla (\,\overline{Q}_{e,s} \nabla_x\Theta(0)\, e_3)|0)\,[\nabla_x\Theta(0)]^{-1},\id_3\bigr\rangle  \notag\\
\,=\,& \;\frac{\lambda}{\lambda+2\,\mu}\, \bigl\langle \overline{Q}_{e,s} ^T(\nabla  m|0)\,[\nabla_x\Theta(0)]^{-1}\,(\nabla n_0|0)\, [\nabla_x\Theta(0)]^{-1},\id_3\bigr\rangle  \\&-\frac{\lambda}{\lambda+2\,\mu}\, \bigl\langle  \overline{Q}_{e,s} ^T(\nabla (\,\overline{Q}_{e,s} \nabla_x\Theta(0)\, e_3)|0)\,[\nabla_x\Theta(0)]^{-1},\id_3\bigr\rangle  \notag\\
&+\frac{\lambda^2}{(\lambda+2\,\mu)^2}\,\Big[ \bigl\langle  \overline{Q}_{e,s} ^T(\nabla (\,\overline{Q}_{e,s} \nabla_x\Theta(0)\, e_3)|0)\,[\nabla_x\Theta(0)]^{-1},\id_3\bigr\rangle \Big]\Big[ \bigl\langle  \overline{Q}_{e,s}^T(\nabla m|0)[\nabla_x\Theta(0)]^{-1},\id_3\bigr\rangle -2\Big].\notag
\end{align}

\begin{remark}
The term $\frac{\lambda^2}{(\lambda+2\,\mu)^2} \bigl\langle  \overline{Q}_{e,s}^T (\nabla (\,\overline{Q}_{e,s}\nabla_x\Theta(0)\, e_3)|0)[\nabla_x\Theta(0)]^{-1},\id_3\bigr\rangle [ \bigl\langle  \overline{Q}_{e,s}^T(\nabla m|0)[\nabla_x\Theta(0)]^{-1},\id_3\bigr\rangle -2]$ represents a nonlinear coupling between midsurface in-plane (membrane) strain and normal curvature, a result of the derivation not present in the underlying three-dimensional theory where only products of deformation gradient and rotations occur\footnote{In addition, this term is not invariant under reflection across the midsurface, i.e $\overline{Q}_e\,=\,(\overline{Q}_{e,1},\overline{Q}_{e,2},\overline{Q}_{e,3})\mapsto (\overline{Q}_{e,1},\overline{Q}_{e,2},-\overline{Q}_{e,3})$ \cite{Naghdi72}.}. Since we have in mind a small strain situation, this product is one order smaller than $\frac{\lambda}{\lambda+2\,\mu} \bigl\langle  (\nabla (\,\overline{Q}_{e,s}\nabla_x\Theta(0)\, e_3)|0)^T$ $\overline{Q}_{e,s}[\nabla_x\Theta(0)]^{-1},\id_3\bigr\rangle $. Therefore, we neglect this term.
The presence of the term $-\frac{\lambda}{\lambda+2\,\mu} \bigl\langle  \overline{Q}_{e,s}^T(\nabla m|0)(\nabla y_0|n_0)^{-1}(\nabla n_0|0)(\nabla y_0|n_0)^{-1},\id_3\bigr\rangle \, $ is not in contradiction  with the Cosserat-plate model \cite{Neff_plate04_cmt} because in the plate case it is automatically zero, since $\nabla n_0=\nabla e_3\equiv 0$.
\end{remark}

Thus, our considered form for $\varrho_m$ and $\varrho_b$ will be
\begin{align}\label{final_rho}
\varrho_m^e\,=\,&1-\frac{\lambda}{\lambda+2\,\mu}[ \bigl\langle  \overline{Q}_{e,s}^T(\nabla m|0)[\nabla_x\Theta(0)]^{-1},\id_3\bigr\rangle -2] \;,\notag \\
\dd\varrho_b^e\,=\,&-\frac{\lambda}{\lambda+2\,\mu} \bigl\langle  \overline{Q}_{e,s}^T(\nabla (\,\overline{Q}_{e,s}\nabla_x\Theta(0)\, e_3)|0)[\nabla_x\Theta(0)]^{-1},\id_3\bigr\rangle   \\
&+\frac{\lambda}{\lambda+2\,\mu} \bigl\langle  \overline{Q}_{e,s}^T(\nabla m|0)[\nabla_x\Theta(0)]^{-1}(\nabla n_0|0)[\nabla_x\Theta(0)]^{-1},\id_3\bigr\rangle \, .\notag
\end{align}
Remark that the condition (\ref{bcn1nou0}) is not satisfied exactly.
However, the formula (\ref{final_rho})$_1$ has a clear physical significance: {in-plane stretch leads to thickness reduction}.

Now, our final aim in the determination of $\varrho_m^e$ and $\varrho_b^e$ is  to compute them for $\overline{Q}_{e,s}\,=\,\id_3$ and $m\,=\,y_0$, which means the elastic deformation is absent, i.e., we compute
\begin{align}\label{finalrho00}
\varrho_m^0&\,=\,\dd 1-\frac{\lambda}{\lambda+2\mu}[ \bigl\langle  (\nabla y_0|0)(\nabla y_0|n_0)^{-1},\id_3\bigr\rangle -2],  \\
\dd\varrho_b^0&\,=\,-\dd\frac{\lambda}{\lambda+2\mu} \bigl\langle  (\nabla (\,\nabla_x\Theta(0)\, e_3)|0)[\nabla_x\Theta(0)]^{-1},\id_3\bigr\rangle + \dd\frac{\lambda}{\lambda+2\mu} \bigl\langle  (\nabla y_0|0)[\nabla_x\Theta(0)]^{-1}(\nabla n_0|0)[\nabla_x\Theta(0)]^{-1},\id_3\bigr\rangle \notag \\
&\,=\,-\dd\frac{\lambda}{\lambda+2\mu} \bigl\langle  (\nabla n_0|0)[\nabla_x\Theta(0)]^{-1},\id_3\bigr\rangle + \dd\frac{\lambda}{\lambda+2\mu} \bigl\langle  (\nabla y_0|0)[\nabla_x\Theta(0)]^{-1}(\nabla n_0|0)[\nabla_x\Theta(0)]^{-1},\id_3\bigr\rangle \,.
\end{align}

The identity
$
\tr({\rm A}_{y_0})\,=\, \bigl\langle  (\nabla y_0|0)(\nabla y_0|n_0)^{-1},\id_3\bigr\rangle \,=\,2,
$ (see Proposition \ref{propAB})
implies that
$
\varrho_m^0\,=\,1.
$
Next, we compute $\varrho_b^0$. 
With the help of the  curvature tensors $
{\rm A}_{y_0}, {\rm B}_{y_0}$ (see Proposition \ref{propAB})
we have
\begin{align}
{\rm tr}[(\nabla y_0|0)&(\nabla y_0|n_0)^{-1}(\nabla n_0|0)(\nabla y_0|n_0)^{-1}] \,=\,-2\,{\rm H}\, .
\end{align}
Hence, we deduce
\begin{equation}\label{finalrho0}
\begin{array}{crl}
\dd\varrho_b^0&\,=\,&\dd\frac{\lambda}{(\lambda+2\mu)}{\rm tr}[{\rm B}_{y_0}]- \dd\frac{\lambda}{(\lambda+2\mu)}{\rm tr}[{\rm L}_{y_0}]\,=\,2\frac{\lambda}{(\lambda+2\mu)}\,{\rm H}\,- 2\dd\frac{\lambda}{(\lambda+2\mu)}\,{\rm H} \,=\,0.
\end{array}
\end{equation}

Thus,  the reference values $\varrho_m^0$ and $\varrho_b^0$ of the parameters $\varrho_m^e$ and $\varrho_b^e$ are given by
$
\varrho_m^0\,=\,1,\  \varrho_b^0\,=\,0,
$
which means that in the absence of elastic deformation the ansatz \eqref{ansatz}  
$
\varphi_s^0(x_1,x_2,x_3)\,=\,y_0(x_1,x_2)+x_3\, n_0(x_1,x_2)\,=\, \Theta(x_1,x_2,x_3)
$
is exact.

\subsection{The ansatz for the deformation gradient}

Having obtained a suitable  form of the relevant coefficients $\varrho_m^e,\,\varrho_b^e$, it is expedient to base the
expansion of the three-dimensional elastic Cosserat energy on a further simplified expression, please compare with \eqref{quadratic_shell_gradient}, namely 
\begin{align}
F_s\,=&\,\nabla_x\varphi_s(x_1,x_2,x_3)\notag\\
\,=\,&(\nabla  m|\, \varrho_m\,\overline{Q}_{e,s}(x_1,x_2)\nabla_x\Theta(x_1,x_2,0)\, e_3) \notag\\&+x_3\, (\nabla \big[\underbrace{\xcancel{\varrho_m}}_{\cong\varrho_m^0= 1}\,\overline{Q}_{e,s}(x_1,x_2)\nabla_x\Theta(x_1,x_2,0)\, e_3\big]|\varrho_b\,
\overline{Q}_{e,s}(x_1,x_2)\nabla_x\Theta(x_1,x_2,0)\, e_3)\nonumber \\
&+\frac{x_3^2}{2}(\nabla \big[\underbrace{\xcancel{\varrho_b}}_{\cong\varrho_b^0= 0}\,\overline{Q}_{e,s}(x_1,x_2)\nabla_x\Theta(x_1,x_2,0)\, e_3\big]|0) \, \nonumber.
\\
\label{reduced_expansion}
\cong&\,(\nabla  m|\, \varrho_m^e\,\overline{Q}_{e,s}(x_1,x_2)\nabla_x\Theta(x_1,x_2,0)\, e_3)
 \\\nonumber  &\qquad +x_3 (\nabla \left[\,\overline{Q}_{e,s}(x_1,x_2)\nabla_x\Theta(x_1,x_2,0)\, e_3\right]|\varrho_b^e\,
\overline{Q}_{e,s}(x_1,x_2)\nabla_x\Theta(x_1,x_2,0)\, e_3).
\end{align}
\begin{remark}{\rm (Raison d'\^{e}tre)}
\begin{itemize}
\item[1)] The reduced model should at no place contain space derivatives of the thickness stretch
$\varrho_m^e$, since in the underlying three-dimensional Cosserat model curvature is only present through
the dislocation density tensor $\alpha_\xi$ (or through
the wryness tensor $\Gamma_\xi$) related only to rotations $\overline{Q}_e$.

\item[2)]If we blithely use the fully reconstructed deformation gradient $F_s$ and integrate analytically through the thickness, we would obtain second order derivatives in the energy (through derivatives on $\varrho_m^e$ and $\varrho_b^e$) both for the midsurface $m$ and microrotation $\overline{Q}_e$, leading to a coupled fourth order problem, a situation which has to be avoided for simplicity and efficiency in a subsequent numerical implementation, taking also in consideration the second order Cosserat bulk problem.

\item[3)] Keeping the quadratic ansatz \eqref{quadratic_shell_gradient} but neglecting only the derivatives of $\varrho_m^e$ and $\varrho_b^e$, i.e., basing the integration through the thickness instead on \eqref{quadratic_shell_gradient}, the reduced ansatz \eqref{reduced_expansion}
would already lead to a second order equilibrium problem and entitle us to skip the quadratic term altogether, since either $h^5$-bending terms appear or $h^3$- product of membrane and bending appear, which can be dominated through Youngs-inequality by a sum of $h^2$-membrane and $h^4$-bending terms, which themselves are subordinate (for small $h$) to the already appearing $h$-membrane and $h^3$-bending terms.

\item[4)] The error induced by the modified ansatz \eqref{reduced_expansion}  in the energy density will be of higher order under the assumption of small elastic midsurface strain.

\item[5)] Finally, it should be observed that by using \eqref{reduced_expansion} we are consistent with  John's general result \cite{John61,John65} that the stress distribution through the thickness is approximately linear for a thin shell.

\end{itemize}
\end{remark}
Motivated by the above remarks on the ansatz for the (reconstructed) deformation gradient  \eqref{reduced_expansion}, the chain rule leads to the approximation 
\begin{align}\label{red2}
&F_{s,\xi}\,=\,\nabla_x\varphi_s(x_1,x_2,x_3)[\nabla_x \Theta(x_1,x_2,x_3)]^{-1}  \\\nonumber
&\cong \widetilde{F}_{e,s}:\,=\,(\nabla  m|\, \varrho_m^e\,\overline{Q}_{e,s}(x_1,x_2)\nabla_x\Theta(x_1,x_2,0)\, e_3)[\nabla_x \Theta(x_1,x_2,x_3)]^{-1}
 \\\nonumber  &\qquad \qquad\qquad +x_3 (\nabla \left[\,\overline{Q}_{e,s}(x_1,x_2)\nabla_x\Theta(x_1,x_2,0)\, e_3\right]|\varrho_b^e\,
\overline{Q}_{e,s}(x_1,x_2)\nabla_x\Theta(x_1,x_2,0)\, e_3)[\nabla_x \Theta(x_1,x_2,x_3)]^{-1}
.
\end{align}

Our model will be constructed under the following assumption  upon the thickness 
\begin{align}
 {h}\,|{\kappa_1}|<\frac{1}{2},\qquad  {h}\,|{\kappa_2}|<\frac{1}{2},
\end{align}
where  $\kappa_1,\kappa_2$ are the  principal curvatures.

In consequence,  using  (iii) from Proposition \ref{propnablatheta}, we find that the (reconstructed) deformation gradient  is given by
\begin{align}\label{e68}
\widetilde{F}_{e,s}\,=\,&\,\frac{1}{b(x_3)}\left[(\nabla  m| \,\overline{Q}_{e,s}\nabla_x\Theta(0)\, e_3)+x_3 (\nabla \left[\,\overline{Q}_{e,s}\nabla_x\Theta(0)\, e_3\right]|0)
+
(\varrho_m^e-1+x_3\varrho_b^e)(0|0| \,\overline{Q}_{e,s}\nabla_x\Theta(0)\, e_3)\right]\notag
 \\
&\qquad \quad \times \left[\id_3+x_3({\rm L}_{y_0}^\flat-2\,{\rm H}\, \id_3)+x_3^2\, {\rm K}\begin{footnotesize}\begin{pmatrix}
0&0&0 \\
0&0&0 \\
0&0&1
\end{pmatrix}\end{footnotesize}
\right] [	\nabla_x \Theta(0)]^{-1}{,}
\end{align}
{where we have set $b(x_3)\coloneqq 1-2\,{\rm H}\,x_3 +{\rm K}\,x_3^2$.}
Next, we need to express the  tensors
\begin{align}\label{defEes}
\widetilde{\mathcal{E}}_{s}&:\,=\,\overline{U}_{e,s}-\id
{_3}
\,=\,\overline{Q}_{e,s}^T\, \widetilde{F}_{e,s}-\id_3,\qquad
\Gamma_{s}  :\,=\,  (\mathrm{axl}(\overline{Q}_{e,s}^T\,\partial_{x_1} \overline{Q}_e)\,|\, \mathrm{axl}(\overline{Q}_{e,s}^T\,\partial_{x_2} \overline{Q}_e)\,|0)[\nabla_x \Theta(x_3)]^{-1}
\end{align}
with the help of  the usual strain measures  in the nonlinear 6-parameter shell theory \cite{Eremeyev06}, see Section \ref{comparisonP}.  Therefore, we introduce the following tensor fields on the surface $\omega_\xi\,$ \cite{Libai98,Pietraszkiewicz-book04,Eremeyev06,Birsan-Neff-MMS-2014,Birsan-Neff-L54-2014}
\begin{align}\label{e55}
\mathcal{E}_{m,s} & :\,=\,    \overline{Q}_{e,s}^T  (\nabla  m|\overline{Q}_{e,s}\nabla_x\Theta(0)\, e_3)[\nabla_x \Theta(0)]^{-1}-\id_3\qquad \qquad \ \ \text{(the {elastic shell strain tensor})} ,  \\
\mathcal{K}_{e,s} & :\,=\,  (\mathrm{axl}(\overline{Q}_{e,s}^T\,\partial_{x_1} \overline{Q}_e)\,|\, \mathrm{axl}(\overline{Q}_{e,s}^T\,\partial_{x_2} \overline{Q}_e)\,|0)[\nabla_x \Theta(0)]^{-1} \, \qquad \text{(elastic shell bending--curvature tensor)}.\notag
\end{align}

\begin{lemma}\label{LemaEsEe} The following identities are satisfied 
\begin{itemize}
\item[i)] $\mathcal{E}_{m,s} \,=\,     ( \overline{Q}_{e,s}^T\,\nabla m -\nabla y_0 |0)[\nabla_x \Theta(0)]^{-1}=Q_0( \overline{R}_s^T\,\nabla m -Q_0^T\nabla y_0 |0)[\nabla_x \Theta(0)]^{-1}$;
\qquad ii) $\mathcal{E}_{m,s} {\rm A}_{y_0} \,= \,\mathcal{E}_{m,s} $; 
\item[iii)] {$\mathcal{K}_{e,s} {\rm A}_{y_0} \,= \,\mathcal{K}_{e,s} $; }
\qquad iv) $\overline{Q}_{e,s}^T\,(\nabla [\overline{Q}_{e,s}\nabla_x\Theta (0)\, e_3]\,|\,0)\,[\nabla_x \Theta(0)]^{-1} \,= \,{\rm C}_{y_0} \mathcal{K}_{e,s}-{\rm B}_{y_0}$.
\end{itemize}
\end{lemma}
\begin{proof}
{Starting with i), we observe}
\begin{align}
\mathcal{E}_{m,s} =    (\overline{Q}_{e,s}^T \nabla  m|n_0)[\nabla_x \Theta(0)]^{-1}-\id_3= (\overline{Q}_{e,s}^T \nabla  m|n_0)[\nabla_x \Theta(0)]^{-1}-(\nabla y_0|n_0)\,[\nabla_x \Theta(0)]^{-1}.
\end{align}
{Hence, we obtain  ii) and iii)  with Proposition  \ref{propAB} (v).}

{The last item follows from the same procedure we used to establish \eqref{qiese}. We have}
\begin{align}\label{keRQ}
\mathcal{K}_{e,s} &\  \,=\,Q_0\Big[ \Big(\mathrm{axl}( \overline{R}_s^T\,\partial_{x_1}  \overline{R}_s)\,|\, \mathrm{axl}( \overline{R}_s^T\,\partial_{x_2}  \overline{R}_s)\,|\,0\Big)  \notag\\&\qquad\qquad - \Big(\mathrm{axl}(Q_0^T(0)\,\partial_{x_1} Q_0(0))\,|\, \mathrm{axl}(Q_0^T(0)\,\partial_{x_2} Q_0(0))\,|\,0\,\Big)\Big] \,[\nabla_x \Theta(0)]^{-1}.
\end{align}
We compute
\begin{align}
\overline{R}_s^T\,\partial_{x_\alpha} \overline{R}_s&\,=\,(d_1\,|\, d_2\,|\,d_3)^T(\partial_{x_\alpha}d_1\,|\, \partial_{x_\alpha}d_2\,|\,\partial_{x_\alpha}d_3)\,=\,\begin{footnotesize}\begin{pmatrix}
0& \bigl\langle  d_1, \partial_{x_\alpha}d_{2}\bigr\rangle  & \bigl\langle  d_1, \partial_{x_\alpha}d_{3}\bigr\rangle   \\
 \bigl\langle  d_2, \partial_{x_\alpha}d_{1}\bigr\rangle &0& \bigl\langle  d_2, \partial_{x_\alpha}d_{3}\bigr\rangle  \\
 \bigl\langle  d_3, \partial_{x_\alpha}d_{1}\bigr\rangle & \bigl\langle  d_3, \partial_{x_\alpha}d_{2}\bigr\rangle & 0
\end{pmatrix}\end{footnotesize},
 \\
{\rm axl}(\overline{R}_s^T\,\partial_{x_\alpha} \overline{R}_s)&\,=\,\Big(- \bigl\langle  d_2, \partial_{x_\alpha} d_{3}\bigr\rangle  \,| \, \bigl\langle  d_1, \partial_{x_\alpha}d_{3}\bigr\rangle  \,| \,- \bigl\langle  d_1, \partial_{x_\alpha}d_{2}\bigr\rangle \Big)^T,\qquad \qquad \alpha\,=\,1,2.\notag
\end{align}
Thus, we deduce
\begin{align}
{\rm C}_{y_0}\,Q_0(0)\,& \Big(\mathrm{axl}(\overline{R}_s^T(\,\partial_{x_1} \overline{R}_s)\,|\, \mathrm{axl}(\overline{R}_s^T\,\partial_{x_2} \overline{R}_s)\,|\,0\,\Big) \,[\nabla_x \Theta(0)]^{-1}\notag \\
&\,=\,Q_0(0)\, \begin{footnotesize}\begin{pmatrix}
0&1&0 \\
-1&0&0 \\
0&0&0
\end{pmatrix}\end{footnotesize}\,
\begin{footnotesize}\begin{pmatrix}
- \bigl\langle  d_2, \partial_{x_1} d_{3}\bigr\rangle  &- \bigl\langle  d_2, \partial_{x_2} d_{3}\bigr\rangle  &0 \\
 \bigl\langle  d_1, \partial_{x_1}d_{3}\bigr\rangle & \bigl\langle  d_1, \partial_{x_2}d_{3}\bigr\rangle &0 \\
- \bigl\langle  d_1, \partial_{x_1}d_{2}\bigr\rangle &- \bigl\langle  d_1, \partial_{x_2}d_{2}\bigr\rangle &0
\end{pmatrix}\end{footnotesize}\,[\nabla_x \Theta(0)]^{-1}
 \\
&\,=\,Q_0(0)\,\overline{R}_s^T\, (\partial_{x_1}d_{3}\,|\, \partial_{x_2}d_{3}\,|\,0)\,[\nabla_x \Theta(0)]^{-1}\,=\,\overline{Q}_{e,s}^T\,(\nabla [\overline{Q}_{e,s}\nabla_x\Theta (0)\, e_3]\,|\,0)\,[\nabla_x \Theta(0)]^{-1}.\notag
\end{align}
Using Proposition \ref{propAB} and \eqref{keRQ} we {conclude} 
\begin{align*}
{\rm C}_{y_0}\mathcal{K}_{e,s} &=\overline{Q}_{e,s}^T\,(\nabla [\overline{Q}_{e,s}\nabla_x\Theta (0)\, e_3]\,|0)\,[\nabla_x \Theta(0)]^{-1}
{+}
{\rm B}_{y_0}\\&=[\overline{Q}_{e,s}^T\,(\nabla [\overline{Q}_{e,s}\nabla_x\Theta (0)\, e_3]\,|0)-\nabla_x \Theta(0) {\rm L}_{y_0}^\flat]\, [\nabla_x \Theta(0)]^{-1}.\qedhere
\end{align*}
\end{proof}

Accordingly, for the strain tensor corresponding to $\, \widetilde{\mathcal{E}}_{s}\,$ and using Lemma \ref{LemaEsEe} we find the following expression of the tensor $\widetilde{\mathcal{E}}_{s}$ defined by \eqref{defEes} 
\begin{align}
\widetilde{\mathcal{E}}_{s}  :\,=\,  
\dfrac{1}{b(x_3)}\,\Big\{ &  \mathcal{E}_{m,s} +(\varrho_m^e-1)\overline{Q}_{e,s}^T(0|0| \,\overline{Q}_{e,s}\nabla_x\Theta(0)\, e_3)[	\nabla_x \Theta(0)]^{-1}\notag \\&+x_3\Big[2\,{\rm H}\, \id_3+\overline{Q}_{e,s}^T(\nabla  m| \,\overline{Q}_{e,s}\nabla_x\Theta(0)\, e_3)({\rm L}_{y_0}^\flat-2\,{\rm H}\, \id_3)[	\nabla_x \Theta(0)]^{-1}\notag \\&\qquad \ +\,{\rm C}_{y_0} \, \mathcal{K}_{e,s} -{\rm B}_{y_0} +(\varrho_m^e-1)\overline{Q}_{e,s}^T(0|0| \,\overline{Q}_{e,s}\nabla_x\Theta(0)\, e_3)\, ({\rm L}_{y_0}^\flat-2\,{\rm H}\, \id_3)[	\nabla_x \Theta(0)]^{-1}\notag \\
&\qquad\ +\,\varrho_b^e\,\overline{Q}_{e,s}^T(0|0| \,\overline{Q}_{e,s}\nabla_x\Theta(0)\, e_3)[	\nabla_x \Theta(0)]^{-1} \Big] \\&+x_3^2\Big[\overline{Q}_{e,s}^T(\nabla  m| \,\overline{Q}_{e,s}\nabla_x\Theta(0)\, e_3) {\rm K}\begin{footnotesize}\begin{pmatrix}
0&0&0 \\
0&0&0 \\
0&0&1
\end{pmatrix}\end{footnotesize}[	\nabla_x \Theta(0)]^{-1}\notag \\&\qquad \ -{\rm K}\,\id_3+\,\overline{Q}_{e,s}^T (\nabla \left[\,\overline{Q}_{e,s}\nabla_x\Theta(0)\, e_3\right]|0)({\rm L}_{y_0}^\flat-2\,{\rm H}\, \id_3)[	\nabla_x \Theta(0)]^{-1}\notag \\
&\qquad \ +(\varrho_m^e-1)\overline{Q}_{e,s}^T(0|0| \,\overline{Q}_{e,s}\nabla_x\Theta(0)\, e_3)\, {\rm K}\begin{footnotesize}\begin{pmatrix}
0&0&0 \\
0&0&0 \\
0&0&1
\end{pmatrix}\end{footnotesize}[	\nabla_x \Theta(0)]^{-1}\notag \\&\qquad \ +\,\varrho_b^e\,\overline{Q}_{e,s}^T(0|0| \,\overline{Q}_{e,s}\nabla_x\Theta(0)\, e_3)\, ({\rm L}_{y_0}^\flat-2\,{\rm H}\, \id_3)[	\nabla_x \Theta(0)]^{-1}\Big]\notag \\
&+x_3^3\Big[\,\overline{Q}_{e,s}^T (\nabla \left[\,\overline{Q}_{e,s}\nabla_x\Theta(0)\, e_3\right]|0){\rm K}\begin{footnotesize}\begin{pmatrix}
0&0&0 \\
0&0&0 \\
0&0&1
\end{pmatrix}\end{footnotesize}[	\nabla_x \Theta(0)]^{-1}\notag
 \\
&\qquad \ \ + \varrho_b^e\,\overline{Q}_{e,s}^T(0|0| \,\overline{Q}_{e,s}\nabla_x\Theta(0)\, e_3)\, {\rm K}\begin{footnotesize}\begin{pmatrix}
0&0&0 \\
0&0&0 \\
0&0&1
\end{pmatrix}\end{footnotesize}[	\nabla_x \Theta(0)]^{-1}\Big]\Big\}.\notag
\end{align}
From Proposition \ref{propAB} we have
\begin{align}
\overline{Q}_{e,s}^T(\nabla  m| \,\overline{Q}_{e,s}\nabla_x\Theta(0)\, e_3) \begin{footnotesize}\begin{pmatrix}
0&0&0 \\
0&0&0 \\
0&0&1
\end{pmatrix}\end{footnotesize}[	\nabla_x \Theta(0)]^{-1}\,=\,(0|0| \,
{\nabla_x\Theta(0)}
\, e_3) \,[	\nabla_x \Theta(0)]^{-1}\,=\,\id_3-{\rm A}_{y_0},
\end{align}
which, together with Lemma \ref{LemaEsEe} and
$
{\rm B}_{y_0}\,=\,[\nabla_x \Theta(0)]\,{\rm L}_{y_0}^\flat\,[	\nabla_x \Theta(0)]^{-1}
$
leads to 
\begin{align}
\widetilde{\mathcal{E}}_{s}  \,=\, 
\dfrac{1}{b(x_3)}\,\Big\{ &  \mathcal{E}_{m,s} +(\varrho_m^e-1)(0|0| \,\nabla_x\Theta(0)\, e_3)[	\nabla_x \Theta(0)]^{-1}\notag \\&+x_3\Big[2\,{\rm H}\, \id_3+(\mathcal{E}_{m,s} +\id_3)\, 	\,({\rm B}_{y_0}-2\,{\rm H}\, \id_3)\notag \\&\qquad \quad +\,{\rm C}_{y_0} \, \mathcal{K}_{e,s} -{\rm B}_{y_0} +[2\,{\rm H}\,(1-\varrho_m^e)+\varrho_b^e](0|0| \,\nabla_x\Theta(0)\, e_3)[	\nabla_x \Theta(0)]^{-1} \Big] \\&+x_3^2\Big[{\rm K}\,\id_3-{\rm K}\,{\rm A}_{y_0}-{\rm K}\,\id_3+\,[{\rm C}_{y_0} \, \mathcal{K}_{e,s} -{\rm B}_{y_0}]({\rm B}_{y_0}-2\,{\rm H}\, \id_3)\notag \\
&\qquad \quad +(\varrho_m^e-1)(0|0|\,\nabla_x\Theta(0)\, e_3)\, {\rm K}\begin{footnotesize}\begin{pmatrix}
0&0&0 \\
0&0&0 \\
0&0&1
\end{pmatrix}\end{footnotesize}[	\nabla_x \Theta(0)]^{-1}\notag \\&\qquad \quad +\,\varrho_b^e\,(0|0|\,\nabla_x\Theta(0)\, e_3)\, ({\rm L}_{y_0}^\flat-2\,{\rm H}\, \id_3)[	\nabla_x \Theta(0)]^{-1}\Big]\notag \\
&+x_3^3\Big[0_3+ \varrho_b^e\,(0|0| \,\nabla_x\Theta(0)\, e_3)\, {\rm K}\begin{footnotesize}\begin{pmatrix}
0&0&0 \\
0&0&0 \\
0&0&1
\end{pmatrix}\end{footnotesize}[	\nabla_x \Theta(0)]^{-1}\Big]\Big\}.\notag
\end{align}
Further, using the Cayley-Hamilton type equation and item i) from Proposition \ref{propAB},  item ii) of Lemma  \ref{LemaEsEe} and 
$
\mathcal{E}_{m,s} \, (0|0| \,\overline{Q}_{e,s}\nabla_x\Theta(0)\, e_3)\,=\,0_3
$ we deduce
\begin{align}\label{e69Dumi}
\widetilde{\mathcal{E}}_{s}   \,=\,
\dfrac{1}{b(x_3)}\,\Big\{ &  \mathcal{E}_{m,s} +(\varrho_m^e-1)(0|0| \,\nabla_x\Theta(0)\, e_3)[	\nabla_x \Theta(0)]^{-1}\notag \\&+x_3\Big[\mathcal{E}_{m,s} ( {\rm B}_{y_0}-2\,{\rm H}\, {\rm A}_{y_0})\ +\,{\rm C}_{y_0} \, \mathcal{K}_{e,s} +A_1\,(0|0| \,\nabla_x\Theta(0)\, e_3)[	\nabla_x \Theta(0)]^{-1} \Big] \\&+x_3^2\Big[\,{\rm C}_{y_0} \, \mathcal{K}_{e,s} ({\rm B}_{y_0}-2\,{\rm H}\,{\rm A}_{y_0}) +A_2\,(0|0| \,\nabla_x\Theta(0)\, e_3)[	\nabla_x \Theta(0)]^{-1}\Big]\notag \\
&+x_3^3\, {\rm K}\,\varrho_b^e(0|0| \,\nabla_x\Theta(0)\, e_3)[	\nabla_x \Theta(0)]^{-1}\Big\},\notag
\end{align}
where
\begin{equation}\label{e70Dumi}
\begin{array}{l}
A_1:\,=\, 2\,{\rm H}\,(1-\varrho_m^e)+ \varrho_b^e ,
\qquad 
A_2:\,=\,  {\rm K}\,(\varrho_m^e-1)- 2\,{\rm H}\,\varrho_b^e.
\end{array}
\end{equation}

Using \eqref{final_rho}, we are able to express $\varrho_m^e$ and $\varrho_b^e$  also in terms of $\mathcal{E}_{m,s} , {\rm B}_{y_0}, {\rm A}_{y_0}, {\rm C}_{y_0}$ and $\mathcal{K}_{e,s} $ 
\begin{align}\label{final_rho-Dumi}
\varrho_m^e
\,=\,&1-\frac{\lambda}{\lambda+2\mu}[ \bigl\langle  \mathcal{E}_{m,s} +\id_3-  (0|0|\nabla_x\Theta(0)\, e_3)[\nabla_x \Theta(0)]^{-1}  ,\id_3\bigr\rangle -2] \,=\,1-\frac{\lambda}{\lambda+2\mu}\tr( \mathcal{E}_{m,s} ), \notag \\
\dd\varrho_b^e\,=\,&-\frac{\lambda}{\lambda+2\mu} \bigl\langle  {\rm C}_{y_0}\mathcal{K}_{e,s} -{\rm B}_{y_0},\id_3\bigr\rangle  -\frac{\lambda}{\lambda+2\mu} \bigl\langle  (\mathcal{E}_{m,s} +\id_3-  (0|0|\nabla_x\Theta(0)\, e_3)[\nabla_x \Theta(0)]^{-1})\,{\rm B}_{y_0},\id_3\bigr\rangle \, \\
\,=\,&-\frac{\lambda}{\lambda+2\mu} \bigl\langle  {\rm C}_{y_0}\mathcal{K}_{e,s} +\mathcal{E}_{m,s} {\rm B}_{y_0},\id_3\bigr\rangle  +\frac{\lambda}{\lambda+2\mu} \bigl\langle   (0|0|\nabla_x\Theta(0)\, e_3)[\nabla_x \Theta(0)]^{-1}\,{\rm B}_{y_0},\id_3\bigr\rangle \, \notag \\
\,=\,&-\frac{\lambda}{\lambda+2\mu}{\rm tr} [{\rm C}_{y_0}\mathcal{K}_{e,s} +\mathcal{E}_{m,s} {\rm B}_{y_0}]\,.\notag
\end{align}
Therefore, in view of \eqref{final_rho-Dumi} and item ii) of Lemma \ref{LemaEsEe}, the coefficients $A_1$ and $A_2$ are given by
\begin{align}\label{e70Dumi-N}
A_1\,&=\,  -\dfrac{\lambda}{\lambda+2\mu}\,\mathrm{tr}\,\Big(  \mathcal{E}_{m,s} \big( {\rm B}_{y_0} - 2\,{\rm H}\,  {\rm A}_{y_0}\big)+  {\rm C}_{y_0}  \mathcal{K}_{e,s}       \Big),\\
A_2\,&=\,  \dfrac{\lambda}{\lambda+2\mu}\,\Big[  \, 2\,{\rm H}\, \mathrm{tr}\,\big(  \mathcal{E}_{m,s} {\rm B}_{y_0}  +  {\rm C}_{y_0}  \mathcal{K}_{e,s} \big) - {\rm K}   \,\mathrm{tr}\,  \mathcal{E}_{m,s}     \Big].\notag
\end{align}
We also note that
\begin{align}\label{e5.60}
(0|0| \,\nabla_x\Theta(0)\, e_3)[	\nabla_x \Theta(0)]^{-1} \,=\,(0|0|n_0)(0|0|n_0)^T\,=\,n_0\otimes n_0.
\end{align} 
In conclusion, the tensor $\widetilde{\mathcal{E}}_{s}$ defined by \eqref{defEes} is completely  expressed in terms of $\mathcal{E}_{m,s} , {\rm B}_{y_0}, {\rm A}_{y_0}, {\rm C}_{y_0}$, $\mathcal{K}_{e,s} $ and $n_0\otimes n_0$.
Similarly, we write $\Gamma_{s} $ defined by \eqref{defEes} in terms of $ {\rm B}_{y_0}$ and $\mathcal{K}_{e,s}$
\begin{align}\label{gammaDumi}
\Gamma_{s} & 
\,= \,\frac{1}{b(x_3)}\,(\mathrm{axl}(\overline{Q}_{e,s}^T\,\partial_{x_1} \overline{Q}_e)\,|\, \mathrm{axl}(\overline{Q}_{e,s}^T\,\partial_{x_2} \overline{Q}_{e,s})\,|0)\,\left[\id_3+x_3({\rm L}_{y_0}^\flat-2\,{\rm H}\, \id_3)+x_3^2\, {\rm K}\begin{footnotesize}\begin{pmatrix}
0&0&0 \\
0&0&0 \\
0&0&1
\end{pmatrix}\end{footnotesize}
\right] [	\nabla_x \Theta(0)]^{-1}\notag\\
&\,= \,\frac{1}{b(x_3)}\,\mathcal{K}_{e,s}\,[	\nabla_x \Theta(0)]\,\Big[\id_3+x_3({\rm L}_{y_0}^\flat-2\,{\rm H}\, \id_3)\Big]\, [	\nabla_x \Theta(0)]^{-1}\\&
\qquad\qquad+x_3^2\, {\rm K}\, \,\frac{1}{b(x_3)}\,(\mathrm{axl}(\overline{Q}_{e,s}^T\,\partial_{x_1} \overline{Q}_{e,s})\,|\, \mathrm{axl}(\overline{Q}_{e,s}^T\,\partial_{x_2} \overline{Q}_{e,s})\,|0)\begin{footnotesize}\begin{pmatrix}
0&0&0 \\
0&0&0 \\
0&0&1
\end{pmatrix}\end{footnotesize}
[	\nabla_x \Theta(0)]^{-1}\notag\\
&\,=\, \,\frac{1}{b(x_3)}\,\Big[\mathcal{K}_{e,s} + x_3 \,(    \mathcal{K}_{e,s} {\rm B}_{y_0}   - 2\,{\rm H}\,    \mathcal{K}_{e,s} )\Big].\notag
\end{align}

\subsection{Dimensionally reduced energy: analytical integration through the thickness}

In what follows, we  find the expression of the strain energy density $ \, W \,=\,W_{\mathrm{mp}}( \widetilde{\mathcal{E}}_{s})+ W_{\mathrm{curv}}(    \Gamma_s)\,$ and   integrate it over the thickness. To this aim, we introduce the following bilinear forms
\begin{align}\label{e71}
\mathcal{W}_{\mathrm{mp}}(  X,  Y)& \,=\,   \mu\, \bigl\langle  \mathrm{sym}\,X,\mathrm{sym}\,   \,Y\bigr\rangle  +  \mu_{\rm c} \bigl\langle  \,\mathrm{skew} \,X,\mathrm{skew}\,   \,Y\bigr\rangle  +\,\dfrac{\lambda}{2}\,\mathrm{tr}  (X)\,\mathrm{tr}  (Y)  \notag \\
& \,=\,    \mu\, \bigl\langle  \mathrm{  dev \,sym}\,X,\mathrm{  dev \,sym}\,   \,Y\bigr\rangle  +  \mu_{\rm c} \bigl\langle  \mathrm{skew}   \,X,\mathrm{skew}\,   \,Y\bigr\rangle  +\,\dfrac{\kappa}{2}\,\mathrm{tr}  (X)\,\mathrm{tr}  (Y),   \\
\mathcal{W}_{\mathrm{curv}}(  X,  Y)& \,=\,\mu\,{L}_{\rm c}^2\,\bigg( b_1\, \bigl\langle  \dev\,\text{sym} \,X , \dev\,\text{sym} \,Y\bigr\rangle +b_2\, \bigl\langle \text{skew}\,X , \skw\,Y\bigr\rangle +\,4\,b_3\,
\tr(X)\,\tr\,(Y)\,\bigg)\notag
\end{align}
for any  $\,   X,\,    Y\in \mathbb{R}^{3\times 3}$. We remark the identities
$
W_{\mathrm{mp}}(  X)\,=\, \mathcal{W}_{\mathrm{mp}}(  X,  X),\ 
W_{\mathrm{curv}}(  X )\,=\, \mathcal{W}_{\mathrm{curv}}(  X,  X).
$ 
Thus, using \eqref{e69Dumi}, Lemma \ref{LemaEsEe} and the notations \eqref{e70Dumi-N}, we obtain
\begin{align}\label{e73}
W_{\mathrm{mp}}( \widetilde{\mathcal{E}}_{s}) & \,=\,  \dfrac{1}{b^2(x_3)} \,W_{\mathrm{mp}} \Big(  \big[  \mathcal{E}_{m,s} + (\varrho_m^e-1)   n_0\otimes    n_0\big] +x_3\big[ \big(  \mathcal{E}_{m,s} \, {\rm B}_{y_0} +  {\rm C}_{y_0} \mathcal{K}_{e,s} \big) -   2\,{\rm H}\,   \mathcal{E}_{m,s} + A_1   n_0\otimes    n_0 \big]  \notag\\
&    \qquad\qquad\qquad \quad
+ x_3^2 \,\big[  {\rm C}_{y_0} \mathcal{K}_{e,s} {\rm B}_{y_0}   - 2\,{\rm H}\,   {\rm C}_{y_0} \mathcal{K}_{e,s} + A_2   n_0\otimes    n_0   \big]  + x_3^3\, {\rm K}\,\varrho_b^e \,  n_0\otimes    n_0\Big).
\end{align}
In order to perform the analytical integration over the thickness, we write $W_{\mathrm{mp}}( \widetilde{\mathcal{E}}_{s})$  as a polynomial in $x_3$ with the coefficients $C_k$, i.e.,
\begin{equation}\label{e74}
W_{\mathrm{mp}}( \widetilde{\mathcal{E}}_{s}) \,=\, \dfrac{1}{b^2(x_3)}\,  \Big( \, \sum_{k\,=\,0}^6\,C_k(x_1,x_2)\,x_3^k\,     \Big),
\end{equation}
where
\begin{align}
C_0(x_1,x_2) &\,=\, W_{\mathrm{mp}} \Big(     \mathcal{E}_{m,s} + (\varrho_m^e-1)   n_0\otimes    n_0\Big),  \notag \\
C_1(x_1,x_2) &\,=\,  2 \, \mathcal{W}_{\mathrm{mp}}   \Big(\,  \mathcal{E}_{m,s} + (\varrho_m^e-1)   n_0\otimes    n_0,\,  \big(  \mathcal{E}_{m,s} \, {\rm B}_{y_0} +  {\rm C}_{y_0} \mathcal{K}_{e,s} \big) -   2\,{\rm H}\,   \mathcal{E}_{m,s} + A_1   n_0\otimes    n_0 \Big),   \notag \\
C_2(x_1,x_2) & \,=\,  W_{\mathrm{mp}} \Big(  \big(  \mathcal{E}_{m,s} \, {\rm B}_{y_0} +  {\rm C}_{y_0} \mathcal{K}_{e,s} \big) -   2\,{\rm H}\,   \mathcal{E}_{m,s} + A_1   n_0\otimes    n_0\Big)  \notag \\
&   
\qquad +  2 \, \mathcal{W}_{\mathrm{mp}}   \Big(\,  \mathcal{E}_{m,s} + (\varrho_m^e-1)   n_0\otimes    n_0,\,  {\rm C}_{y_0} \mathcal{K}_{e,s} {\rm B}_{y_0}   - 2\,{\rm H}\,   {\rm C}_{y_0} \mathcal{K}_{e,s} + A_2   n_0\otimes    n_0 \Big),\label{e75}
\\
C_3(x_1,x_2) & \,=\,   2\,\mathcal{W}_{\mathrm{mp}} \Big(    \mathcal{E}_{m,s} + (\varrho_m^e-1)   n_0\otimes    n_0,\,   {\rm K}\,\varrho_b^e \,  n_0\otimes    n_0  \Big) \notag  \\
&   
\qquad +  2 \, \mathcal{W}_{\mathrm{mp}}   \Big(\,\big(  \mathcal{E}_{m,s} \, {\rm B}_{y_0} +  {\rm C}_{y_0} \mathcal{K}_{e,s} \big) -   2\,{\rm H}\,   \mathcal{E}_{m,s} + A_1   n_0\otimes    n_0,\,  {\rm C}_{y_0} \mathcal{K}_{e,s} {\rm B}_{y_0}   - 2\,{\rm H}\,   {\rm C}_{y_0} \mathcal{K}_{e,s} + A_2   n_0\otimes    n_0 \Big),\notag
\\
C_4(x_1,x_2)  & \,=\,  W_{\mathrm{mp}} \Big(  {\rm C}_{y_0} \mathcal{K}_{e,s} {\rm B}_{y_0}   - 2\,{\rm H}\,   {\rm C}_{y_0} \mathcal{K}_{e,s} + A_2   n_0\otimes    n_0 \Big) \notag \\
&   
\qquad +  2 \, \mathcal{W}_{\mathrm{mp}}   \Big(\,\big(  \mathcal{E}_{m,s} \, {\rm B}_{y_0} +  {\rm C}_{y_0} \mathcal{K}_{e,s} \big) -   2\,{\rm H}\,   \mathcal{E}_{m,s} + A_1   n_0\otimes    n_0,\,  {\rm K}\,\varrho_b^e \,  n_0\otimes    n_0 \Big), \notag \end{align}
\begin{align}
\hspace{-3.5cm}C_5(x_1,x_2) & \,=\,   2 \, \mathcal{W}_{\mathrm{mp}}   \Big(  {\rm C}_{y_0} \mathcal{K}_{e,s} {\rm B}_{y_0}   - 2\,{\rm H}\,   {\rm C}_{y_0} \mathcal{K}_{e,s} + A_2   n_0\otimes    n_0,\,  {\rm K}\,\varrho_b^e \,  n_0\otimes    n_0 \Big), \notag \\
\hspace{-3.5cm} C_6(x_1,x_2) & \,=\,  \, W_{\mathrm{mp}}   \Big(  {\rm K}\,\varrho_b^e \,  n_0\otimes    n_0 \Big).\notag
\end{align}
Making use of the expansion
(since $x_3 \in \big(-\frac{h}{2}, \frac{h}{2}\,\big)$ and $\,h\,$ is small)
\begin{align}\label{e76}
\dfrac{1}{b(x_3)} &\,=\,   \dfrac{1}{1-2\,{\rm H}\,x_3+{\rm K}\,x_3^2} \notag\\ & \,=\,  1+ 2\,{\rm H}\,x_3+ (4\,{\rm H}^2\,-{\rm K})\,x_3^2+
(8\,{\rm H}^3\,-4\,{\rm H\, K})\,x_3^3     +({\rm K}^2-12 \,{\rm H}^2 \, {\rm K}+16\,{\rm H}^4)\,x_3^4+O(x_3^5),
\end{align}
and of the relations \eqref{minprmod} and Proposition \refeq{propnablatheta} i), the integration can be pursued as follows
\begin{align}\label{e77}
\dd \int_{\Omega_h}   \!\!W_{\mathrm{mp}}( \widetilde{\mathcal{E}}_{s})  \,\mathrm{det} \big[  \nabla_x   \Theta(  x) \big]\, \mathrm d V \,=\,
\int_{\Omega_h}   &\!\Big( \, \sum_{k\,=\,0}^6\,C_k(x_1,x_2)\,x_3^k  \Big)\!\Big[ 1+ 2\,{\rm H}\,x_3+ (4\,{\rm H}^2\,-{\rm K})\,x_3^2+    (8\,{\rm H}^3\,-4\,{\rm H\, K})x_3^3   \notag    \\
&\qquad\qquad\qquad  +({\rm K}^2-12 \,{\rm H}^2 \, {\rm K}+16\,{\rm H}^4)\,x_3^4+O(x_3^5)
\Big] \,{\rm det}(\nabla y_0|n_0)\,\mathrm dV   \notag\\
\,=\, \dd\int_\omega &\Big\{   h\,C_0+\,\dfrac{h^3}{12}\,\Big[ (4\,{\rm H}^2\,-{\rm K})C_0+ 2\,{\rm H}\,C_1+C_2 \Big]\notag \\&\quad+ \,\dfrac{h^5}{80}\,\Big[ ({\rm K}^2-12 \,{\rm H}^2 \, {\rm K}+16\,{\rm H}^4)\,C_0    + (8\,{\rm H}^3\,-4\,{\rm H\, K})\,C_1  \\
&\qquad\qquad +  (4\,{\rm H}^2\,-{\rm K})\,C_2+ 2\,{\rm H}\,C_3+C_4 \Big]
\Big\}\,{\rm det}(\nabla y_0|n_0) \,\mathrm da +O(h^7),\notag
\end{align}
where ${\rm d} a\,=\,{\rm d} x_1{\rm d}x_2$. 

In view of \eqref{e77}, we need to find appropriate expressions for the coefficients $\,C_0,\,C_1,\,C_2,\,C_3,\,C_4\,$ defined by \eqref{e75}. In this line, we designate by $\, \mathcal{W}_{\mathrm{shell}}(  X,  Y)\,$ the bilinear form
\begin{align}\label{e78}
\mathcal{W}_{\mathrm{shell}}(  X,  Y)& \,=\,   \mu\, \bigl\langle  \mathrm{sym}\,  \,X,\mathrm{sym}\,   \,Y\bigr\rangle  +  \mu_{\rm c} \bigl\langle  \,\mathrm{skew} \,X,\mathrm{skew}\,   \,Y\bigr\rangle  +\,\dfrac{\lambda\,\mu}{\lambda+2\mu}\,\mathrm{tr}  (X)\,\mathrm{tr}  (Y)  \notag \\
& \,=\,   \mu\, \bigl\langle  \mathrm{  dev \,sym} \,X,\mathrm{  dev \,sym}\,   \,Y\bigr\rangle  +  \mu_{\rm c} \bigl\langle  \mathrm{skew}   \,X,\mathrm{skew}\,   \,Y\bigr\rangle  +\,\dfrac{2\,\mu\,(2\,\lambda+\mu)}{3(\lambda+2\,\mu)}\,\mathrm{tr}  (X)\,\mathrm{tr}  (Y),  \\
W_{\mathrm{shell}}(  X) & \,=\,   \mathcal{W}_{\mathrm{shell}}(  X,  X)\notag
\end{align}
and we observe that
\begin{equation}\label{e79}
\mathcal{W}_{\mathrm{shell}}(  X,  Y)+ \,\dfrac{\lambda^2}{2\,(\lambda+2\mu)}\,\mathrm{tr}  (X)\,\mathrm{tr}  (Y)\,=\, \mathcal{W}_{\mathrm{mp}}(  X,  Y) ,
\end{equation}
since $  \,\dfrac{\kappa}{2 }\, - \,\dfrac{\lambda^2}{2\,(\lambda+2\mu)} \,=\, \,\dfrac{2\,\mu\,(2\lambda+\mu)}{3\,(\lambda+2\mu)}\,$ . Using the notations \eqref{e71}, \eqref{e78}, we obtain
\begin{lemma} The following identities 
\begin{align}\label{e80}
\mathcal{W}_{\mathrm{mp}}\big(  X+\alpha\,  n_0\otimes    n_0,\,  Y+\beta\,  n_0\otimes    n_0\big)&\,=\, \mathcal{W}_{\mathrm{mp}}(  X,  Y)+ \dfrac{\lambda}{2}\,\big( \alpha\,\mathrm{tr}  (Y)+\beta\, \mathrm{tr}  (X)\big)+\dfrac{\lambda+2\mu}{2} \,\alpha\,\beta,
 \notag\\
\mathcal{W}_{\mathrm{mp}}\Big(  X-\,\dfrac{\lambda}{\lambda+2\mu}\,\big( \mathrm{tr}  (X)\big)  n_0\otimes    n_0,\,  Y+\beta\,  n_0\otimes    n_0\Big)&\,=\,  \mathcal{W}_{\mathrm{shell}}(  X,  Y),
\end{align}
 hold true for all tensors $\,   X,\,    Y\in \mathbb{R}^{3\times 3}$ of the form $(*|*|0)\cdot[	\nabla_x \Theta(0)]^{-1}$ and all $\alpha,\beta\in \mathbb{R}$.
\end{lemma}
\begin{proof}
	In view of the definition \eqref{e71} we see that
	\begin{align}
	\mathcal{W}_{\mathrm{mp}}\big(  X+\alpha\,  n_0\otimes    n_0,\,  Y+\beta\,  n_0\otimes    n_0\big)\,=\,&  \mu\, \bigl\langle  \mathrm{sym} \,X\,+\alpha\,  n_0\otimes    n_0, \mathrm{sym}\,   \,Y\,+\beta\,  n_0\otimes    n_0\bigr\rangle 
	 \\
&	+  \mu_{\rm c} \bigl\langle  \,\mathrm{skew} \,X,\mathrm{skew}\, Y\bigr\rangle  +\,\dfrac{\lambda}{2}\,\big( \mathrm{tr} (  X)+\alpha\big)\,\big(\mathrm{tr}{(Y)}
\,+\beta\big).\notag
	\end{align}
According to \eqref{e5.60},  $  n_0\otimes    n_0 \,=\, (0|0|n_0)\cdot[	\nabla_x \Theta(0)]^{-1} $.
		Since for $\,   X\,=\,(*|*|0)\cdot[	\nabla_x \Theta(0)]^{-1}$  we have \begin{align}
	 \bigl\langle  \mathrm{sym} \,X ,     n_0\otimes    n_0\bigr\rangle \,&=\, 	 \bigl\langle   (*|*|0)\,[	\nabla_x \Theta(0)]^{-1},    (0|0|n_0)\,[	\nabla_x \Theta(0)]^{-1}\bigr\rangle 
	\notag \\
	&=\, 	 \bigl\langle   
	(0|0|n_0)^T
	\,(*|*|0) ,    [[	\nabla_x \Theta(0)]^T\,	\nabla_x \Theta(0)]^{-1}\bigr\rangle \,=	 \bigl\langle   \begin{footnotesize}
	\begin{pmatrix}
	0&0&0\\
	0&0&0\\
	*&*&0
	\end{pmatrix}
	\end{footnotesize},  \widehat{\rm I}_{y_0}^{-1}\bigr\rangle \,=\,0\,,
	\end{align} we get
	\begin{align}
	\mathcal{W}_{\mathrm{mp}}\big(  X+\alpha\,  n_0\otimes    n_0,\,  Y+\beta\,  n_0\otimes    n_0\big)\,=\,&  \mu\, \bigl\langle  \mathrm{sym} \,X,\mathrm{sym}\,   \,Y\bigr\rangle  + \mu\,\alpha\,\beta\,  \bigl\langle    n_0\otimes    n_0 ,    n_0\otimes    n_0\bigr\rangle  \notag
	 +
	\mu_{\rm c} \bigl\langle  \,\mathrm{skew} \,X,\mathrm{skew}\,   \,Y\bigr\rangle   \\
	&+\,\dfrac{\lambda}{2}\,\mathrm{tr}  (X)\,\mathrm{tr}  (Y) +\dfrac{\lambda}{2}\,\big( \alpha\,\mathrm{tr}  (Y)+\beta\, \mathrm{tr}  (X)\big)+\dfrac{\lambda}{2} \,\alpha\,\beta \notag
	 \\
	\,=\,& \mathcal{W}_{\mathrm{mp}}(  X,  Y)+ \dfrac{\lambda}{2}\,\big( \alpha\,\mathrm{tr}  (Y)+\beta\, \mathrm{tr}  (X)\big)+\dfrac{\lambda+2\mu}{2} \,\alpha\,\beta.
	\end{align}
	This means that the relation \eqref{e80}$_1$ holds true, for any  $\,\alpha,\beta\in\mathbb{R}$.
	
	If we write \eqref{e80}$_1$ with $\,\alpha\,= \,\dfrac{-\lambda}{\lambda+2\mu}\,\big( \mathrm{tr}  (X)\big)$, then we obtain
	\begin{equation*} 
	\begin{array}{l}
	\mathcal{W}_{\mathrm{mp}}\Big(  X-\,\dfrac{\lambda}{\lambda+2\mu}\,\big( \mathrm{tr}  (X)\big)  n_0\otimes    n_0,\,  Y+\beta\,  n_0\otimes    n_0\Big)\,=\,   \\
	\qquad\qquad \,=\,
	\mathcal{W}_{\mathrm{mp}}(  X,  Y)
	+ \dfrac{\lambda}{2}\,\Big( \,\dfrac{-\lambda}{\lambda+2\mu}\, \mathrm{tr}  (X)\,\mathrm{tr}  (Y))+\beta\, \mathrm{tr}  (X)\Big) +\dfrac{\lambda+2\mu}{2}\cdot\dfrac{-\lambda}{\lambda+2\mu}\, \mathrm{tr}  (X)\,\beta
	 \\
	\qquad\qquad \,=\,\mathcal{W}_{\mathrm{mp}}(  X,  Y) - \,\dfrac{\lambda^2}{2\,(\lambda+2\,\mu)}\,\mathrm{tr}  (X)\,\mathrm{tr}  (Y)\,=\,
	\mathcal{W}_{\mathrm{shell}}(  X,  Y),
	\end{array}
	\end{equation*}
	where we have  used  the formula \eqref{e79}.
	Thus, the relation \eqref{e80}$_2$ is also proved.
\end{proof}
  By virtue of \eqref{e80}$_2$, \eqref{final_rho-Dumi}{$_1$}, \eqref{e70Dumi-N}, \eqref{e55} and Lemma \ref{LemaEsEe} we get
\begin{align}\label{e81}
\mathcal{W}_{\mathrm{mp}}&\Big(  \mathcal{E}_{m,s} + \big( \varrho_m^e-1\big)  n_0\otimes    n_0,\,  Y+\beta\,  n_0\otimes    n_0\Big)\,=\,  \mathcal{W}_{\mathrm{shell}}(  \mathcal{E}_{m,s} ,  Y), \notag \\
\mathcal{W}_{\mathrm{mp}}&\Big(\big(  \mathcal{E}_{m,s} \, {\rm B}_{y_0} +  {\rm C}_{y_0} \mathcal{K}_{e,s} \big) -   2\,{\rm H}\,   \mathcal{E}_{m,s} + A_1   n_0\otimes    n_0, Y+\beta\,  n_0\otimes    n_0\Big)= \mathcal{W}_{\mathrm{shell}}\big(  \mathcal{E}_{m,s} \, {\rm B}_{y_0} +  {\rm C}_{y_0} \mathcal{K}_{e,s} \! -\!   2\,{\rm H}\,   \mathcal{E}_{m,s} , Y\big), \notag \\
W_{\mathrm{mp}}&\Big(   {\rm C}_{y_0} \mathcal{K}_{e,s} {\rm B}_{y_0}   - 2\,{\rm H}\,   {\rm C}_{y_0} \mathcal{K}_{e,s} + A_2   n_0\otimes    n_0 \Big) \,=\, W_{\mathrm{shell}}\big(   {\rm C}_{y_0} \mathcal{K}_{e,s} {\rm B}_{y_0}   - 2\,{\rm H}\,  {\rm C}_{y_0} \mathcal{K}_{e,s}  \big)  \\&\qquad \qquad\qquad\qquad\qquad\qquad\qquad\qquad\qquad\qquad + \dfrac{\lambda^2}{2(\lambda+2\mu)}\Big[   \mathrm{tr}\big((  \mathcal{E}_{m,s} \, {\rm B}_{y_0} +  {\rm C}_{y_0} \mathcal{K}_{e,s} )   {\rm B}_{y_0} \big) \Big]^2{,}
\notag
\end{align}
{
where for the last identity we used \eqref{e80}$_1$ and the relation 
\[
\mathcal{E}_{m,s} \, {\rm B}_{y_0}^2=2\,{\rm H}\,   \mathcal{E}_{m,s} \, {\rm B}_{y_0} -{\rm K}\,   \mathcal{E}_{m,s} \, {\rm A}_{y_0}=2\,{\rm H}\,   \mathcal{E}_{m,s} \, {\rm B}_{y_0} -{\rm K}\,   \mathcal{E}_{m,s}.
\]
}
From \eqref{e75} and \eqref{e81} we get
\begin{align}\label{e82}
C_0 & \,=\,   W_{\mathrm{shell}}\big(    \mathcal{E}_{m,s} \big),  \notag \\
C_1 & \,=\,  2 \,  \mathcal{W}_{\mathrm{shell}}  \big(  \mathcal{E}_{m,s} ,\,  \mathcal{E}_{m,s} \, {\rm B}_{y_0} +  {\rm C}_{y_0} \mathcal{K}_{e,s}  -   2\,{\rm H}\,   \mathcal{E}_{m,s} \big) \,=\, -4\,{\rm H}\,W_{\mathrm{shell}}\big(    \mathcal{E}_{m,s} \big)+ 2 \,  \mathcal{W}_{\mathrm{shell}}  \big(  \mathcal{E}_{m,s} ,\,  \mathcal{E}_{m,s} \, {\rm B}_{y_0} +  {\rm C}_{y_0} \mathcal{K}_{e,s} \big), \notag \\
C_2 & \,=\,    W_{\mathrm{shell}}  \big(   \mathcal{E}_{m,s} \, {\rm B}_{y_0} +  {\rm C}_{y_0} \mathcal{K}_{e,s}  -   2\,{\rm H}\,   \mathcal{E}_{m,s} \big)+ 2 \,  \mathcal{W}_{\mathrm{shell}}  \big(  \mathcal{E}_{m,s} ,\, {\rm C}_{y_0} \mathcal{K}_{e,s} {\rm B}_{y_0}  -   2\,{\rm H}\,   {\rm C}_{y_0} \mathcal{K}_{e,s} \big),   \\
C_3 & \,=\,  2 \,  \mathcal{W}_{\mathrm{shell}}  \big(  \mathcal{E}_{m,s} \, {\rm B}_{y_0} +  {\rm C}_{y_0} \mathcal{K}_{e,s}  -   2\,{\rm H}\,   \mathcal{E}_{m,s} ,\, {\rm C}_{y_0} \mathcal{K}_{e,s} {\rm B}_{y_0}  -   2\,{\rm H}\,   {\rm C}_{y_0} \mathcal{K}_{e,s} \big),  \notag \\
C_4 & \,=\,   W_{\mathrm{shell}}  \big(  {\rm C}_{y_0} \mathcal{K}_{e,s} {\rm B}_{y_0}  -   2\,{\rm H}\,   {\rm C}_{y_0} \mathcal{K}_{e,s} \big)+\,\dfrac{\lambda^2}{2(\lambda+2\mu)}\,\Big[   \mathrm{tr}\big((  \mathcal{E}_{m,s} \, {\rm B}_{y_0} +  {\rm C}_{y_0} \mathcal{K}_{e,s} )   {\rm B}_{y_0} \big) \Big]^2.\notag
\end{align}
With the  relations \eqref{e82} we can replace the coefficients $C_0, C_1,C_2, C_3, C_4$  appearing in \eqref{e77} and we obtain 
\begin{align}\label{e83}
(4\,{\rm H}^2&-{\rm K})C_0+ 2\,{\rm H}\,C_1+C_2 \,=\,  -{\rm K}\, W_{\mathrm{shell}}\big(    \mathcal{E}_{m,s} \big) + W_{\mathrm{shell}}  \big(   \mathcal{E}_{m,s} \, {\rm B}_{y_0} +  {\rm C}_{y_0} \mathcal{K}_{e,s} \big)   \notag\\&\qquad \qquad \qquad\qquad \qquad\qquad+ 2 \,  \mathcal{W}_{\mathrm{shell}}  \big(  \mathcal{E}_{m,s} ,\, {\rm C}_{y_0} \mathcal{K}_{e,s} {\rm B}_{y_0}  -   2\,{\rm H}\,   {\rm C}_{y_0} \mathcal{K}_{e,s} \big), \vspace{8pt}\notag \\
({\rm K}^2- & 12 \,{\rm H}^2 \, {\rm K}+16\,{\rm H}^4)C_0
+ (8\,{\rm H}^3\,-4\,{\rm H\, K})\,C_1+  (4\,{\rm H}^2\,-{\rm K})\,C_2+ 2\,{\rm H}\,C_3+C_4     \notag\\
&\qquad   \,=\,-{\rm K}\,  W_{\mathrm{shell}}  \big(   \mathcal{E}_{m,s} \, {\rm B}_{y_0} +   {\rm C}_{y_0} \mathcal{K}_{e,s} \big) +  W_{\mathrm{shell}} \big((  \mathcal{E}_{m,s} \, {\rm B}_{y_0} +  {\rm C}_{y_0} \mathcal{K}_{e,s} )   {\rm B}_{y_0}\, \big)   \\& \hspace{1.5cm}+\,\dfrac{\lambda^2}{2(\lambda+2\mu)}\,\Big[   \mathrm{tr}\big((  \mathcal{E}_{m,s} \, {\rm B}_{y_0} +  {\rm C}_{y_0} \mathcal{K}_{e,s} )   {\rm B}_{y_0} \big) \Big]^2   \notag \\
&\qquad \,=\,
-{\rm K}\,  W_{\mathrm{shell}}  \big(   \mathcal{E}_{m,s} \, {\rm B}_{y_0} +   {\rm C}_{y_0} \mathcal{K}_{e,s} \big) +  W_{\mathrm{mp}} \big((  \mathcal{E}_{m,s} \, {\rm B}_{y_0} +  {\rm C}_{y_0} \mathcal{K}_{e,s} )   {\rm B}_{y_0} \,\big)  .\notag
\end{align}
Inserting \eqref{e83} into \eqref{e77} (and neglecting the terms of order $O(h^7)$) we obtain the following result of the integration
\begin{align}\label{e84}
\dd \int_{\Omega_h}   \,W_{\mathrm{mp}}( \widetilde{\mathcal{E}}_{s})  \,\mathrm{det} \big[  \nabla_x   \Theta(  x) \big]\, \mathrm d  V \,=\,   
\dd\int_{\omega}   \, \Big[ & \Big(h-{\rm K}\,\dfrac{h^3}{12}\Big)\,
W_{\mathrm{shell}}\big(    \mathcal{E}_{m,s} \big)    +  \Big(\dfrac{h^3}{12}\,-{\rm K}\,\dfrac{h^5}{80}\Big)\,
W_{\mathrm{shell}}  \big(   \mathcal{E}_{m,s} \, {\rm B}_{y_0} +   {\rm C}_{y_0} \mathcal{K}_{e,s} \big)  \notag \\
& +
\dfrac{h^3}{12} \,2\,  \mathcal{W}_{\mathrm{shell}}  \big(  \mathcal{E}_{m,s} ,\, {\rm C}_{y_0} \mathcal{K}_{e,s} {\rm B}_{y_0}  -   2\,{\rm H}\,   {\rm C}_{y_0} \mathcal{K}_{e,s} \big)\\\
& + \,\dfrac{h^5}{80} \,
W_{\mathrm{mp}} \big((  \mathcal{E}_{m,s} \, {\rm B}_{y_0} +  {\rm C}_{y_0} \mathcal{K}_{e,s} )   {\rm B}_{y_0} \,\big)
\Big] \,{\rm det}(\nabla y_0|n_0)       \,\mathrm da  .\notag
\end{align}

Remark that using Lemma \ref{LemaEsEe}  and Proposition \ref{propAB} we deduce
\begin{align}
{\rm C}_{y_0}\mathcal{K}_{e,s}{\rm B}_{y_0}-2\,
{\mathrm H}
\, {\rm C}_{y_0}\, \mathcal{K}_{e,s}&\,=\, {\rm C}_{y_0}\, \mathcal{K}_{e,s}({\rm B}_{y_0}-2\, \mathrm{ H}\, {\rm A}_{y_0})\notag\\&\,=\,
( \mathcal{E}_{m,s}{\rm B}_{y_0}+{\rm C}_{y_0}\, \mathcal{K}_{e,s})({\rm B}_{y_0}-2\, \mathrm{ H}\, {\rm A}_{y_0})-\mathcal{E}_{m,s}{\rm B}_{y_0}({\rm B}_{y_0}-2\, 
{\mathrm H}
\, {\rm A}_{y_0})\\&\,=\,
( \mathcal{E}_{m,s}{\rm B}_{y_0}+{\rm C}_{y_0}\, \mathcal{K}_{e,s})({\rm B}_{y_0}-2\,\mathrm{ H}\, {\rm A}_{y_0})+
{\mathrm K}
\,\mathcal{E}_{m,s}\notag. 
\end{align}
Therefore, using again Lemma \ref{LemaEsEe}  and Proposition \ref{propAB}, the energy density is rewritten in the following form
\begin{align}\label{eq:ausgerechnete_koeff}
\Big(h-{\rm K}\,&\dfrac{h^3}{12}\Big)\,
W_{\mathrm{shell}}\big(    \mathcal{E}_{m,s} \big)    +  \Big(\dfrac{h^3}{12}\,-{\rm K}\,\dfrac{h^5}{80}\Big)\,
W_{\mathrm{shell}}  \big(   \mathcal{E}_{m,s} \, {\rm B}_{y_0} +   {\rm C}_{y_0} \mathcal{K}_{e,s} \big)   \notag\\
&\qquad\quad  +
\dfrac{h^3}{12} \,2\,  \mathcal{W}_{\mathrm{shell}}  \big(  \mathcal{E}_{m,s} ,\, {\rm C}_{y_0} \mathcal{K}_{e,s} {\rm B}_{y_0}  -   2\,{\rm H}\,   {\rm C}_{y_0} \mathcal{K}_{e,s} \big)+ \,\dfrac{h^5}{80} \,
W_{\mathrm{mp}} \big((  \mathcal{E}_{m,s} \, {\rm B}_{y_0} +  {\rm C}_{y_0} \mathcal{K}_{e,s} )   {\rm B}_{y_0} \,\big)\notag\\
\,=\,&
\Big(h+{\rm K}\,\dfrac{h^3}{12}\Big)\,
W_{\mathrm{shell}}\big(    \mathcal{E}_{m,s} \big)+  \Big(\dfrac{h^3}{12}\,-{\rm K}\,\dfrac{h^5}{80}\Big)\,
W_{\mathrm{shell}}  \big(   \mathcal{E}_{m,s} \, {\rm B}_{y_0} +   {\rm C}_{y_0} \mathcal{K}_{e,s} \big) \notag \\&+
\dfrac{h^3}{12} \,2\,  \mathcal{W}_{\mathrm{shell}}  \big(  \mathcal{E}_{m,s} ,
( \mathcal{E}_{m,s}{\rm B}_{y_0}+{\rm C}_{y_0}\, \mathcal{K}_{e,s})({\rm B}_{y_0}-2\, \mathrm{ H}\, {\rm A}_{y_0}) \big) \notag\\&+ \,\dfrac{h^5}{80} \,
W_{\mathrm{mp}} \big((  \mathcal{E}_{m,s} \, {\rm B}_{y_0} +  {\rm C}_{y_0} \mathcal{K}_{e,s} )   {\rm B}_{y_0} \,\big)\notag\\
\,=\,&
\Big(h+{\rm K}\,\dfrac{h^3}{12}\Big)\,
W_{\mathrm{shell}}\big(    \mathcal{E}_{m,s} \big)+  \Big(\dfrac{h^3}{12}\,-{\rm K}\,\dfrac{h^5}{80}\Big)\,
W_{\mathrm{shell}}  \big(   \mathcal{E}_{m,s} \, {\rm B}_{y_0} +   {\rm C}_{y_0} \mathcal{K}_{e,s} \big)  \\&
-\dfrac{h^3}{3} \mathrm{ H}\,\mathcal{W}_{\mathrm{shell}}  \big(  \mathcal{E}_{m,s} ,
\mathcal{E}_{m,s}{\rm B}_{y_0}+{\rm C}_{y_0}\, \mathcal{K}_{e,s} \big)+
\dfrac{h^3}{6}\, \mathcal{W}_{\mathrm{shell}}  \big(  \mathcal{E}_{m,s} ,
( \mathcal{E}_{m,s}{\rm B}_{y_0}+{\rm C}_{y_0}\, \mathcal{K}_{e,s}){\rm B}_{y_0} \big) \notag\\&+ \,\dfrac{h^5}{80} \,
W_{\mathrm{mp}} \big((  \mathcal{E}_{m,s} \, {\rm B}_{y_0} +  {\rm C}_{y_0} \mathcal{K}_{e,s} )   {\rm B}_{y_0} \,\big).\notag
\end{align}

Analogously, using \eqref{gammaDumi},  we integrate the curvature part of the strain energy density
\begin{equation}\label{e85}
\begin{array}{l}
\dd \int_{\Omega_h}   \,W_{\mathrm{curv}}(  {\Gamma}_s)  \,\mathrm{det} \big[  \nabla_x   \Theta(  x) \big]\, \mathrm d  V\,=\,  \int_{\Omega_h}   \,W_{\mathrm{curv}}\big(  \mathcal{K}_{e,s} + x_3 \,(    \mathcal{K}_{e,s} {\rm B}_{y_0}   - 2\,{\rm H}\,    \mathcal{K}_{e,s} )\,\big)
\,\dfrac{{\rm det}(\nabla y_0|n_0)  }{b(x_3)} \,\mathrm d  V   \\
\qquad \,=\,
\dd\int_{\Omega_h}   \,\Big( D_0+ D_1\,x_3 +D_2\,x_3^2\,     \Big)\,\Big[ 1+ 2\,{\rm H}\,x_3+ (4\,{\rm H}^2\,-{\rm K})\,x_3^2
+    (8\,{\rm H}^3\,-4\,{\rm H\, K})x_3^3   \\
\qquad\qquad \qquad \qquad\qquad \qquad \qquad \qquad +({\rm K}^2-12 \,{\rm H}^2 \, {\rm K}+16\,{\rm H}^4)\,x_3^4+O(x_3^5)
\Big] \,{\rm det}(\nabla y_0|n_0)  \,\mathrm dV   \\
\qquad     \,=\, \dd\int_\omega \,\Big\{   h\,D_0+\,\dfrac{h^3}{12}\,\Big[ (4\,{\rm H}^2\,-{\rm K})D_0+ 2\,{\rm H}\,D_1+D_2 \Big]+ \,\dfrac{h^5}{80}\,\Big[ ({\rm K}^2-12 \,{\rm H}^2 \, {\rm K}+16\,{\rm H}^4)\,D_0      \\
\qquad\qquad\qquad \qquad\qquad \quad  + (8\,{\rm H}^3\,-4\,{\rm H\, K})\,D_1+  (4\,{\rm H}^2\,-{\rm K})\,D_2  \Big]
\Big\} \, {\rm det}(\nabla y_0|n_0)   \,\mathrm da +O(h^7),
\end{array}
\end{equation}
where we have denoted by $D_k$ the coefficients of $x_3^k$ ($k\,=\,0,1,2$) in the expression
\begin{align}\label{e86}
&W_{\mathrm{curv}}\big(  \mathcal{K}_{e,s} + x_3 \,(    \mathcal{K}_{e,s} {\rm B}_{y_0}   - 2\,{\rm H}\,    \mathcal{K}_{e,s} )\,\big)\,=\,  D_0(x_1,x_2)+ D_1(x_1,x_2)\,x_3 +D_2(x_1,x_2)\,x_3^2,\qquad\text{with}  \\
&D_0 \,=\, W_{\mathrm{curv}}\big(    \mathcal{K}_{e,s} \big), \quad
D_1 \,=\, 2 \, \mathcal{W}_{\mathrm{curv}} \big(  \mathcal{K}_{e,s} , \,   \mathcal{K}_{e,s} {\rm B}_{y_0}   - 2\,{\rm H}\,    \mathcal{K}_{e,s} \big),\quad
D_2 \,=\,   W_{\mathrm{curv}} \big(    \mathcal{K}_{e,s} {\rm B}_{y_0}   - 2\,{\rm H}\,    \mathcal{K}_{e,s} \big).\notag
\end{align}
We write the coefficients of $\,\dfrac{h^3}{12}\,$ and $\,\dfrac{h^5}{80}\,$ in \eqref{e85} with the help of \eqref{e86}$_{2,3,4}$ as follows
\begin{align}\label{e87}
&(4\,{\rm H}^2\,-{\rm K})D_0+ 2\,{\rm H}\,D_1+D_2 \,=\, -{\rm K}\, W_{\mathrm{curv}}\big(  \mathcal{K}_{e,s} \big) + W_{\mathrm{curv}}\big(  \mathcal{K}_{e,s}   {\rm B}_{y_0} \,  \big),\notag
 \\
&({\rm K}^2-12 \,{\rm H}^2 \, {\rm K}+16\,{\rm H}^4)\,D_0   + (8\,{\rm H}^3\,-4\,{\rm H\, K})\,D_1+  (4\,{\rm H}^2\,-{\rm K})\,D_2 \,=\,
 \\
&\qquad\qquad \,=\, {\rm K}^2\, W_{\mathrm{curv}}\big(  \mathcal{K}_{e,s} \big)
-4\,{\rm H\, K} \,  \mathcal{W}_{\mathrm{curv}}\big(  \mathcal{K}_{e,s} ,\,  \mathcal{K}_{e,s}   {\rm B}_{y_0} \,  \big)+ (4\,{\rm H}^2\,-{\rm K})\, W_{\mathrm{curv}}\big(  \mathcal{K}_{e,s}   {\rm B}_{y_0} \,  \big)\notag
 \\
&\qquad\qquad
\,=\, -{\rm K}\, W_{\mathrm{curv}}\big(  \mathcal{K}_{e,s} {\rm B}_{y_0} \,   \big) + W_{\mathrm{curv}}\big(  \mathcal{K}_{e,s}   {\rm B}_{y_0}^2  \big).\notag
\end{align}
Inserting \eqref{e87} into \eqref{e85} (and again neglecting the terms of order $O(h^7)$) we arrive at the following result of this integration
\begin{align}\label{e88}
\dd \int_{\Omega_h}   &\,W_{\mathrm{curv}}(  {\Gamma}_s)  \,\mathrm{det} \big[  \nabla_x   \Theta(  x) \big]\, \mathrm d  V \,=\,  \\
& \,=\,
\dd\int_{\omega}    \, \Big[  \Big(h-{\rm K}\,\dfrac{h^3}{12}\Big)\,
W_{\mathrm{curv}}\big(  \mathcal{K}_{e,s} \big)    +  \Big(\dfrac{h^3}{12}\,-{\rm K}\,\dfrac{h^5}{80}\Big)\,
W_{\mathrm{curv}}\big(  \mathcal{K}_{e,s}   {\rm B}_{y_0} \,  \big)  + \,\dfrac{h^5}{80} \,
W_{\mathrm{curv}}\big(  \mathcal{K}_{e,s}    {\rm B}_{y_0}^2  \big)
\Big] \,{\rm det}(\nabla y_0|n_0)        \,\mathrm d a .\notag
\end{align}

In order to write the external loads potential in the shell model, we perform next the integration over the thickness of the relations \eqref{loadpot4}. Thus, from \eqref{defTheta} and \eqref{ansatz} we find
\begin{align}
\tilde{v}(x_i)  & \,=\, \varphi(x_i) - \Theta(x_i)\,=\, \Big( m + x_3\varrho_m\overline{Q}_{e,s}\, n_0 +\dd\frac{x_3^2}{2}\varrho_b \overline{Q}_{e,s}\, n_0 \Big) - (y_0+x_3n_0) 
\notag \\
& \,=\, (m-y_0) + x_3 (\varrho_m\overline{Q}_{e,s}\, n_0 -n_0) + \dd\frac{x_3^2}{2}\varrho_b \overline{Q}_{e,s}\, n_0\,.
\notag
\end{align}
We insert this into \eqref{loadpot4}$ _1 $ and use the approximation $\varrho_m\cong \varrho_m ^0=1  $ , $ \varrho_b\cong \varrho_b^0=0 $ as in \eqref{final_rho} to obtain the simplified form
\begin{equation}\label{e2o}
\begin{array}{l}
\dd\int_{\Omega_h} \bigl\langle  \tilde{f}, \tilde{v} \bigr\rangle   \, dV\, \,=\, 
\int_{\omega} \left( \bigl\langle \int_{-h/2}^{h/2} \tilde{f}\,dx_3 ,   m-y_0 \bigr\rangle    + 
 \bigl\langle \int_{-h/2}^{h/2} x_3 \tilde{f}\,dx_3 ,  (\overline{Q}_{e,s}-\id_3)\, n_0 \bigr\rangle   
\right) da\,.
\end{array}
\end{equation}
Denoting with $ \tilde{t}^{\pm}(x_1,x_2):\,=\, \tilde{t}(x_1,x_2, \pm \frac{h}{2} ) $ and taking into account that $ \Gamma_t \,=\,  \Big(\omega\times \Big\{\frac{h}{2}\Big\} \Big) \cup \Big(\omega\times \Big\{-\frac{h}{2}\Big\} \Big) \cup \Big(\gamma_t\times (-\frac{h}{2} , \frac{h}{2})\Big) $, we obtain similarly
\begin{align} 
\dd\int_{\Gamma_t} \bigl\langle  \tilde{t}, \tilde{v} \bigr\rangle   \, dS  \,=\, &\int_{\omega}  \bigl\langle  \tilde{t}^{\pm},  (m-y_0)  \pm
\frac{h}{2}  (\varrho_m\overline{Q}_{e,s}\, n_0-n_0) + \frac{h^2}{8}  \varrho_b\overline{Q}_{e,s}\, n_0 \bigr\rangle   
\,da
\notag \\
& + \int_{\gamma_t} \int_{-h/2}^{h/2}  \bigl\langle  \,\tilde{t}, (m-y_0)  + 
x_3  (\varrho_m\overline{Q}_{e,s}\, n_0-n_0) + \frac{x_3^2}{2}  \varrho_b\overline{Q}_{e,s}\, n_0 \bigr\rangle   \,
dx_3\,ds
\notag\,.
\end{align}
Using the same approximation as before ($\varrho_m\cong \varrho_m ^0=1  $ , $ \varrho_b\cong \varrho_b^0=0 $, see \eqref{final_rho}) we find
\begin{align} \label{e3o}
\dd\int_{\Gamma_t} \bigl\langle  \tilde{t},  \tilde{v} \bigr\rangle   \, dS  \,=\, & \int_{\omega}  \bigl\langle  \tilde{t}^{+} + \tilde{t}^{-}  ,   m-y_0  \bigr\rangle    \,da 
+ 
\int_{\omega}  \bigl\langle  
\frac{h}{2}  (\tilde{t}^{+} - \tilde{t}^{-}),  (\overline{Q}_{e,s}-\id_3)\, n_0   \bigr\rangle   \,
da  \\
& + \int_{\gamma_t}   \bigl\langle  \int_{-h/2}^{h/2} \tilde{t}\,dx_3 ,  m-y_0  \bigr\rangle   \,ds + 
\int_{\gamma_t}   \bigl\langle  \int_{-h/2}^{h/2} x_3\tilde{t}\,dx_3 ,   (\overline{Q}_{e,s}-\id_3)\, n_0  \bigr\rangle   \,ds 
\notag\,,
\end{align}
where $ ds $ is the arclength element along the curve $ \gamma_t $ and $ da\,=\, dx_1dx_2\, $.
With \eqref{e2o} and \eqref{e3o}, the potential of external applied loads $ \overline{\Pi}(m,\overline{Q}_{e,s})\,=\, \widetilde{\Pi}(\varphi, \overline{R}) $ in \eqref{loadpot4} can be written in the form
\begin{align}\label{e4o}
\overline{\Pi}(m,\overline{Q}_{e,s})\,=\,  \Pi_\omega(m,\overline{Q}_{e,s}) + \Pi_{\gamma_t}(m,\overline{Q}_{e,s})\,,
\end{align}
with
\begin{align}
\Pi_\omega(m,\overline{Q}_{e,s}) \,=\, \dd\int_{\omega} \bigl\langle  \bar{f},  \bar u \bigr\rangle   \, da + \Lambda_\omega(\overline{Q}_{e,s}),\qquad 
\Pi_{\gamma_t}(m,\overline{Q}_{e,s})\,=\, \dd\int_{\gamma_t} \bigl\langle  \bar{t},  \bar u \bigr\rangle   \, ds + \Lambda_{\gamma_t}(\overline{Q}_{e,s})\,,
\end{align}
where $ \bar u(x_1,x_2) \,=\, m(x_1,x_2)-y_0(x_1,x_2) $ is the displacement vector of the midsurface and
\begin{align} \label{e5o}
\bar f & \,=\, \dd \int_{-h/2}^{h/2} \tilde{f}\,dx_3 + (\tilde{t}^{+} + \tilde{t}^{-}), \qquad 
\bar t \,=\, \int_{-h/2}^{h/2} \tilde{t}\,dx_3 \,,  \\
\Lambda_\omega(\overline{Q}_{e,s}) & \,=\, \int_{\omega}  \bigl\langle  \dd \int_{-h/2}^{h/2} x_3\,\tilde{f}\,dx_3 +
\frac{h}{2}  (\tilde{t}^{+} - \tilde{t}^{-}),  (\overline{Q}_{e,s}-\id_3)\, n_0   \bigr\rangle   \,
da + \overline{\Pi}_\omega(\overline{Q}_{e,s}),
\notag \\
\Lambda_{\gamma_t}(\overline{Q}_{e,s}) & \,=\, \int_{\gamma_t}  \bigl\langle  \dd \int_{-h/2}^{h/2} x_3\,\tilde{t}\,dx_3 ,  (\overline{Q}_{e,s}-\id_3)\, n_0   \bigr\rangle   \,
ds + \overline{\Pi}_{\gamma_t}(\overline{Q}_{e,s})
\notag
\end{align}
and $\overline{\Pi}_\omega(\overline{Q}_{e,s})\,=\, \widetilde{\Pi}_{\Omega_h}(\overline{R})  $ , $ \overline{\Pi}_{\gamma_t}(\overline{Q}_{e,s})\,=\, \widetilde{\Pi}_{\Gamma_t}(\overline{R}) $, since $ \overline{R} $ is independent of $ x_3\,$.\\ The functions $\overline{\Pi}_\omega\,, \overline{\Pi}_{\gamma_t} : L^2 (\omega, \textrm{SO}(3))\rightarrow\mathbb{R} $ are assumed to be continuous and bounded operators.

\section{The new   geometrically nonlinear   Cosserat shell model}\setcounter{equation}{0}

\subsection{Formulation of the minimization problem}

Gathering our results, see  \eqref{minprmod}, \eqref{e84}{, \eqref{eq:ausgerechnete_koeff}} and \eqref{e88}, we have obtained the following two-dimensional minimization problem   for the
deformation of the midsurface $m:\omega
\,{\to}\,
\mathbb{R}^3$ and the microrotation of the shell
$\overline{Q}_{e,s}:\omega
\,{\to}\,
\textrm{SO}(3)$ solving on $\omega
\,\subset\mathbb{R}^2
$: minimize with respect to $ (m,\overline{Q}_{e,s}) $ the  functional
\begin{equation}\label{e89}
I\,=\, \int_{\omega}    \, \Big[  \,
W_{\mathrm{memb}}\big(  \mathcal{E}_{m,s}  \big) +W_{\mathrm{memb,bend}}\big(  \mathcal{E}_{m,s} ,\,  \mathcal{K}_{e,s} \big)   +
W_{\mathrm{bend,curv}}\big(  \mathcal{K}_{e,s}    \big)
\Big] \,{\rm det}(\nabla y_0|n_0)       \,\mathrm d a - \overline{\Pi}(m,\overline{Q}_{e,s})\,,
\end{equation}
where the  membrane part $\,W_{\mathrm{memb}}\big(  \mathcal{E}_{m,s} \big) \,$, the membrane--bending part $\,W_{\mathrm{memb,bend}}\big(  \mathcal{E}_{m,s} ,\,  \mathcal{K}_{e,s} \big) \,$ and the bending--curvature part $\,W_{\mathrm{bend,curv}}\big(  \mathcal{K}_{e,s}    \big)\,$ of the shell energy density are given by
\begin{align}\label{e90}
W_{\mathrm{memb}}\big(  \mathcal{E}_{m,s} \big)\,=\,&  \Big(h+{\rm K}\,\dfrac{h^3}{12}\Big)\,
W_{\mathrm{shell}}\big(    \mathcal{E}_{m,s} \big), \notag\\    
W_{\mathrm{memb,bend}}\big(  \mathcal{E}_{m,s} ,\,  \mathcal{K}_{e,s} \big)\,=\,&    \Big(\dfrac{h^3}{12}\,-{\rm K}\,\dfrac{h^5}{80}\Big)\,
W_{\mathrm{shell}}  \big(   \mathcal{E}_{m,s} \, {\rm B}_{y_0} +   {\rm C}_{y_0} \mathcal{K}_{e,s} \big)  \\&
-\dfrac{h^3}{3} \mathrm{ H}\,\mathcal{W}_{\mathrm{shell}}  \big(  \mathcal{E}_{m,s} ,
\mathcal{E}_{m,s}{\rm B}_{y_0}+{\rm C}_{y_0}\, \mathcal{K}_{e,s} \big)+
\dfrac{h^3}{6}\, \mathcal{W}_{\mathrm{shell}}  \big(  \mathcal{E}_{m,s} ,
( \mathcal{E}_{m,s}{\rm B}_{y_0}+{\rm C}_{y_0}\, \mathcal{K}_{e,s}){\rm B}_{y_0} \big) \notag\\&+ \,\dfrac{h^5}{80} \,
W_{\mathrm{mp}} \big((  \mathcal{E}_{m,s} \, {\rm B}_{y_0} +  {\rm C}_{y_0} \mathcal{K}_{e,s} )   {\rm B}_{y_0} \,\big),  \notag \\
W_{\mathrm{bend,curv}}\big(  \mathcal{K}_{e,s}    \big) \,=\,&  \Big(h-{\rm K}\,\dfrac{h^3}{12}\Big)\,
W_{\mathrm{curv}}\big(  \mathcal{K}_{e,s} \big)    +  \Big(\dfrac{h^3}{12}\,-{\rm K}\,\dfrac{h^5}{80}\Big)\,
W_{\mathrm{curv}}\big(  \mathcal{K}_{e,s}   {\rm B}_{y_0} \,  \big)  + \,\dfrac{h^5}{80} \,
W_{\mathrm{curv}}\big(  \mathcal{K}_{e,s}   {\rm B}_{y_0}^2  \big)\notag
\end{align}
and
\begin{align}
\mathcal{E}_{m,s} & \,=\,\overline{Q}_{e,s}^T \widetilde{F}_{m}-\id_3,\qquad \widetilde{F}_{m}\,=\, (\nabla  m|\overline{Q}_{e,s}\nabla_x\Theta(0)\, e_3)[\nabla_x \Theta(0)]^{-1} , \notag\vspace{3.5mm}\\
\mathcal{K}_{e,s} & \,=\,  (\mathrm{axl}(\overline{Q}_{e,s}^T\,\partial_{x_1} \overline{Q}_{e,s})\,|\, \mathrm{axl}(\overline{Q}_{e,s}^T\,\partial_{x_2} \overline{Q}_{e,s})\,|0)[\nabla_x \Theta(0)]^{-1},\notag \\
\Theta(x_1,x_2,x_3) &\,=\,y_0(x_1,x_2)+x_3 \, n_0(x_1,x_2),  \qquad \nabla_x \Theta(0)\,=\,(\nabla y_0|n_0), \qquad n_0\,=\,\left(\nabla_x \Theta\big(0\big)\right)\, e_3 \notag \\
{\rm B}_{y_0}&\,=\,-(\nabla n_0|0) \,[\nabla_x \Theta(0)]^{-1},\qquad
{\rm C}_{y_0}\,=\,\det	(\nabla_x \Theta(0))\,[\nabla_x \Theta(0)]^{-T}\begin{footnotesize}\begin{pmatrix}
0&1&0 \\
-1&0&0 \\
0&0&0
\end{pmatrix}\end{footnotesize}\,  [	\nabla_x \Theta(0)]^{-1}\in \mathfrak{so}(3), \notag\\
{\rm K} &\,=\,{\rm det}{({\rm L}_{y_0})}\, ,\qquad 
2\,{\rm H}\, \,=\,{\rm tr}({{\rm L}_{y_0}}),\qquad {\rm L}_{y_0} \,=\,-([{\nabla  y_0}]^T\,{\nabla  y_0})^{-1} ([{\nabla  n_0}]^T\,{\nabla  y_0}),
 \notag\\
W_{\mathrm{shell}}(  X) & \,=\,   \mu\,\lVert  \mathrm{sym} \,X\rVert^2 +  \mu_{\rm c}\,\lVert \mathrm{skew} \,X\rVert^2 +\,\dfrac{\lambda\,\mu}{\lambda+2\mu}\,\big[ \mathrm{tr}(X)\big]^2, \\
\mathcal{W}_{\mathrm{shell}}(  X,  Y)& \,=\,   \mu\, \bigl\langle  \mathrm{sym} \,X, \mathrm{sym}\,   \,Y\bigr\rangle  +  \mu_{\rm c} \bigl\langle  \,\mathrm{skew} \,X, \mathrm{skew}\,   \,Y\bigr\rangle  +\,\dfrac{\lambda\,\mu}{\lambda+2\mu}\,\mathrm{tr}  (X)\,\mathrm{tr}  (Y),  \notag 
\end{align}
\begin{align}\hspace*{-3.3cm}
W_{\mathrm{mp}}(  X)&\,=\, \mu\,\lVert \mathrm{sym} \,X\rVert^2+  \mu_{\rm c}\,\lVert \mathrm{skew} \,X\rVert^2 +\,\dfrac{\lambda}{2}\,\big[ \mathrm{tr}(X)\big]^2,\notag \\
\hspace*{-3.3cm}W_{\mathrm{curv}}(  X )&\,=\,\mu\, {L}_{\rm c}^2 \left( b_1\,\lVert  \dev\,\text{sym} \,X\rVert^2+b_2\,\lVert \text{skew}\,X\rVert^2+\,4\,b_3\,
\big[ \mathrm{tr}(X)\big]^2\right) .\notag
\end{align}
In this formulation,  all the constitutive coefficients  are deduced from the three-dimensional formulation, without using any a posteriori fitting of some two-dimensional constitutive coefficients.

The potential of applied external loads $ \overline{\Pi}(m,\overline{Q}_{e,s}) $ appearing in \eqref{e89} is expressed by the relations \eqref{e4o}, \eqref{e5o}.

We consider the following boundary conditions for the midsurface deformation $m$ and rotation field $\overline{R}_s$ on the Dirichlet part of the lateral boundary $\gamma_0\subset \partial \omega$:
\begin{align}
m\mid_{\gamma_0}&\,=\,m_0,\ \  \ \ \ \ \ \ \text{simply supported (fixed, welded)} \qquad \qquad 
 \overline{R}_s\mid_{\gamma_0}\,=\,\hat{R},\ \  \ \ \ \ \ \ (\text{clamped)}.\notag
\end{align}

It is possible to use 	
	  the referential fundamental forms $ {\rm I}_{y_0} $, $ {\rm II}_{y_0}$ and $ {\rm L}_{y_0} $  instead of  the matrices $ {\rm A}_{y_0}$, $ {\rm B}_{y_0}$ and  $ {\rm C}_{y_0}$, and to rewrite all the arguments of the energy terms as
	\begin{align}\label{eq5}
	\mathcal{E}_{m,s}=&\,  [\nabla_x \Theta(0)]^{-T}
	\begin{footnotesize}\left( \begin{array}{c|c}
	(\overline{Q}_{e,s} \nabla y_0)^{T} \nabla m- {\rm I}_{y_0} & 0 \vspace{4pt}\\
	(\overline{Q}_{e,s}  n_0)^{T} \nabla m & 0
	\end{array} \right)\end{footnotesize} [\nabla_x \Theta(0)]^{-1} \notag\\=&\,
	[\nabla_x \Theta(0)]^{-T}
	\begin{footnotesize}\left( \begin{array}{c|c}
	\mathcal{G} & 0 \vspace{4pt}\\
	\mathcal{T}  & 0
	\end{array} \right)\end{footnotesize} [\nabla_x \Theta(0)]^{-1}
	\vspace{6pt}\notag\\
	\mathrm{C}_{y_0} \mathcal{K}_{e,s} = &\, [\nabla_x \Theta(0)]^{-T}
	\begin{footnotesize}\left( \begin{array}{c|c}
	(\overline{Q}_{e,s} \nabla y_0)^{T} \nabla (\overline{Q}_{e,s} n_0)+ {\rm II}_{y_0} & 0 \vspace{4pt}\notag\\
	0 & 0
	\end{array} \right)\end{footnotesize} [\nabla_x \Theta(0)]^{-1}\notag\\=&\, [\nabla_x \Theta(0)]^{-T}
	\begin{footnotesize}\left( \begin{array}{c|c}
	-\mathcal{R} & 0 \vspace{4pt}\notag\\
	0 & 0
	\end{array} \right)\end{footnotesize} [\nabla_x \Theta(0)]^{-1},\notag\\
	\mathcal{E}_{m,s} {\rm B}_{y_0} = &\,  [\nabla_x \Theta(0)]^{-T}
	\begin{footnotesize}\left( \begin{array}{c|c}
	\mathcal{G} \,{\rm L}_{y_0} & 0 \vspace{4pt}\notag\\
	\mathcal{T}\,{\rm L}_{y_0} & 0
	\end{array} \right)\end{footnotesize} [\nabla_x \Theta(0)]^{-1} ,
	\vspace{10pt}\\
	\mathcal{E}_{m,s} {\rm B}^2_{y_0} = &\, [\nabla_x \Theta(0)]^{-T}
	\begin{footnotesize}\left( \begin{array}{c|c}
	\mathcal{G}\, {\rm L}^2_{y_0} & 0 \vspace{4pt}\\
	\mathcal{T} \,{\rm L}^2_{y_0} & 0
	\end{array} \right)\end{footnotesize} [\nabla_x \Theta(0)]^{-1} ,
	\vspace{10pt}\\
	\mathrm{C}_{y_0} \mathcal{K}_{e,s} {\rm B}_{y_0} = &\, \,- [\nabla_x \Theta(0)]^{-T}
	\begin{footnotesize}\left( 
	\mathcal{R}\, {\rm L}_{y_0}\right)^\flat\end{footnotesize} [\nabla_x \Theta(0)]^{-1},
	\vspace{10pt}\notag\\\notag
	\mathrm{C}_{y_0} \mathcal{K}_{e,s} {\rm B}^2_{y_0} = &\, - [\nabla_x \Theta(0)]^{-T}
	\begin{footnotesize}\left( 
	\mathcal{R}\, {\rm L}^2_{y_0}\right)^\flat \end{footnotesize}[\nabla_x \Theta(0)]^{-1},
	\vspace{10pt}\notag\\
	\mathcal{E}_{m,s} {\rm B}_{y_0}  + \mathrm{C}_{y_0} \mathcal{K}_{e,s} 
	= &\, [\nabla_x \Theta(0)]^{-T}
	\begin{footnotesize}\left( \begin{array}{c|c}
	\mathcal{G} \,{\rm L}_{y_0}- \mathcal{R} & 0 \vspace{4pt}\\
	\mathcal{T} \,{\rm L}_{y_0} & 0
	\end{array} \right)\end{footnotesize} [\nabla_x \Theta(0)]^{-1}\notag,\\
	\mathcal{E}_{m,s} {\rm B}_{y_0}^2  + \mathrm{C}_{y_0} \mathcal{K}_{e,s} {\rm B}_{y_0}
	= &\, [\nabla_x \Theta(0)]^{-T}
	\begin{footnotesize}\left( \begin{array}{c|c}
	\mathcal{G} \,{\rm L}_{y_0}^2- \mathcal{R}{\rm L}_{y_0} & 0 \vspace{4pt}\\
	\mathcal{T} \,{\rm L}_{y_0}^2 & 0
	\end{array} \right)\end{footnotesize} [\nabla_x \Theta(0)]^{-1}\notag
	,
	\end{align}
	where
	\begin{align}\label{eq4}
	\mathcal{G} :=&\, (\overline{Q}_{e,s} \nabla y_0)^{T} \nabla m- {\rm I}_{y_0}\not\in {\rm Sym}(2)\qquad\qquad\quad\ \ \ \,\textrm{\it the change of metric tensor (in-plane deformation)},\notag
	\\
	\mathcal{T}:=& \, (\overline{Q}_{e,s}  n_0)^T (\nabla m) \qquad \qquad \qquad \quad \quad\qquad \ \qquad \ \  \textrm{\it the transverse shear deformation (row) vector\footnotemark},
	\\
	\mathcal{R} :=& \, -(\overline{Q}_{e,s} \nabla y_0)^{T} \nabla (\overline{Q}_{e,s} n_0)- {\rm II}_{y_0}\not\in {\rm Sym}(2)
	\quad  \ \ \, \textrm{\it the bending  strain tensor}.
	\notag\end{align}
\footnotetext{The vector $d_3=\overline{Q}_{e,s}\,n_0$ represents the classical director, which does not have to be orthogonal to the deformed midsurface. If $\mathcal{T}=(0,0)$, then $\overline{Q}_{e,s}\,n_0$ is orthogonal to the deformed midsurface. An alternative form  of the transverse shear deformation row vector is $\mathcal{T}= \, \left(\bigl\langle\overline{Q}_{e,s}  n_0, \partial_{x_1} m\bigr\rangle,\bigl\langle\overline{Q}_{e,s}  n_0, \partial_{x_2} m\bigr\rangle\right)$.}
Regarding the arguments of the  bending-curvature energy density $ W_{\mathrm{bend,curv}} $, we can express the tensor $ \mathcal{K}_{e,s}$  in terms of the tensor $ \mathrm{C}_{y_0}\, \mathcal{K}_{e,s}$ and the vector $ \mathcal{K}_{e,s}^T n_0$, according to Proposition \ref{propAB} and  to the decomposition 
	\begin{align}  \label{descK}
	\mathcal{K}_{e,s} = {\rm A}_{y_0} \, \mathcal{K}_{e,s} +(0|0|n_0) \,(0|0|n_0)^T \, \mathcal{K}_{e,s}= \mathrm{C}_{y_0}( - \mathrm{C}_{y_0} \mathcal{K}_{e,s}) +(0|0|n_0) \,(0|0|\mathcal{K}_{e,s}^T\,n_0)^T\, \, .
	\end{align} 
We have already seen that   ${\rm C}_{y_0} \mathcal{K}_{e,s}$ from the above decomposition can be expressed in terms of the {\it bending strain} tensor $ \mathcal{R} $, see \eqref{eq5},
while
	the remaining vector $\mathcal{K}_{e,s}^T\,n_0$ from \eqref{descK} is completely characterized by the row vector
	\begin{equation}
	\label{e5d}
	\mathcal{N} :=  n_0^T \big(\mbox{axl}(\overline{Q}_{e,s}^T\partial_{x_1}\overline{Q}_{e,s})\,|\, \mbox{axl}(\overline{Q}_{e,s}^T\partial_{x_2}\overline{Q}_{e,s}) \big) ,
	\end{equation}
	which is   called the row vector of {\it drilling bendings}. 

\begin{remark}Summarizing, the present shell model is derived:
\begin{itemize}
	\item  under the assumption that
$ {h}\,|{\kappa_1}|<\frac{1}{2},\  {h}\,|{\kappa_2}|<\frac{1}{2}$
\item considering an approximation of the elastic rotation $\overline{Q}_{e}:\Omega_h\rightarrow \text{\rm{SO}}(3)$ 
\begin{equation}
\overline{Q}_{e}(x_1,x_2,x_3)\cong\overline{Q}_{e,s}(x_1,x_2)\,=\,\overline{R}_s(x_1,x_2)\,Q_0^T(x_1,x_2,0)\, ;
\end{equation}
\item  choosing an	 \textit{8-parameter quadratic ansatz} in the thickness direction for the reconstructed total deformation $\varphi_s:\Omega_h\subset \mathbb{R}^3\rightarrow \mathbb{R}^3$ of the shell-like structure
\begin{align}
\varphi_s(x_1,x_2,x_3)\,=\,&m(x_1,x_2)+\bigg(x_3\varrho_m(x_1,x_2)+\dd\frac{x_3^2}{2}\varrho_b(x_1,x_2)\bigg)\overline{Q}_{e,s}(x_1,x_2)\nabla_x\Theta(x_1,x_2,0)\, e_3;
\end{align}
\item  taking the exact form of $\varrho_m$ and considering a suitable approximation for  $\varrho_b$ (coming from a generalized plane stress condition)
\begin{align}
\varrho_m\,=\,&1-\frac{\lambda}{\lambda+2\mu}[ \bigl\langle  \overline{Q}_{e,s}^T(\nabla m|0)[\nabla_x\Theta(0)]^{-1},\id_3\bigr\rangle -2]\,=\,:\varrho_m^e, \notag\\
\varrho_b
\,=\,& \frac{\lambda}{\lambda+2\,\mu}\, \bigl\langle \overline{Q}_{e,s} ^T(\nabla  m|0)\,[\nabla_x\Theta(0)]^{-1}\,(\nabla n_0|0)\, [\nabla_x\Theta(0)]^{-1},\id_3\bigr\rangle  \\&-\frac{\lambda}{\lambda+2\,\mu}\, \bigl\langle  \overline{Q}_{e,s} ^T(\nabla (\,\overline{Q}_{e,s} \nabla_x\Theta(0)\, e_3)|0)\,[\nabla_x\Theta(0)]^{-1},\id_3\bigr\rangle  \notag\\
&+\frac{\lambda^2}{(\lambda+2\,\mu)^2}\,\Big[ \bigl\langle  \overline{Q}_{e,s} ^T(\nabla (\,\overline{Q}_{e,s} \nabla_x\Theta(0)\, e_3)|0)\,[\nabla_x\Theta(0)]^{-1},\id_3\bigr\rangle \Big]\Big[ \bigl\langle  \overline{Q}_{e,s}^T(\nabla m|0)[\nabla_x\Theta(0)]^{-1},\id_3\bigr\rangle -2\Big]\notag \\
\cong&\, \varrho_b^e:\,=\,\frac{\lambda}{\lambda+2\,\mu}\, \bigl\langle \overline{Q}_{e,s} ^T(\nabla  m|0)\,[\nabla_x\Theta(0)]^{-1}\,(\nabla n_0|0)\, [\nabla_x\Theta(0)]^{-1},\id_3\bigr\rangle  \notag\\&-\frac{\lambda}{\lambda+2\,\mu}\, \bigl\langle  \overline{Q}_{e,s} ^T(\nabla (\,\overline{Q}_{e,s} \nabla_x\Theta(0)\, e_3)|0)\,[\nabla_x\Theta(0)]^{-1},\id_3\bigr\rangle ;\notag
\end{align}
\item choosing a further approximation of the deformation gradient (by neglecting space derivatives of $\varrho_m^e$ and $\varrho_b^e$, respectively)
\begin{align}
F_s&\,=\,\nabla_x\varphi_s(x_1,x_2,x_3)\,=\,(\nabla  m|\, \varrho_m\,\overline{Q}_{e,s}(x_1,x_2)\nabla_x\Theta(x_1,x_2,0)\, e_3) \notag\\&\hspace{3.3cm}+x_3\, (\nabla \left[\varrho_m\,\overline{Q}_{e,s}(x_1,x_2)\nabla_x\Theta(x_1,x_2,0)\, e_3\right]|\varrho_b\,
\overline{Q}_{e,s}(x_1,x_2)\nabla_x\Theta(x_1,x_2,0)\, e_3)\nonumber \\
&\hspace{3.3cm}+\frac{x_3^2}{2}(\nabla \left[\varrho_b\,\overline{Q}_{e,s}(x_1,x_2)\nabla_x\Theta(x_1,x_2,0)\, e_3\right]|0) \\
&\cong\widetilde{F}_s:\,=\,(\nabla  m|\, \varrho_m^e\,\overline{Q}_{e,s}(x_1,x_2)\nabla_x\Theta(x_1,x_2,0)\, e_3)
 \notag\\\nonumber  &\qquad\qquad  +x_3 (\nabla \left[\,\overline{Q}_{e,s}(x_1,x_2)\nabla_x\Theta(x_1,x_2,0)\, e_3\right]|\varrho_b^e\,
\overline{Q}_{e,s}(x_1,x_2)\nabla_x\Theta(x_1,x_2,0)\, e_3),\notag
\end{align}
and therefore, the following approximation of the reconstructed gradient
\begin{align}
F_{s,\xi}&\,=\,\nabla_x\varphi_s(x_1,x_2,x_3)[\nabla_x \Theta(x_1,x_2,x_3)]^{-1}  \\\nonumber
&\cong \widetilde{F}_{e,s}:\,=\,(\nabla  m|\, \varrho_m^e\,\overline{Q}_{e,s}(x_1,x_2)\nabla_x\Theta(x_1,x_2,0)\, e_3)[\nabla_x \Theta(x_1,x_2,x_3)]^{-1}
 \\\nonumber  &\qquad \qquad \quad +x_3 (\nabla \left[\,\overline{Q}_{e,s}(x_1,x_2)\nabla_x\Theta(x_1,x_2,0)\, e_3\right]|\varrho_b^e\,
\overline{Q}_{e,s}(x_1,x_2)\nabla_x\Theta(x_1,x_2,0)\, e_3)[\nabla_x \Theta(x_1,x_2,x_3)]^{-1}
.
\end{align}
Moreover,  we have used the full expressions of $	[	\nabla_x \Theta(x_3)]^{-1}$ and $	{\rm det}(\nabla_x \Theta(x_3))$ 
\begin{align}
[	\nabla_x \Theta(x_3)]^{-1}&\,= \,\frac{1}{1-2\,{\rm H}\, x_3+{\rm K}\, x_3^2}\left[\id_3+x_3({\rm L}_{y_0}^\flat-2\,{\rm H}\, \id_3)+x_3^2\, {\rm K}\begin{footnotesize}\begin{pmatrix}
0&0&0 \\
0&0&0 \\
0&0&1
\end{pmatrix}\end{footnotesize}
\right] [	\nabla_x \Theta(0)]^{-1}, \notag\\
{\rm det}(\nabla_x \Theta(x_3))&\,=\,{\rm det}{(\nabla_x \Theta(0))}\,\Big[1-2\,x_3\,{\rm H}\,+x_3^2 \, {\rm K}\Big]\, .
\end{align}
\item  neglecting the terms of order $O(h^7)$ in the final form of the energy.
\end{itemize}
\end{remark}

After a shell model is proposed, there is a basic requirement: the 2D-shell model must be invariant w.r.t. a reparametrization of the midsurface coordinates. Here,  the total elastically stored
energy 
$$
W_{\mathrm{memb}}\big(  \mathcal{E}_{m,s}  \big)+W_{\mathrm{memb,bend}}\big(  \mathcal{E}_{m,s} ,\,  \mathcal{K}_{e,s} \big)    +
W_{\mathrm{bend,curv}}\big(  \mathcal{K}_{e,s}    \big)
$$
depends on the midsurface deformation gradient $\nabla m$ and microrotations $\overline{Q}_{e,s}$ together with
their space
derivatives only through the frame-indifferent tensors $\mathcal{E}_{m,s} , \mathcal{K}_{e,s} , {\rm B}_{y_0}$ and ${\rm C}_{y_0}$, which are invariant to the reparametrization.

\subsection{Consistency with the Cosserat plate model}

In the case of Cosserat plates we have $\Theta(x_1,x_2,x_3)\,=\,(x_1,x_2,x_3)$ and
\begin{align}
\nabla_x \Theta(x_3)&\,=\,\id_3, \qquad y_0(x_1,x_2)\,=\,(x_1,x_2)\,=\,:{\rm id}(x_1,x_2), \qquad Q_0\,=\,\id_3,\qquad n_0\,=\,e_3, \qquad d_i^0\,=\,e_i,\\
\qquad {\rm B}_{\rm id}&\,=\,0_3, \qquad {\rm C}_{\rm id}\,=\,\begin{footnotesize}\begin{footnotesize}\begin{pmatrix}
0&1&0 \\
-1&0&0 \\
0&0&0
\end{pmatrix}\end{footnotesize}\end{footnotesize}\in \mathfrak{so}(3), \qquad 
{\rm L}_{\rm id} \,=\,0_2,\qquad {\rm K} \,=\,0\, ,\qquad 
{\rm H}\, \,=\,0.\notag
\end{align}
Therefore, for the Cosserat plate model the minimization problem reads: find the
deformation of the midsurface $m:\omega
\,{\to}\,
\mathbb{R}^3$ and the microrotation of the shell
$\overline{Q}_{e,s}:\omega
\,{\to}\,
\textrm{SO}(3)$ solving on $\omega
{\,
\subset\mathbb{R}^2}
$:
\begin{equation}
I\,=\, \int_{\omega}    \, \Big[  \,
W_{\mathrm{memb}}\big(  \mathcal{E}_{m,s}  \big)+W_{\mathrm{memb,bend}}\big(  \mathcal{E}_{m,s} ,\,  \mathcal{K}_{e,s} \big)    +
W_{\mathrm{bend,curv}}\big(  \mathcal{K}_{e,s}    \big)
\Big]        \,\mathrm d a\  
{\to}
 \min.\;\text{w.r.t. $(m,\overline{Q}_{e,s})$}
\end{equation}
where the  membrane part $\,W_{\mathrm{memb}}\big(  \mathcal{E}_{m,s} \big)$, the membrane-bending part $W_{\mathrm{memb,bend}}\big(  \mathcal{E}_{m,s} ,\,  \mathcal{K}_{e,s} \big) $ and the bending--curvature part $\,W_{\mathrm{bend,curv}}\big(  \mathcal{K}_{e,s}    \big)\,$ of the shell energy density are given by
\begin{align} 
W_{\mathrm{memb}}\big(  \mathcal{E}_{m,s} \big)\,=\,& \, h\,
W_{\mathrm{shell}}\big(    \mathcal{E}_{m,s} \big),\qquad  W_{\mathrm{memb,bend}}\big(  \mathcal{E}_{m,s} ,\,  \mathcal{K}_{e,s} \big)\,=\,   \,   \dfrac{h^3}{12}\,
W_{\mathrm{shell}}  \big(      {\rm C}_{\rm id} \mathcal{K}_{e,s} \big), \vspace{1.5mm}\notag\\
W_{\mathrm{bend,curv}}\big(  \mathcal{K}_{e,s}    \big) \,=\,&\,  h\,
W_{\mathrm{curv}}\big(  \mathcal{K}_{e,s} \big),
\end{align}
and
\begin{align}
\mathcal{E}_{m,s} & \,=\,\overline{Q}_{e,s}^T \widetilde{F}_{m}-\id_3,\qquad \widetilde{F}_{m}\,=\, (\nabla  m|\overline{Q}_{e,s}\, e_3), \notag \\
\mathcal{K}_{e,s} & \,=\,  (\mathrm{axl}(\overline{Q}_{e,s}^T\,\partial_{x_1} \overline{Q}_{e,s})\,|\, \mathrm{axl}(\overline{Q}_{e,s}^T\,\partial_{x_2} \overline{Q}_{e,s})\,|0),\notag\quad	{\rm C}_{\rm id} \mathcal{K}_{e,s}\,=\,\overline{Q}_{e,s}^T\,(\nabla [\overline{Q}_{e,s}\, e_3]\,|\,0), \\
W_{\mathrm{shell}}(  X) & \,=\,   \mu\,\lVert \, \mathrm{sym} \,X\rVert^2 +  \mu_{\rm c}\,\lVert \mathrm{skew} \,X\rVert^2 +\,\dfrac{\lambda\,\mu}{\lambda+2\mu}\,\big[ \mathrm{tr}   \,X\,\big]^2, \\
W_{\mathrm{curv}}(  X )&\,=\,\mu\,{L}_{\rm c}^2\left( b_1\,\lVert  \dev\,\text{sym} \,X\rVert^2+b_2\,\lVert \text{skew}\,X\rVert^2+\,4\,b_3\,
[\tr(X)]^2\right) .\notag
\end{align}
In view of Lemma
		\ref{LemaEsEe} the following identity is satisfied 
		\begin{align} 
		\overline{Q}_{e,s}^T\,(\nabla [\overline{Q}_{e,s}\, e_3]\,|\,0) \,=\,{\rm C}_{\rm id}\, \mathcal{K}_{e,s}.
		\end{align}
		Hence, ${\rm C}_{\rm id} \mathcal{K}_{e,s}$ coincides with the second order {\it non-symmetric bending tensor} $$\mathfrak{K}_b:\,=\,	\overline{Q}_{e,s}^T\,(\nabla [\overline{Q}_{e,s}\, e_3]\,|\,0)+{\rm B}_{ y_0}\,=\,\overline{Q}_{e,s}^T\,(\nabla [\overline{Q}_{e,s}\, e_3]\,|\,0)$$ considered by Neff in \cite{Neff_plate04_cmt}. 
	 In consequence, the particular case considered in this subsection is the  geometrically nonlinear {C}osserat-shell model including size effects introduced in \cite{Neff_plate04_cmt} and then used in a numerical approach by Sander et al. \cite{Sander-Neff-Birsan-15}, but corresponds to another representation of the curvature  energy from the three-dimensional formulation.

\newpage
For purpose of comparison: 
\begin{remark}
This Cosserat plate model was derived:
\begin{itemize}
\item considering an approximation of the elastic rotation $\overline{Q}_{e}:\Omega_h\rightarrow \text{\rm{SO}}(3)$ 
\begin{equation}
\overline{Q}_{e}(x_1,x_2,x_3)\cong\overline{Q}_{e,s}(x_1,x_2)\, ;
\end{equation}
\item  choosing an	 \textit{8-parameter quadratic ansatz} in the thickness direction for the reconstructed total deformation $\varphi_s:\Omega_h\subset \mathbb{R}^3\rightarrow \mathbb{R}^3$ of the shell-like structure
\begin{align}
\varphi_s(x_1,x_2,x_3)\,=\,&m(x_1,x_2)+\bigg(x_3\varrho_m(x_1,x_2)+\dd\frac{x_3^2}{2}\varrho_b(x_1,x_2)\bigg)\overline{Q}_{e,s}(x_1,x_2)\, e_3\,;
\end{align}
\item  taking the exact form of $\varrho_m$ and considering a suitable approximation for  $\varrho_b$ (coming from a generalized plane stress condition)
\begin{align}
\varrho_m\,=\,&1-\frac{\lambda}{\lambda+2\mu}[ \bigl\langle  \overline{Q}_{e,s}^T(\nabla m|0),\id_3\bigr\rangle -2]\,=\,:\varrho_m^e, \notag\\
\varrho_b
\,=\,&-\frac{\lambda}{\lambda+2\,\mu}\, \bigl\langle  \overline{Q}_{e,s} ^T(\nabla (\,\overline{Q}_{e,s} \, e_3)|0),\id_3\bigr\rangle  \notag\\
&+\frac{\lambda^2}{(\lambda+2\,\mu)^2}\,\Big[ \bigl\langle  \overline{Q}_{e,s} ^T(\nabla (\,\overline{Q}_{e,s} \, e_3)|0),\id_3\bigr\rangle \Big]\Big[ \bigl\langle  \overline{Q}_{e,s}^T(\nabla m|0),\id_3\bigr\rangle -2\Big]\notag \\
\cong&\, \varrho_b^e:\,=\,-\frac{\lambda}{\lambda+2\,\mu}\, \bigl\langle  \overline{Q}_{e,s} ^T(\nabla (\,\overline{Q}_{e,s}\, e_3)|0),\id_3\bigr\rangle ;\notag
\end{align}
\item choosing a further approximation of the deformation gradient (by neglecting space derivatives of $\varrho_m^e$ and $\varrho_b^e$, respectively)
\begin{align}
F_s\,=\,&(\nabla  m|\, \varrho_m\,\overline{Q}_{e,s}(x_1,x_2)\, e_3)+x_3\, (\nabla \left[\varrho_m\,\overline{Q}_{e,s}(x_1,x_2)\, e_3\right]|\varrho_b\,
\overline{Q}_{e,s}(x_1,x_2)\, e_3)\nonumber \\
&\hspace{2cm}+\frac{x_3^2}{2}(\nabla \left[\varrho_b\,\overline{Q}_{e,s}(x_1,x_2)\, e_3\right]|0) \\
\cong&\,\widetilde{F}_s:\,=\,(\nabla  m|\, \varrho_m^e\,\overline{Q}_{e,s}(x_1,x_2)\, e_3)
 +x_3 (\nabla \left[\,\overline{Q}_{e,s}(x_1,x_2)\, e_3\right]|\varrho_b^e\,
\overline{Q}_{e,s}(x_1,x_2)\, e_3);\notag
\end{align}
\item  neglecting the terms of order $O(h^5)$ in the final form of the energy.
\end{itemize}
\end{remark}

\section{A comparison  with  the general 6-parameter shell model}\setcounter{equation}{0}\label{comparisonP}

In this section we give an overview of the quantities appearing in the general 6-parameter shell model presented in \cite{Eremeyev06}. Eremeyev and Pietraskiewicz  have considered the classical multiplicative decomposition  ${F}\,=\,F^eF^0\,$ of  the (reconstructed) total deformation gradient into elastic and initial (``plastic'') parts \cite{Neff_Cosserat_plasticity05,HutterSFB02}, i.e.,  $F^e$ represents the (reconstructed) elastic shell deformation gradient, while  $F^0\,=\,P$ is the initial deformation gradient.   In the general 6-parameter shell model, the following form for the elastic strain tensor ${E}^e$ is used (written in matrix notation)
\begin{equation}\label{15,1}
\begin{array}{l}
{E}^e \,=\, \overline{Q}_e^{T}(\,\partial_{x_1} m\,|\,  \partial_{x_2} m\,|\, \overline{Q}_e n_0\,) \,P^{-1} - \id_3,\quad\text{or}  \quad 
{E}^e \,=\, \overline{U} ^{\,e}-\id_3 \,=\,  \overline{Q}_e^{T}\,F^e - \id_3  \,=\, \overline{Q}_e^{T} {F}  \,P^{-1}-  \id_3,
\end{array}
\end{equation}
with
\begin{align}\label{15,2}
 {F}  \,=&\,F^e \,F^0,\qquad {F} \,\,=\, (\,\partial_{x_1} m\,|\,  \partial_{x_2} m\,|\, \overline{Q}_e n_0\,) \,=\, (\,\nabla m\,|\, \overline{Q}_e n_0\,) , \notag \\
 {F}^e\,=&\,\, (\,\partial_{x_1} m\,|\,  \partial_{x_2} m\,|\, \overline{Q}_e n_0\,) \,P^{-1}, \qquad 
{F}^0\,=\,P\,=\,\,(\,\partial_{x_1} y_{0}\,|\, \partial_{x_2} y_{0}\,|\, n_0\,)\,=\, (\,\nabla y_0\,|\,   n_0\,),\\
\overline U ^{\,e}\,=&\,\,\, \overline{Q}_e^{T}(\,\partial_{x_1} m\,|\,  \partial_{x_2} m\,|\,  \overline{Q}^{e} n_0\,) \,P^{-1} \,=\, \overline{Q}_e^{T}(\,\nabla m\,|\, \overline{Q}_e n_0\,) \,P^{-1}.\notag
\end{align}
Since, $\nabla_x \Theta (0)=(\,\nabla y_0\,|\,   n_0\,)=P$ and $n_0=Q_0\, e_3\,=\,[\nabla_x \Theta (0)]\, e_3$, we  remark that 
$
E^e=\mathcal{E}_{m,s}.
$
Hence, in the general 6-parameter shell model  the same elastic shell strain tensor $\mathcal{E}_{m,s}$ as in our shell model is used.

Regarding the bending curvature tensor, 
in the general 6-parameter shell model the tensor $\mathcal{K}$ is the total bending--curvature tensor, while $\mathcal{K}^0$ is the initial  bending-curvature (or structure curvature tensor of $\Omega_h$). The matrix $\mathcal{K}^e\,=\,\mathcal{K}-\mathcal{K}_0 $ is given by
\begin{equation}\label{17}
\mathcal{K}^e \,=\, ( \,\text{axl}(\overline{Q}_e^{T}\partial_{x_1}  \overline{Q}_e) \,|\, \text{axl}(\overline{Q}_e^{T}\partial_{x_2}  \overline{Q}_e) \,|\,0\,) \,P^{-1}\, ,
\end{equation}
or
\begin{align}\label{18}
\notag\mathcal{K}^e \,&=\,Q_0\,L\,P^{-1}\,=\,\mathcal{K}-\mathcal{K}^0\qquad\text{with}   \qquad R=\overline{Q}_{e}Q_0, \\ L&=\,  ( \,\text{axl}(R^T\partial_{x_1} R)-\text{axl}(Q_0^T\partial_{x_1} Q_{0})  \, | \, \text{axl}(R^T\partial_{x_2} R) -\text{axl}(Q_{0}^T\partial_{x_2} Q_{0}) \, |\, \,0 \, )_{3\times 3}\, ,  \\
\mathcal{K}\,&=\, Q_0 ( \,\text{axl}(R^T\partial_{x_1} R)   \, | \, \text{axl}(R^T\partial_{x_2} R)  \, |\, \,0 \, )P^{-1},   
\qquad \mathcal{K}^0\,=\, Q_0 ( \, \text{axl}(Q_{0}^T\partial_{x_1} Q_{0})  \, | \,  \text{axl}(Q_{0}^T\partial_{x_2} Q_{0}) \, |\, \,0 \, )P^{-1}.\notag
\end{align}

Using again that  $P=\nabla_x \Theta (0)$, we have that the total bending-curvature tensor from Pietraskiewicz and his collaborators coincides with the elastic shell-bending-curvature tensor from our model, i.e.,
$
\mathcal{K}^e=\mathcal{K}_{e,s}.
$
Therefore, we conclude that:

\begin{remark}
	\begin{itemize}
		\item[]
		\item[1)] A direct comparison with our model shows  that the strain tensors  $E^e$ and $\mathcal{K}^e$  from the general 6-parameter shell model  corresponds to the tensor $\mathcal{E}_{m,s} $ and $\mathcal{K}_{e,s} $, respectively,  from our model;
		\item[2)] While the general 6-parameter shell model  is not deduced from a three dimensional energy, the strain  tensors  $E^e$ and $\mathcal{K}^e$ are directly introduced in the model as work-conjugate strain measures \cite{Eremeyev1}, without any explanation about how a reconstructed (three-dimensional) ansatz,  which minimizes (approximatively) a three-dimensional variational problem, leads to the form of the constitutive tensors  $E^e$ and $\mathcal{K}^e$. {However, discussions of these two-dimensional strain measures and
			their three-dimensional counterparts one can find in the book by Libai and Simmonds \cite{Libai98} as well as in the book by Chroscielewski, Makowski and Pietraszkiewicz \cite{Pietraszkiewicz-book04}. }
		\item[3)] Contrary to the general 6-parameter shell model, in our description, the roles and the deduction of the strain tensors $\mathcal{E}_{m,s} $ and $\mathcal{K}_{e,s}$ is  explained by the dimensional reduction method:
		\begin{itemize}
			\item the elastic shell bending-curvature tensor $\mathcal{K}_{e,s}$ is appearing in the model from the form of the three-dimensional curvature energy (see \eqref{minprob}), after using Nye's formula \eqref{Nye1}, the chain rule (see \eqref{qiese}) and using the ansatz \eqref{ansatzR} (see \eqref{gammaDumi}).
			\item the elastic shell strain tensor $\mathcal{E}_{m,s} $ is appearing  in the modelling process  after the ansatz for the (reconstructed) deformation gradient  \eqref{red2} is proposed and it is suggested by the expressions of $\varrho_m^e$ and $\varrho_b^e$ (see \eqref{final_rho-Dumi}), in order to satisfy (approximatively) the Neumann plane-stress boundary conditions in the reference configuration.
			\item while $\varrho_m^e$ depends only on $\mathcal{E}_{m,s} $,  $\varrho_b^e$ depends on both tensors $\mathcal{E}_{m,s} $ and $\mathcal{K}_{e,s}$. This means that the symmetric thickness stretch about the midsurface is influenced only by the elastic shell strain tensor $\mathcal{E}_{m,s} $, while  the asymmetric thickness stretch  about the midsurface is influenced by both the elastic shell strain tensor $\mathcal{E}_{m,s} $ and the elastic shell bending-curvature tensor $\mathcal{K}_{e,s}$.
		\end{itemize}
		\item[4)] The complete description and role of the involved tensors in the dimensional deduction process is the effect of three factors:
		\begin{itemize}
			\item in the deduction of our model we start with a three-dimensional variational problem for an (three-dimensional) elastic body. A shell is actually a three-dimensional body.
			\item we use the matrix formulation in the entire modelling process.
			\item we propose a specific isotropic form for the three-dimensional curvature energy in the parent three-dimensional variational problem.
		\end{itemize}
	\end{itemize}
\end{remark}

In the resultant 6-parameter theory of shells, the strain energy density for isotropic shells has been presented in various forms. The simplest expression $W_{\rm P}(\mathcal{E}_{m,s} ,\mathcal{K}_{e,s})$ has been proposed in the papers \cite{Pietraszkiewicz-book04,Pietraszkiewicz10} in the form
\begin{align}\label{56}
2\,W_{\rm P}(\mathcal{E}_{m,s} ,\mathcal{K}_{e,s})=& \,C\big[\,\nu \,(\mathrm{tr}\, \mathcal{E}_{m,s}^{\parallel})^2 +(1-\nu)\, \mathrm{tr}((\mathcal{E}_{m,s}^{\parallel} )^T \mathcal{E}_{m,s}^{\parallel} )\big]  + \alpha_{s\,}C(1-\nu) \, \lVert  \mathcal{E}_{m,s}^T  {n}_0\rVert^2\notag \\
 &+\,D \big[\,\nu\,(\mathrm{tr}\, \mathcal{K}_{e,s}^{\parallel})^2 + (1-\nu)\, \mathrm{tr}((\mathcal{K}_{e,s}^{\parallel} )^T \mathcal{K}_{e,s}^{\parallel} )\big]  + \alpha_{t\,}D(1-\nu) \,     \lVert \mathcal{K}_{e,s}^T  {n}_0\rVert^2,
\end{align}
where the decompositions  of $\mathcal{E}_{m,s}$ and $\mathcal{K}_{e,s}$ into two orthogonal directions (in the tangential plane and in the normal direction)\footnote{Here we have used that $ \bigl\langle  {\rm A}_{y_0}, \id_3- {\rm A}_{y_0}\bigr\rangle = \bigl\langle  {\rm A}_{y_0}, \id_3\bigr\rangle  -  \bigl\langle  {\rm A}_{y_0}{\rm A}_{y_0}^T,\id_3\bigr\rangle = \bigl\langle  {\rm A}_{y_0}, \id_3\bigr\rangle  -  \bigl\langle  {\rm A}_{y_0}^2,\id_3\bigr\rangle = \bigl\langle  {\rm A}_{y_0}, \id_3\bigr\rangle  -  \bigl\langle  {\rm A}_{y_0},\id_3\bigr\rangle =0$. } are considered
\begin{align}\label{54,1}
\mathcal{E}_{m,s}^{\parallel}&={\rm A}_{y_0}\mathcal{E}_{m,s}=(\id_3-n_0\otimes n_0)\mathcal{E}_{m,s},\qquad\qquad\quad\ \mathcal{K}_{e,s}^{\parallel}={\rm A}_{y_0}\,\mathcal{K}_{e,s}=(\id_3-n_0\otimes n_0) \,\mathcal{K}_{e,s},\\
\mathcal{E}_{m,s}^{\perp}&=(\id_3-{\rm A}_{y_0})\,\mathcal{E}_{m,s}=n_0\otimes n_0\,\mathcal{E}_{m,s},\qquad\qquad\quad \mathcal{K}_{e,s}^{\perp}=(\id_3-{\rm A}_{y_0})\,\mathcal{K}_{e,s}=n_0\otimes n_0\,\mathcal{K}_{e,s}.\notag
\end{align}
Here, we have used that, since 
${\rm A}_{y_0}\,=\,\id_3-(0|0|n_0)\,(0|0|n_0)^T=\,\id_3-n_0\otimes n_0$, for all $X\in \mathbb{R}^{3\times 3}$ the following equalities holds
\begin{align}\label{EKperp}
\lVert X^{\perp}\rVert^2&=\lVert (\id_3-{\rm A}_{y_0})\,X\rVert^2= \bigl\langle  \,X,(\id_3-{\rm A}_{y_0})^2\,X\bigr\rangle = \bigl\langle  \,X,(\id_3-{\rm A}_{y_0})\,X\bigr\rangle = \bigl\langle  \,X,(0|0|n_0)\,(0|0|n_0)^T\,X\bigr\rangle \notag\\&= \bigl\langle  \,(0|0|n_0)^T X,(0|0|n_0)^T\,X\bigr\rangle =\lVert X\,(0|0|n_0)^T\rVert^2
=\lVert X^T\,(0|0|n_0)\rVert^2=\lVert X^T\,n_0\rVert^2.
\end{align}
The constitutive coefficient $C=\frac{E\,h}{1-\nu^2}\,$ is the stretching (in-plane) stiffness of the shell, $D=\frac{E\,h^3}{12(1-\nu^2)}\,$ is the bending stiffness, and $\alpha_s\,$, $\alpha_t$ are two shear correction factors. Also,    $ E=\frac{\mu\,(3\,\lambda+2\,\mu)}{\lambda+\mu}$ and $ \nu=\frac{\lambda}{2\,(\lambda+\mu)}
$ denote the Young modulus and  Poisson ratio of the isotropic and homogeneous material. In the numerical treatment of non-linear shell problems, the values of the shear correction factors have been set to  $\alpha_s=5/6$, $\alpha_t=7/10$ in \cite{Pietraszkiewicz10}. The value $\alpha_s=5/6$ is a classical suggestion, which has been previously deduced analytically by Reissner in the case of plates \cite{Reissner45,Naghdi72}.
Also, the value $\,\alpha_t=7/10\,$ was proposed earlier in \cite[see p.78]{Pietraszkiewicz-book79} and has been suggested in the work \cite{Pietraszkiewicz79}.
However, the discussion concerning the possible values of  shear correction factors for shells is long and controversial in the literature \cite{Naghdi72,Naghdi-Rubin95}.

We write the strain energy density \eqref{56}   in the equivalent form
\begin{align}\label{57}
2\,W_{\rm P}(\mathcal{E}_{m,s} ,\mathcal{K}_{e,s})
=& \,C\,(1-\nu)\,\big[ \lVert \sym({\rm A}_{y_0} \mathcal{E}_{m,s})\rVert^2+ \lVert \skw({\rm A}_{y_0} \mathcal{E}_{m,s})\rVert^2\big]  \notag\\
&+\,C\,\nu \,[{\rm tr}({\rm A}_{y_0} \mathcal{E}_{m,s})]^2 + \alpha_{s\,}C(1-\nu) \, \lVert  \mathcal{E}_{m,s}^T  {n}_0\rVert^2\notag \\
&+\,D\,(1-\nu)\,\big[ \lVert \sym({\rm A}_{y_0} \mathcal{K}_{e,s})\rVert^2+ \lVert \skw({\rm A}_{y_0} \mathcal{K}_{e,s})\rVert^2\big]  \\
&+\,D\,\nu \,[{\rm tr}({\rm A}_{y_0} \mathcal{K}_{e,s})]^2  + \alpha_{t\,}D(1-\nu) \,     \lVert \mathcal{K}_{e,s}^T  {n}_0\rVert^2\,.\notag
\end{align}
The coefficients in \eqref{57} are expressed in terms of the Lam\'e constants of the material $\lambda$ and $\mu$ now  by the relations
\begin{equation*}
C \,\nu\,=\frac{4\,\mu\,(\lambda+\mu)}{\lambda+2\,\mu}\,h\,,\qquad\! C(1\!-\!\nu)=2\,\mu\, h,\qquad\!
D \,\nu\,=\frac{4\,\mu\,(\lambda+\mu)}{\lambda+2\,\mu}\,\frac{h^3}{12}\,,\qquad\!
D(1\!-\!\nu)=\mu\,\dfrac{ h^3}{6}\,.
\end{equation*}

In \cite{Eremeyev06}, Eremeyev and Pietraszkiewicz have proposed a more general form of the strain energy density, namely
\begin{align}\label{58}
2\, W_{\rm EP}(\mathcal{E}_{m,s} ,\mathcal{K}_{e,s})=& \, \alpha_1\,\big(\mathrm{tr} \, \mathcal{E}_{m,s}^{\parallel}\big)^2 +\alpha_2 \, \mathrm{tr} \big(\mathcal{E}_{m,s}^{\parallel}\big)^2    + \alpha_3 \,\mathrm{tr}\big((\mathcal{E}_{m,s}^{\parallel})^T \mathcal{E}_{m,s}^{\parallel} \big)  + \alpha_4  \,   \lVert  \mathcal{E}_{m,s}^T  {n}_0\rVert^2 \notag\\
& + \beta_1\,\big(\mathrm{tr}\,  \mathcal{K}_{e,s}^{\parallel}\big)^2 +\beta_2\,  \mathrm{tr} \big(\mathcal{K}_{e,s}^{\parallel}\big)^2    + \beta_3\, \mathrm{tr}\big((\mathcal{K}_{e,s}^{\parallel})^T \mathcal{K}_{e,s}^{\parallel} \big)  + \beta_4 \,    \lVert  \mathcal{K}_{e,s}^T  {n}_0\rVert^2.
\end{align}
Already, note the absence of coupling terms involving $\mathcal{K}_{e,s}^{\parallel}$ and $\mathcal{E}_{m,s}^{\parallel}$. The eight coefficients $\alpha_k\,$, $\beta_k$ ($k=1,2,3,4$) can depend in general on the structure curvature tensor  $\mathcal{K}^0$ of the reference configuration. 
We can decompose the strain energy density  \eqref{58}  in the  in-plane part  $W_{\rm \text{plane}-EP}(\mathcal{E}_{m,s})$ and the curvature part $W_{\rm \text{curv}-EP}(\mathcal{K}_{e,s})$ and write their expressions  in  the form
\begin{align}\label{59}
W_{\rm EP}(\mathcal{E}_{m,s} ,\mathcal{K}_{e,s})=& \,W_{\rm \text{plane}-EP}(\mathcal{E}_{m,s})+ W_{\rm \text{curv}-EP}(\mathcal{K}_{e,s})\,,
\\
2\,W_{\rm \text{plane}-EP}(\mathcal{E}_{m,s})=&\, (\alpha_2\!+\!\alpha_3) \,\lVert  \text{sym}\, \mathcal{E}_{m,s}^{\parallel}\rVert^2\!+ (\alpha_3\!-\!\alpha_2)\,\lVert \text{skew}\, \mathcal{E}_{m,s}^{\parallel}\rVert^2\!+ \alpha_1\,\big(\mathrm{tr} ( \mathcal{E}_{m,s}^{\parallel})\big)^2 +  \alpha_4\,\lVert \mathcal{E}_{m,s}^Tn^0  \rVert^2, \vspace{4pt}\notag\\
2\, W_{\rm \text{curv}-EP}(\mathcal{K}_{e,s})=&\, (\beta_2\!+\!\beta_3) \,\lVert  \text{sym}\, \mathcal{K}_{e,s}^{\parallel}\rVert^2\!+ (\beta_3\!-\!\beta_2)\,\lVert \text{skew}\, \mathcal{K}_{e,s}^{\parallel}\rVert^2\!+ \beta_1\,\big(\mathrm{tr} ( \mathcal{K}_{e,s}^{\parallel})\big)^2
+ \beta_4\,\lVert \mathcal{K}_{e,s}^Tn^0  \rVert^2.\notag
\end{align}

Since in all the energies presented until now in this section, there exists no coupling terms in $\mathcal{E}_{m,s}$ and $\mathcal{K}_{e,s}$, in the rest of this section, we compare them with a particular form of the energy proposed in our new model, i.e.,
\begin{align}\label{epart}
W_{\rm our}\big(  \mathcal{E}_{m,s},\mathcal{K}_{e,s}  \big)=\,&  \Big(h+{\rm K}\,\dfrac{h^3}{12}\Big)\,
W_{\mathrm{shell}}\big(    \mathcal{E}_{m,s} \big) + \Big(h-{\rm K}\,\dfrac{h^3}{12}\Big)\,
W_{\mathrm{curv}}\big(  \mathcal{K}_{e,s} \big),
\end{align}
where
\begin{align}
W_{\mathrm{shell}}(  X) & = \mu\,\lVert  \mathrm{sym} \,X\rVert^2 +  \mu_{\rm c}\,\lVert \mathrm{skew} \,X\rVert^2 +\,\dfrac{\lambda\,\mu}{\lambda+2\mu}\,\big[ \mathrm{tr}   \,X\,\big]^2, \notag\\
W_{\mathrm{curv}}(  X )&=\mu\, {L}_{\rm c}^2 \left( b_1\,\lVert  \dev \,\text{sym} \,X\rVert^2+b_2\,\lVert \text{skew}\,X\rVert^2+\,4\,b_3\,
[\tr(X)]^2\right) .\notag
\end{align}
To this aim, we consider the decompositions \eqref{54,1} and an arbitrary matrix $X\,=\,(*|*|0)\,[	\nabla_x \Theta(0)]^{-1}$. Since 
${\rm A}_{y_0}^2={\rm A}_{y_0}\in\mathrm{Sym}(3)$ and $X{\rm A}_{y_0}=X$ we have
\begin{align*}
 \bigl\langle  (\id_3-{\rm A}_{y_0}) \,X ,  {\rm A}_{y_0} \,X\bigr\rangle = \bigl\langle  ({\rm A}_{y_0}-{\rm A}_{y_0}^2) \,X ,  \,X\bigr\rangle =0,
\end{align*}
but also
\begin{align}
 (\id_3-{\rm A}_{y_0}) \,X^T=\big(X(\id_3-{\rm A}_{y_0})\big)^T=\big(X-X{\rm A}_{y_0}\big)^T=0,
\end{align}
and consequently
\begin{align*}
 \bigl\langle  X^T (\id_3-{\rm A}_{y_0}) , {\rm A}_{y_0} \,X\bigr\rangle  = 0 \qquad \text{as well as} \qquad
 \bigl\langle  X^T (\id_3-{\rm A}_{y_0}) ,(\id_3-{\rm A}_{y_0}) \,X\bigr\rangle  = 0.
\end{align*}

 Hence, we deduce that for all $X\,=\,(*|*|0)\cdot[	\nabla_x \Theta(0)]^{-1}$ we have the following split in the expression of the considered quadratic forms
 \begin{align}
 W_{\mathrm{shell}}(  X) & =\,   \mu\,\lVert  \mathrm{sym}   \,X^\parallel\rVert^2 +  \mu_{\rm c}\,\lVert \mathrm{skew}  \,X^\parallel\rVert^2+\,   \frac{\mu+\mu_{\rm c}}{2}\,\lVert X^\perp\rVert^2 +\,\dfrac{\lambda\,\mu}{\lambda+2\mu}\,\big[ \mathrm{tr}   (X)\big]^2, \notag\\
 W_{\mathrm{curv}}(  X )&=\,\mu\, {L}_{\rm c}^2 \left( b_1\,\lVert\text{sym} \,X^\parallel\rVert^2+b_2\,\lVert \text{skew}\,X^\parallel\rVert^2+\frac{b_1+b_2}{2}\,\lVert X^\perp\rVert^2+\,\frac{12\,b_3-b_1}{3}\,
 [\tr(X)]^2\right)
 {,}\notag
 \end{align}
where we have set $X^\parallel\coloneqq{\rm A}_{y_0} \,X$ and $X^\perp\coloneqq(\id_3-{\rm A}_{y_0}) \,X$.
Moreover, using that for all $X\,=\,(*|*|0)\,[	\nabla_x \Theta(0)]^{-1}$  it holds true that
\begin{align}\label{symA1-A}
\tr(X^\perp)= \tr\big( (\id_3-{\rm A}_{y_0}) X \big) = \tr(X)-\tr({\rm A}_{y_0} X ) = \tr(X)-\tr(X\,{\rm A}_{y_0} )
=0,
\end{align}

 we obtain for our model
\begin{align}\label{expW}
W_{\mathrm{shell}}\big(    \mathcal{E}_{m,s} \big)=&\, \mu\,\lVert  \mathrm{sym}\,   \,\mathcal{E}_{m,s}^{\parallel}\rVert^2 +  \mu_{\rm c}\,\lVert\mathrm{skew}\,   \,\mathcal{E}_{m,s}^{\parallel}\rVert^2 +\,\dfrac{\lambda\,\mu}{\lambda+2\mu}\,\big[ \mathrm{tr}     (\mathcal{E}_{m,s}^{\parallel})\big]^2+\frac{\mu +  \mu_{\rm c}}{2}\,\lVert \mathcal{E}_{m,s}^{\perp}\rVert^2\\
=&\, \mu\,\lVert  \mathrm{sym}\,   \,\mathcal{E}_{m,s}^{\parallel}\rVert^2 +  \mu_{\rm c}\,\lVert \mathrm{skew}\,   \,\mathcal{E}_{m,s}^{\parallel}\rVert^2 +\,\dfrac{\lambda\,\mu}{\lambda+2\mu}\,\big[ \mathrm{tr}    (\mathcal{E}_{m,s}^{\parallel})\big]^2+\frac{\mu +  \mu_{\rm c}}{2}\,\lVert \mathcal{E}_{m,s}^T\,n_0\rVert^2,\notag
\end{align}
and
\begin{align}
W_{\mathrm{curv}}( \mathcal{K}_{e,s} )=&\,\mu\, {L}_{\rm c}^2 \left( b_1\,\lVert  \text{sym} \,\mathcal{K}_{e,s}^{\parallel}\rVert^2+b_2\,\lVert \text{skew}\,\mathcal{K}_{e,s}^{\parallel}\rVert^2\
+\frac{12\,b_3-b_1}{3}\,
[\tr(\mathcal{K}_{e,s}^\parallel)]^2+ \frac{ b_1+b_2}{2}\,\lVert \mathcal{K}_{e,s}^\perp\rVert^2\right).
\end{align}

For the final comparison between the models, we rewrite  our particular energy  in the form
\begin{align}\label{epartfinal}
W_{\rm our}\big(  \mathcal{E}_{m,s},\mathcal{K}_{e,s}  \big)=\,&  \Big(h+{\rm K}\,\dfrac{h^3}{12}\Big)\!\Big[\mu\,\lVert  \mathrm{sym}\,   \,\mathcal{E}_{m,s}^{\parallel}\rVert^2 + \! \mu_{\rm c}\,\lVert \mathrm{skew}\,   \,\mathcal{E}_{m,s}^{\parallel}\rVert^2 \!+\!\dfrac{\lambda\,\mu}{\lambda+2\mu}\,\big[ \mathrm{tr}   \, (\mathcal{E}_{m,s}^{\parallel})\big]^2+\frac{\mu +  \mu_{\rm c}}{2}\,\lVert \mathcal{E}_{m,s}^T\,n_0\rVert^2\Big] \notag\\&+ \Big(h-{\rm K}\,\dfrac{h^3}{12}\Big)\,\mu\, {L}_{\rm c}^2 \Big[ b_1\,\lVert  \text{sym}\, \mathcal{K}_{e,s}^{\parallel}\rVert^2+b_2\,\lVert \text{skew}\,\mathcal{K}_{e,s}^{\parallel}\rVert^2\
+\frac{12\,b_3-b_1}{3}\,
[\tr(\mathcal{K}_{e,s})]^2 \\&\qquad\qquad \qquad\qquad\qquad+ \frac{ b_1+b_2}{2}\,\lVert \mathcal{K}_{e,s}^\perp\rVert^2\Big].\notag
\end{align}
\newpage
This allows us to conclude
\begin{remark}
\begin{itemize}\item[]
	\item[i)]  By comparing our $W_{\rm our}\big(  \mathcal{E}_{m,s},\mathcal{K}_{e,s}  \big)$ with $W_{\rm EP}\big(  \mathcal{E}_{m,s},\mathcal{K}_{e,s}  \big)$ we deduce the following identification of the constitutive coefficients $\alpha_1\,,...,\alpha_4,\beta_1\,,...,\beta_4$
	\begin{align}
	\alpha_1&=\Big(h+{\rm K}\,\dfrac{h^3}{12}\Big)\,\dfrac{2\mu\lambda}{2\mu+\lambda}\,,\qquad \qquad\alpha_2=\Big(h+{\rm K}\,\dfrac{h^3}{12}\Big)\,(\mu-\mu_{\rm c}),\notag\\ \alpha_3&=\Big(h+{\rm K}\,\dfrac{h^3}{12}\Big)\,(\mu+\mu_{\rm c}),\qquad \qquad\alpha_4=\Big(h+{\rm K}\,\dfrac{h^3}{12}\Big) (\mu+\mu_{\rm c}),\\
	\beta_1&=2\,\Big(h-{\rm K}\,\dfrac{h^3}{12}\Big)\,\mu\, {L}_{\rm c}^2  \frac{12\,b_3-b_1}{3},\quad 
	\beta_2=\Big(h-{\rm K}\,\dfrac{h^3}{12}\Big)\,\mu\, {L}_{\rm c}^2  (b_1-b_2),\notag\\ \beta_3&=\Big(h-{\rm K}\,\dfrac{h^3}{12}\Big)\,\mu\, {L}_{\rm c}^2  (b_1+b_2), \qquad \  \beta_4=\Big(h-{\rm K}\,\dfrac{h^3}{12}\Big)\,\mu\, {L}_{\rm c}^2  (b_1+b_2).\notag
	\end{align} 
	\item[ii)] We observe that 
	\begin{equation}\label{59,4}
	\mu_{\rm c}^{\mathrm{drill}}\,:=\,\alpha_3-\alpha_2\,=\,2\Big(h+{\rm K}\,\dfrac{h^3}{12}\Big)\,\mu_{\rm c}\,,
	\end{equation}
	which means that the in-plane rotational couple modulus $\,\mu_{\rm c}^{\mathrm{drill}}\,$ of the Cosserat shell model is determined by the Cosserat couple modulus $\,\mu_{\rm c}\,$ of the 3D Cosserat material. {An analogous conclusion is given in \cite{AltenbachEremeyev09} where linear deformations
		are considered.}

	\item[iii)] In our shell model, the constitutive coefficients are those from the three-dimensional formulation, while the influence of the curved initial shell configuration appears explicitly  in the expression of the coefficients  of the energies for the reduced two-dimensional variational problem.

	\item[iv)] The major difference between our model and the previously considered   general 6-parameter shell model is that we include terms up to order $O(h^5)$ and that, even in the case of a simplified theory of order $O(h^3)$, additional mixed terms like the membrane--bending part $\,W_{\mathrm{memb,bend}}\big(  \mathcal{E}_{m,s} ,\,  \mathcal{K}_{e,s} \big) \,$ and $W_{\mathrm{curv}}\big(  \mathcal{K}_{e,s}   {\rm B}_{y_0} \,  \big) $ are included, which are otherwise difficult to guess.
	
	\item[v)] In this section we have considered only a particular form $W_{\rm our}\big(  \mathcal{E}_{m,s},\mathcal{K}_{e,s}  \big)$  of the density energy considered in our Cosserat-shell model \eqref{e90}. However, beside the fact that mixed membrane--bending terms are included, the constitutive coefficients in our shell model depend on both the  Gau\ss \ curvature ${\rm K}$ and the mean curvature ${\rm H}$, see item i) and compare to \eqref{e90}. Moreover, due to the bilinearity of the density energy, if the final form of the energy density is expressed as a quadratic form  in terms of $\big(  \mathcal{E}_{m,s} ,\,  \mathcal{K}_{e,s} \big) $, as in the $W_{\rm EP}\big(  \mathcal{E}_{m,s},\mathcal{K}_{e,s}  \big)$, then we remark that the dependence on the mean curvature is not only the effect of the presence of the mixed terms, due to the Cayley-Hamilton equation $
	{\rm B}_{y_0}^2=2\,{\rm H}\, {\rm B}_{y_0}-{\rm K}\, {\rm A}_{y_0}
	$. See for instance the energy term $W_{\mathrm{curv}}\big(  \mathcal{K}_{e,s}   {\rm B}_{y_0}^2 \,  \big) $ or even  $
	W_{\mathrm{mp}} \big((  \mathcal{E}_{m,s} \, {\rm B}_{y_0} +  {\rm C}_{y_0} \mathcal{K}_{e,s} )   {\rm B}_{y_0} \,\big)$ from \eqref{e90}.
\end{itemize}
\end{remark}

\section{Final comments}
In this article, using a step by step transparent method, we have  extended the modelling from flat Cosserat shells (plate) to the most general case of initially curved isotropic Cosserat shells. 
For flat shells, in a numerical approach Sander et al. \cite{Sander-Neff-Birsan-15} have shown that the new shell model offers a very good concordance with available experiments in the framework of nonlinear shell modelling.  Our ansatz allows for a consistent shell model up to order $O(h^5)$ in the shell thickness. Interestingly, all $O(h^5)$ terms in the shell energy depend on the initial curvature of the shell and vanish for a flat shell. The  $O(h^5)$ terms do not come up with a definite sign, such that the additional terms can be stabilizing as well as destabilizing, depending on the local shell geometry. However, all occurring material coefficients of the shell model are uniquely determined from the isotropic three-dimensional Cosserat model and the given initial geometry of the shell. Hence,  in contrast to other Cosserat shell models, we give an explicit form of the curvature energy, and therefore, we fill a certain gap in the general 6-parameter shell theories, which all leave the precise structure of the constitutive equations wide open. 

\bigskip

\begin{footnotesize}
\noindent{\bf Acknowledgements:}   This research has been funded by the Deutsche Forschungsgemeinschaft (DFG, German Research Foundation) -- Project no. 415894848: NE 902/8-1 (P. Neff and P. Lewintan) and
BI 1965/2-1 (M. B\^irsan)).  The  work of I.D. Ghiba  was supported by a grant of the Romanian Ministry of Research
and Innovation, CNCS--UEFISCDI, project number
PN-III-P1-1.1-TE-2019-0397, within PNCDI III.

\bibliographystyle{plain} 

\addcontentsline{toc}{section}{References}


\appendix\setcounter{equation}{0}

\appendix\section*{Appendix}\setcounter{section}{1}
\subsection{Notation}\label{Sectnotation}

We denote by $\mathbb{R}^{n\times n}$, $n\in\mathbb{N}$, the set of real $n\times n$ second order tensors, written with
capital letters. We adopt the usual abbreviations of Lie-group theory, i.e.,
${\rm GL}(n)=\{X\in\mathbb{R}^{n\times n}\;|\det({X})\neq 0\}$ the general linear group,
${\rm SL}(n)=\{X\in {\rm GL}(n)\;|\det({X})=1\},\;
\mathrm{O}(n)=\{X\in {\rm GL}(n)\;|\;X^TX=\id_n\},\;{\rm SO}(n)=\{X\in {\rm GL}(n)| X^TX=\id_n,\det({X})=1\}$ with
corresponding Lie-algebras $\mathfrak{so}(n)=\{X\in\mathbb{R}^{n\times n}\;|X^T=-X\}$ of skew symmetric tensors
and $\mathfrak{sl}(n)=\{X\in\mathbb{R}^{n\times n}\;| \tr({X})=0\}$ of traceless tensors. Here, 
for $a,b\in\mathbb{R}^n$ we let $ \bigl\langle {a},{b}\bigr\rangle _{\mathbb{R}^n}$  denote the scalar product on $\mathbb{R}^n$ with
associated
{(squared)}
vector norm $\lVert a\rVert_{\mathbb{R}^n}^2= \bigl\langle {a},{a}\bigr\rangle _{\mathbb{R}^n}$.
The standard Euclidean scalar product on $\mathbb{R}^{n\times n}$ is given by
$ \bigl\langle  {X},{Y}\bigr\rangle _{\mathbb{R}^{n\times n}}={\rm tr}(X Y^T)$, and thus the
{(squared)}
Frobenius tensor norm is
$\lVert {X}\rVert^2= \bigl\langle  {X},{X}\bigr\rangle _{\mathbb{R}^{n\times n}}$. In the following we omit the index
$\mathbb{R}^n,\mathbb{R}^{n\times n}$. The identity tensor on $\mathbb{R}^{n \times n}$ will be denoted by $\id_n$, so that
${\rm tr}({X})= \bigl\langle {X},{\id}_n\bigr\rangle $. We let ${\rm Sym}(n)$ and ${\rm Sym}^+(n)$ denote the symmetric and positive definite symmetric tensors, respectively.  For all $X\in\mathbb{R}^{3\times3}$ we set ${\rm sym}\, X=\frac{1}{2}(X^T+X)\in{\rm Sym}(3)$, $\skw\, X\,=\frac{1}{2}(X-X^T)\in \mathfrak{so}(3)$ and the deviatoric part $\dev  \, X=X-\frac{1}{n}\;\tr(X)\,\id_n\in \mathfrak{sl}(n)$  and we have
the \emph{orthogonal Cartan-decomposition  of the Lie-algebra} $\mathfrak{gl}(3)$
\begin{align}
\mathfrak{gl}(3)&=\{\mathfrak{sl}(3)\cap {\rm Sym}(3)\}\oplus\mathfrak{so}(3) \oplus\mathbb{R}\!\cdot\! \id_3,\qquad\qquad
X=\dev  \,\sym \,X+ \skw\,X+\frac{1}{3}\tr(X)\, \id_3\,.
\end{align}
{We make use of the operator $\mathrm{axl}: \mathfrak{so}(3)\to\mathbb{R}^3$ associating with a matrix $A\in \mathfrak{so}(3)$ the vector $\mathrm{axl}{\,A}\coloneqq(-A_{23},A_{13},-A_{12})^T$.}

For $X\in {\rm GL}(n)$,  ${\rm Adj}({X})$  denotes the tensor of
transposed cofactors, while  the $(i, j)$ entry of the cofactor is the $(i, j)$-minor times a sign factor.  
For vectors $\xi,\eta\in\mathbb{R}^n$, we have the tensor product
$(\xi\otimes\eta)_{ij}=\xi_i\,\eta_j$. A matrix having the  three  column vectors $A_1,A_2, A_3$ will be written as 
$
(A_1\,|\, A_2\,|\,A_3).
$ 
For a given matrix $M\in \mathbb{R}^{2\times 2}$ we define the lifted quantity $
M^\flat =\begin{footnotesize}\begin{pmatrix}
M_{11}& M_{12}&0 \\
M_{21}&M_{22}&0 \\
0&0&0
\end{pmatrix}\end{footnotesize}
\in \mathbb{R}^{3\times 3}.
$

Let $\Omega$ be an open domain of $\mathbb{R}^{3}$. The usual Lebesgue spaces of square integrable functions, vector or tensor fields on $\Omega$ with values in $\mathbb{R}$, $\mathbb{R}^3$ or $\mathbb{R}^{3\times 3}$, respectively will be denoted by ${\rm L}^2(\Omega)$. Moreover, we introduce the standard Sobolev spaces 
$
{\rm H}^1(\Omega)\,=\,\{u\in {\rm L}^2(\Omega)\, |\, \nabla\, u\in {\rm L}^2(\Omega)\}, $ 
${\rm H}({\rm curl};\Omega)\,=\,\{v\in {\rm L}^2(\Omega)\, |\, {\rm curl}\, v\in {\rm L}^2(\Omega)\}
$
of functions $u$ or vector fields $v$, respectively.  For vector fields $u=\left(    u_1, u_2,u_3\right)$ with  $u_i\in {\rm H}^{1}(\Omega)$, $i=1,2,3$,
we define
$$
\nabla \,u:=\left(
\nabla\,  u_1\,|\,
\nabla\, u_2\,|\,
\nabla\, u_3
\right)^T,
$$
while for tensor fields $P$ with rows in ${\rm H}({\rm curl}\,; \Omega)$, i.e.,
$
P=\begin{footnotesize}\begin{pmatrix}
P^T.e_1\,|\,
P^T.e_2\,|\,
P^T\, e_3
\end{pmatrix}\end{footnotesize}^T$ with $(P^T.e_i)^T\in {\rm H}({\rm curl}\,; \Omega)$, $i=1,2,3$,
we define
$$
{\rm Curl}\,P:=\begin{footnotesize}\begin{pmatrix}
{\rm curl}\, (P^T.e_1)^T\,|\,
{\rm curl}\, (P^T.e_2)^T\,|\,
{\rm curl}\, (P^T\, e_3)^T
\end{pmatrix}\end{footnotesize}^T
.
$$
The corresponding Sobolev-spaces will be denoted by
$
{\rm H}^1(\Omega)$ and   ${\rm H}^1(\Curl;\Omega)$, respectively. We will use the notations: $\nabla_\xi$, $\nabla_x$, $\Curl_\xi$, $\Curl_x$ etc. to indicate the variables for which these quantities are calculated.

\subsection{Prerequisites from classical differential geometry}\label{sectGeo}
\index{differential geometry}
{Let $\omega\subset\mathbb{R}^2$ be an open domain.} A given
{regular}
mapping $y_0:\omega
\,{\to}\,
\mathbb{R}^3$, describing a \textit{surface imbedded in the
	three-dimensional space} is called {\it regular} whenever ${\rm rank}(\nabla  y_0)\,=\,2$. The column vector
\begin{align}
n_0:\,=\,\frac{\partial_{x_1} y_0\times \partial_{x_2} y_0}{\norm{\partial_{x_1} y_0\times \partial_{x_2} y_0}}\,
\end{align}
is the {\it Gau{\ss}  unit normal field} on the surface.
We need to compute
\begin{align}
{\rm Adj}[(\nabla y_0|0)]\,=\,&
\begin{footnotesize}
\begin{pmatrix}
0& 0& 0\vspace{1mm}\\
0& 0& 0\vspace{1mm}\\
\left|\begin{array}{ll}
\partial_{x_1}  y_{02}&\partial_{x_2}  y_{02}\vspace{1mm}\\
\partial_{x_1} y_{03}&\partial_{x_2} y_{03}
\end{array}\right|
&-\left|\begin{array}{ll}
\partial_{x_1}  y_{01}&\partial_{x_2}  y_{01}\vspace{1mm}\\
\partial_{x_1} y_{03}&\partial_{x_2} y_{03}
\end{array}\right|
&\left|\begin{array}{ll}
\partial_{x_1}  y_{01}&\partial_{x_2}  y_{01}\vspace{1mm}\\
\partial_{x_1}  y_{02}&\partial_{x_2}  y_{02}
\end{array}\right|
\end{pmatrix}\end{footnotesize}\, .
\end{align}
Hence, it follows
\begin{align}
n_0:\,=\,\frac{{\rm Cof}(\nabla y_0|0)\, e_3}{\norm{{\rm Cof}(\nabla y_0|0)\, e_3}},\qquad {\rm Cof}(X)=[{\rm Adj}(X)]^T\quad  \forall \, X\in \mathbb{R}^{3\times 3}.
\end{align}
The map $n_0: \omega
\,{\to}\,
 S^2$ is called the {\it
	Gau{\ss} map} (where $S^2$ is the unit sphere in $\mathbb{R}^3$) and the moving 3-frame $(\partial_{x_1} y_0|\partial_{x_2} y_0|n_0)$ is called the {\it 	Gau{\ss} frame} \index{	Gau{\ss} frame}of the surface $y_0(\omega)$, which in general is not orthonormal. The matrix representation of the {\it first fundamental form (metric)} {on  $y_0(\omega)$} is given through
\begin{align}
\index{fundamental form ! first}\index{fundamental form ! first extended}
\label{first_fundamental_form}
{\rm I}_{y_0}&:\,=\,[{\nabla  y_0}]^T\,{\nabla  y_0}\,=\,\begin{footnotesize}\begin{pmatrix}
\norm{{\partial_{x_1} y_0}}^2 & \Mprod{{\partial_{x_1} y_0}}{{\partial_{x_2} y_0}} \\
\Mprod{{\partial_{x_1} y_0}}{{\partial_{x_2} y_0}} &\norm{{\partial_{x_2} y_0}}^2
\end{pmatrix}\end{footnotesize}\in\mathbb{R}^{2\times 2}\, .
\end{align}
Because ${\rm rank} (\nabla y_0)\,=\,2$, the tensor $[\nabla y_0]^T\nabla y_0$ is positive definite. 

The metric \index{metric} alone is not sufficient to describe the shape of a surface in the ambient
three-dimensional Euclidean space,
the curvature is also needed. However, in the case $({\nabla  y_0}|n_0)\in{\rm{SO}}(3)$, the metric is indeed enough. With the metric, the length and angles (and changes of length and angles) of a surface can be completely described.

The matrix representation of the {\it  second fundamental form} {on  $y_0(\omega)$} providing a measure for curvature of the
surface is given by
\begin{align}
\index{fundamental form ! second}\index{fundamental form ! second extended}
\label{second_fundamental_form}
&{\rm II}_{y_0}:\,=\,-[{\nabla  y_0}]^T\,{\nabla  n_0}\,=\,-({\partial_{x_1} y_0}|{\partial_{x_2} y_0})^T ({\partial_{x_1} n_0}|{\partial_{x_2} n_0})\,=\,-\begin{footnotesize}\begin{pmatrix}
\Mprod{{\partial_{x_1} y_0}}{\partial_{x_1} n_0} &  \Mprod{{\partial_{x_1} y_0}}{{\partial_{x_2} n_0}}  \\
\Mprod{{\partial_{x_2} y_0}}{{\partial_{x_1} n_0}} &  \Mprod{{\partial_{x_2} y_0}}{{\partial_{x_2} n_0}}  \\
\end{pmatrix}\end{footnotesize}\in \mathbb{R}^{2\times 2}\, .
\end{align}
Since $n_0$ is orthogonal to the tangent space $T_x y_0$ of the surface $y_0$, the relation
$0\,=\,\partial_{x_1} \Mprod{{\partial_{x_2} y_0}}{n_0}\,=\,\partial_{x_2} \Mprod{{\partial_{x_1} y_0}}{n_0 }$
shows easily that ${\rm II}_{y_0}\in {\rm Sym}(2)$.

The {\it third fundamental form} of the surface {$y_0(\omega)$} in matrix representation is defined as
\begin{align}
\label{third_fundamental_form}
{\rm III}_{y_0}:\,=\,{\nabla  n_0}^T\,{\nabla  n_0}\,=\, \begin{footnotesize}\begin{pmatrix}
\norm{{\partial_{x_1} n_0}}^2 & \Mprod{{\partial_{x_1} n_0}}{{\partial_{x_2} n_0}} \\
\Mprod{{\partial_{x_2} n_0}}{{\partial_{x_1} n_0}} & \norm{{\partial_{x_2} n_0}}^2
\end{pmatrix}\end{footnotesize}\in \mathbb{R}^{2\times 2}\, .
\end{align}

Since  $[\nabla y_0]^T\nabla y_0$ is positive definite, the first fundamental form induces a scalar product
$
g(\xi_1,\xi_2):\,=\, \bigl\langle  {\rm I}_{y_0}\xi_1,\xi_2\bigr\rangle _{\mathbb{R}^2}\, ,
$
while the second fundamental form induces a symmetric bilinear form
$
\widetilde{g}(\xi_1,\xi_2):\,=\, \bigl\langle  {\rm II}_{y_0}\xi_1,\xi_2\bigr\rangle _{\mathbb{R}^2}\, .
$

The first fundamental form and the second fundamental form are connected by {\it the Weingarten map} (or shape operator) {which we again identify with the associated matrix} ${\rm L}_{y_0}\in\mathbb{R}^{2\times 2}$, i.e., $\forall\,\,\xi_1,\xi_2\in\mathbb{R}^2$  we have
$
 \bigl\langle  {\rm I}_{y_0} \xi_1,{\rm L}_{y_0}\, \xi_2\bigr\rangle _{\mathbb{R}^2}\,=\, \bigl\langle  {\rm II}_{y_0} \xi_1, \xi_2\bigr\rangle _{\mathbb{R}^2}\, .
$
This is an implicit definition of the Weingarten map ${\rm L}_{y_0}$.  Using the definitions  of the first and the second fundamental form, we have
\begin{align}
& \bigl\langle  [{\nabla  y_0}]^T\,{\nabla  y_0}\  \xi_1,{\rm L}_{y_0} \xi_2\bigr\rangle _{\mathbb{R}^2}\,=\, \bigl\langle  -[{\nabla  y_0}]^T\,{\nabla  n_0} \ \xi_1, \xi_2\bigr\rangle _{\mathbb{R}^2}\\
\quad \Leftrightarrow\quad 
& \bigl\langle  ({\rm L}^T_{y_0}[{\nabla  y_0}]^T\,{\nabla  y_0}\ +[{\nabla  y_0}]^T\,{\nabla  n_0}) \ \xi_1, \xi_2\bigr\rangle _{\mathbb{R}^2}\,=\,0,  \ \ \  \forall\ \xi_1,\xi_2\in\mathbb{R}^2\, .\notag
\end{align}
Thus, we have
\begin{align}
{\rm L}^T_{y_0}[{\nabla  y_0}]^T\,{\nabla  y_0}\ +[{\nabla  y_0}]^T\,{\nabla  n_0}\,=\,0\, 
\quad \Leftrightarrow\quad 
&
{\rm L}^T_{y_0}[{\nabla  y_0}]^T\,{\nabla  y_0}\ \,=\,-[{\nabla  y_0}]^T\,{\nabla  n_0}.
\end{align}
Hence, we obtain the following alternative expression for the {Weingarten map} 
{via the so called {\it Weingarten equations}}
:
\begin{align}
{\rm L}_{y_0} \,=\,-([{\nabla  y_0}]^T\,{\nabla  y_0})^{-1} ({\nabla  n_0}^T\,{\nabla  y_0})\qquad \text{or} \qquad   {\rm L}_{y_0}\,=\, {\rm I}_{y_0}^{-1} {\rm II}_{y_0}\, .
\end{align}

Moreover, using the symmetry of the second fundamental form we see that the Weingarten map satisfies:
\begin{align}\label{relw12}
\nabla y_0\, {\rm L}_{y_0}\,=\,-\nabla n_0.
\end{align}
We have also
\begin{align}\label{relw13}
{\rm III}_{y_0} \,=\, {\rm I}_{y_0}{\rm L}_{y_0}^2 \,=\, {\rm II}_{y_0} {\rm I}_{y_0}^{-1} {\rm II}_{y_0}\,.
\end{align}

The {\it Gau{\ss} curvature} ${\rm K}$ of the surface {$y_0(\omega)$}  is determined by
\begin{align}
\index{Gauss curvature}
\label{gauss_curvatureA}
{\rm K} :\,=\,{{\rm det}({{\rm II}_{y_0}\,{\rm I}_{y_0}^{-1}}})\,=\,{\rm det}{({\rm L}_{y_0})}\, ,
\end{align}
and the {\it mean curvature} $\,{\rm H}\,$ through
\begin{align}
\index{mean curvature}
\label{mean_curvatureA}
2\,{\rm H}\, :\,=\,{\rm tr}({{\rm L}_{y_0}}) \, .
\end{align}

The following classification is standard. The surface $y_0$ is locally
\begin{align}
\left\{\begin{array}{l}
\text{elliptic} \\
\text{parabolic} \\
\text{hyperbolic}\\
\text{planar}\
\end{array}\right. \text{ \ \ \ at\ \ \  } (x_1,x_2) \in \omega \text{\ \ \ if \ \ \ }{\rm K} \text{ \ \ is\ \ } \left\{\begin{array}{l}
>0 \\
\,=\,0 \quad {\text{and } {\rm H}\neq0} \\
<0\\
\,=\,0 \quad {\text{and } {\rm H}=0} \
\end{array}\right.\, .
\end{align}

It is well known that  $\,{\rm H} \,=\,0$ is satisfied for all sufficiently regular stationary points of the minimal surface area functional.
The Caley-Hamilton theorem implies
\[ 
{\rm L}_{y_0}^2 -2\,{\rm H}\, {\rm L}_{y_0} +{\rm K}\,\id_{2}\,=\,0.
\]
Thus, the relation $${\rm III}_{y_0}-2\,{\rm H}\,{\rm II}_{y_0}+{\rm K}\,{\rm I}_{y_0}\,=\,0$$ (\cite[Prop. 3.5.6]{Klingenberg78}) is a consequence of the
Caley-Hamilton theorem and shows that ${\rm III}_{y_0}$ is symmetric and is not independent of ${\rm I}_{y_0},{\rm II}_{y_0}$.
The principal curvatures $\kappa_1,\kappa_2$ \index{principal curvatures} are the solutions of
the characteristic equation of
${\rm L}_{y_0}$, i.e., $$\kappa^2-{\rm tr}({{\rm L}_{y_0}})\,\kappa+{\rm det}{({\rm L}_{y_0})}\,=\,\kappa^2-2\,{\rm H}\,\kappa+{\rm K}\,=\,0.$$

We define the lifted quantity $\widehat{\rm I}_{y_0} \in \mathbb{R}^{3\times 3}$ by
\begin{align}
\label{first_fundamental_form_lift1}
\widehat{\rm I}_{y_0}\:\,=\,({\nabla  y_0}|n_0)^T({\nabla  y_0}|n_0)\,=\,\begin{footnotesize}\begin{pmatrix}
\lVert {\partial_{x_1} y_0}\rVert^2 &  \bigl\langle {\partial_{x_1} y_0},{\partial_{x_2} y_0}\bigr\rangle  &0 \\
 \bigl\langle  {\partial_{x_1} y_0},{\partial_{x_2} y_0}\bigr\rangle  &\lVert \partial_{x_2} y_0\rVert^2&0 \\
0 &0&1 \\
\end{pmatrix}\end{footnotesize}\,=\,\begin{footnotesize}\begin{pmatrix}
&  {\rm I}_{y_0}  & 0 \\
&   & 0  \\
0 & 0 & 1
\end{pmatrix}\end{footnotesize}= {\rm I}_{y_0}^\flat+\widehat{0}_3\, ,
\end{align}
where $ {\rm I}_{y_0}^\flat=\begin{footnotesize}\begin{pmatrix}
&  {\rm I}_{y_0}  & 0 \\
&   & 0  \\
0 & 0 & 0
\end{pmatrix}\end{footnotesize}$ and $\widehat{0}_3=\begin{footnotesize}\begin{pmatrix}
0& 0  & 0 \\
0& 0  & 0  \\
0 & 0 & 1
\end{pmatrix}\end{footnotesize}$.
Hence, $\widehat{\rm I}_{y_0}$ has the properties
\begin{align}
{\rm det}{({\rm I}_{y_0})}\,=\,{\rm det}{(\widehat{\rm I}_{y_0} )}\,=\,{\rm det}{({\nabla  y_0}|n_0)}^2 \notag,\ \  \ \    {\rm tr}({\rm I}_{y_0})+1\,=\,{\rm tr}(\widehat{\rm I}_{y_0} )\, .
\end{align}

Corresponding to the second fundamental form we define the lifted quantity $\widehat{\rm II}_{y_0}\in\mathbb{R}^{3\times3}$ by
\begin{align}
\widehat{\rm II}_{y_0}&\,=\,-(\nabla y_0|n_0)^T(\nabla n_0|n_0)\,=\, -\begin{footnotesize}\begin{pmatrix}
 \bigl\langle {{\partial_{x_1} y_0}},{{\partial_{x_1} n_0}}\bigr\rangle  &   \bigl\langle {{\partial_{x_1} y_0}},{{\partial_{x_2} n_0}} \bigr\rangle & 0 \\
 \bigl\langle {{\partial_{x_2} y_0}},{{\partial_{x_1} n_0}} \bigr\rangle &   \bigl\langle {{\partial_{x_2} y_0}},{{\partial_{x_2} n_0}} \bigr\rangle & 0  \\
0 & 0 & 1
\end{pmatrix}\end{footnotesize}\,=\,\begin{footnotesize}\begin{pmatrix}
&  {\rm II}_{y_0}  & 0 \\
&   & 0  \\
0 & 0 & -1
\end{pmatrix}\end{footnotesize}\,= \,{\rm II}_{y_0}^\flat-\widehat{0}_3\, ,
\end{align}
where $ {\rm II}_{y_0}^\flat=\begin{footnotesize}\begin{pmatrix}
&  {\rm II}_{y_0}  & 0 \\
&   & 0  \\
0 & 0 & 0
\end{pmatrix}\end{footnotesize}$.
It has the properties
\begin{equation}
\quad{\rm det}{({\rm II}_{y_0})}\,=\,-{\rm det}{(\widehat{\rm II}_{y_0})}, \ \ \
\quad{\rm tr}({\rm II}_{y_0})\,=\,{\rm tr}(\widehat{\rm II}_{y_0})+1\, .
\end{equation}

Let us consider as well the lifted Weingarten map $\widehat{\rm L}_{y_0}:\mathbb{R}^3\rightarrow \mathbb{R}^3$ defined by
\begin{align}
\widehat{\rm L}_{y_0}^T\,=\,\widehat{\rm II}_{y_0}\widehat{\rm I}_{y_0}^{-1}\, .
\end{align}

Thus, we have
\begin{align}
\widehat{\rm L}_{y_0}^T\,=\,& -({\nabla  y_0}|n_0)^T({\nabla  n_0}|n_0)[({\nabla  y_0}|n_0)^T({\nabla  y_0}|n_0)]^{-1}\nonumber \\\,=\,&\begin{footnotesize}\begin{pmatrix}
&  {\rm II}_{y_0}  & 0 \\
&   & 0  \\
0 & 0 & -1
\end{pmatrix}\end{footnotesize}\begin{footnotesize}\begin{pmatrix}
&  {\rm I}_{y_0}^{-1}  & 0 \\
&   & 0  \\
0 & 0 & 1
\end{pmatrix}\end{footnotesize}\,=\,\begin{footnotesize}\begin{pmatrix}
& {\rm II}_{y_0} {\rm I}_{y_0}^{-1}   & 0 \\
&  & 0  \\
0 & 0 &- 1
\end{pmatrix}\end{footnotesize}\,=\,\begin{footnotesize}\begin{pmatrix}
&  {\rm L}_{y_0}^T  & 0 \\
&  & 0  \\
0 & 0 &- 1
\end{pmatrix}\end{footnotesize}\, .
\end{align}
The lifted Weingarten map $\widehat{\rm L}$ has the following properties
\begin{align}
{\rm det}{(\widehat{\rm L}_{y_0})}\,=\,-{\rm det}{({\rm L}_{y_0})},\ \ \ \   {\rm tr}(\widehat{{\rm L}}_{y_0})&\,=\,{\rm tr}({\rm L}_{y_0})-1\, .
\end{align}

\subsection{Properties of the diffeomorphism $\Theta$  }\label{proofLemmaMircea}

\begin{lemma}\label{lemmaMircea}
	For all   $A\in\mathbb{R}^{2\times 2}$, there exists $a>0$, such that for all $x_3\in (-a,a)$, the formula 
	\begin{align}
	(\id_2-x_3\, A)^{-1}\,=\,\frac{1}{1- x_3\,{\rm tr} A+\,x_3^2\det{A}}[(1- x_3\, {\rm tr} A)\id_2+x_3 A]
	\end{align}
	holds true.
\end{lemma}
\begin{proof} Consider an arbitrary matrix $A\in \mathbb{R}^{2\times 2}$. From the continuity of the mapping $x_3\mapsto 1- x_3\,\tr( A)+\,x_3^2\det({A})$, it exists an $a>0$ such that $1- x_3\,\tr (A)+\,x_3^2\det({A})>0$ for all $x_3\in (-a,a)$. Taking  $x_3\in (-a,a)$, 
	we compute
	\begin{align}
	(\id_2-x_3\, A)^{-1}\, (\id_2-x_3\, A)&\,=\,	\frac{1}{1- x_3\,\tr (A)+\,x_3^2\det({A})}[(1- x_3\, \tr (A))\id_2+x_3 A]\, [\id_2-x_3\, A] \\
	&\,=\,	\frac{1}{1- x_3\,\tr (A)+\,x_3^2\det({A})}[(1- x_3\, \tr (A))\id_2-\, x_3\,A+x_3^2\, \tr (A)\, \,A+ x_3 A-x_3^2\, A^2]\notag \\
	&\,=\,	\frac{1}{1-x_3\,\tr A+\,x_3^2\det{A}}[(1- x_3\, \tr (A))\id_2+x_3^2\, (\underbrace{\tr (A)\, \,A- A^2}_{\det (A)\, \id_2})]\,=\,\id_2.\notag\qedhere
	\end{align}
\end{proof}
With the help of the above lemma, we prove the following proposition.
\begin{proposition}\label{propnablatheta}
	The diffeomorphism $\Theta$ has the following properties for all $x_3$:
	\begin{itemize}
		\item[i)] ${\rm det}(\nabla_x \Theta(x_3))\,=\,{\rm det}(\nabla y_0|n_0)\Big[1-2\,x_3\,{\rm H}\,+x_3^2 \, {\rm K}\Big]$;
		\item[ii)]	$\nabla_x\Theta(x_3)$  belongs to ${\rm GKC }:=\{X\in \mathrm{GL}^+(3)\ |\ X^T \, X\, e_3\,=\,\varrho^2 e_3,\ \varrho\in \mathbb{R}^+\}$;
		\item[iii)] if $ {h}\,|{\kappa_1}|<\frac{1}{2}, {h}\,|{\kappa_2}|<\frac{1}{2}$, then for all $x_3\in \left(-\frac{h}{2},\frac{h}{2}\right)$: \\
		$
		[	\nabla_x \Theta(x_3)]^{-1}\,=\,\dd \frac{1}{1-2\,{\rm H}\, x_3+{\rm K}\, x_3^2}\left[\id_3+x_3({\rm L}_{y_0}^\flat-2\,{\rm H}\, \id_3)+x_3^2\, {\rm K}\begin{footnotesize}\begin{pmatrix}
		0&0&0 \\
		0&0&0 \\
		0&0&1
		\end{pmatrix}\end{footnotesize}
		\right] [	\nabla_x \Theta(0)]^{-1};
		$
		
	\end{itemize}
\end{proposition}
\begin{proof}
	Since $\lVert n_0\rVert^2\,=\,1$ we have $ \bigl\langle  n_0, \partial_{x_\alpha} n_0\bigr\rangle =0$,  $ \bigl\langle  n_0, \partial_{x_\alpha} y_0\bigr\rangle =0$, $\alpha=1,2$.
	Using the  Weingarten map (or shape operator) ${\rm L}_{y_0}\in \mathbb{R}^{2\times 2}$ defined in Appendix \ref{sectGeo} by relation (\ref{relw12}) we deduce the following form of $\nabla_x \Theta$:
	\begin{align}
	\nabla_x \Theta(x_3)\,=\,(\nabla y_0|\,n_0)+x_3(\nabla n_0|0)\,=\,(\nabla y_0|\,n_0)-x_3(\nabla y_0 \, {\rm L}_{y_0}|0),
	\end{align}
	which implies
	\begin{align}\label{repreznL}
	\nabla_x \Theta(x_3)\,=\,\nabla_x \Theta(0)\begin{footnotesize}\begin{pmatrix}
	& \id_2-x_3\, {\rm L}_{y_0}&0 \\
	& & 0 \\
	0&0&1
	\end{pmatrix}\end{footnotesize}.
	\end{align}
	Then, we have
$
	{\rm det}{(\nabla_x \Theta (x_3))}\,=\,{\rm det}{(\nabla_x \Theta(0))}\,{\rm det}(\id_2-x_3{\rm L}_{y_0}).
$
	Using the two-dimensional expansion of the determinant 
	$
	{\rm det}(\id_2-x_3{\rm L}_{y_0})\,=\,1-x_3\,{\rm tr}({\rm L}_{y_0})+ x_3^2\, {\rm det}({\rm L}_{y_0}),
$
	we deduce
$
	{\rm det}(\nabla_x \Theta(x_3))\,=\,{\rm det}{(\nabla_x \Theta(0))}\,\Big[1-x_3\,{\rm tr}({\rm L}_{y_0})+x_3^2\, {\rm det}({\rm L}_{y_0})\Big]\, .
$
	In terms of the mean curvature $\,{\rm H}\,$ and the Gauss curvature ${\rm K}$, we have therefore the well known formula i).

	Regarding ii), 
	  {we have already seen in \eqref{first_fundamental_form_lift1} that $\nabla_x \Theta(0)=\,(\nabla y_0|\,n_0)\in\mathrm{GKC}$, so that the conclusion $\nabla_x \Theta (x_3)\in{\rm GKC}$ then follows from the decomposition \eqref{repreznL}.}

	In order to prove iii), we prove that the conditions $ {h}\,|{\kappa_1}|<\frac{1}{2},\, {h}\,|{\kappa_2}|<\frac{1}{2}$, ensure that 
	$$
	1-2\,{\rm H}\, x_3+{\rm K}\, x_3^2\neq 0
	\qquad \text{ for all} \qquad x_3\in \left(-\frac{h}{2},\frac{h}{2}\right).$$ 
		It follows that
$
h^2|K|=h^2\, |\kappa_1|\,|\kappa_2|<\frac{1}{4}, \  2\,h\, |H|=h\, |\kappa_1+\kappa_2|<1.
$
Hence,
$
	1-2\,{\rm H}\, x_3+{\rm K}\, x_3^2\geq 1-2\,|{\rm H}|\, |x_3|-|{\rm K}|\,| x_3^2|>0.
$
Thus,  we use the representation \eqref{repreznL}
	and we apply  Lemma \ref{lemmaMircea} and iii) is proven.
	\end{proof}

\subsection{Properties of the tensors ${\rm A}_{y_0}, {\rm B}_{y_0}$  and ${\rm C}_{y_0 }$  }\label{AppendixpropAB}

\begin{proposition}\label{propAB}
	The  tensors ${\rm A}_{y_0}, {\rm B}_{y_0}$  and ${\rm C}_{y_0 }$ have the following properties:
	\begin{itemize}
		\item [i)] $ {\rm A}_{y_0}\,=\,[\nabla_x \Theta(0)]^{-T} {\rm I}_{y_0}^\flat [\nabla_x \Theta(0)]^{-1}\in {\rm Sym}(3),$\quad 
		$\tr({\rm A}_{y_0})\,=\,2,$ \quad	${\det}({\rm A}_{y_0})\,=\,0$,

		${\rm A}_{y_0}\,=\,\id_3-(0|0|\nabla_x \Theta(0)\, e_3)\,[	\nabla_x \Theta(0)]^{-1}\,=\,\id_3-(0|0|n_0)\,(0|0|n_0)^T$;

		\item[ii)] ${\rm B}_{y_0}\,=\,[\nabla_x \Theta(0)]^{-T} {\rm II}_{y_0}^\flat [\nabla_x \Theta(0)]^{-1}\,=\,\nabla_x \Theta(0) {\rm L}_{y_0}^\flat [\nabla_x \Theta(0)]^{-1}\in {\rm Sym}(3)$,  
		
		$\tr({\rm B}_{y_0})\,=\,2\,{\rm H}\,$,\quad  ${\det}[{\rm B}_{y_0}]\,=\,0,$ \quad   $\tr( {\rm Cof}\,{\rm B}_{y_0})\,=\, {\rm K},$\quad 
		${\rm Cof}\,{\rm B}_{y_0}\,=\,\nabla_x \Theta(0)\,\begin{footnotesize}\begin{pmatrix}
		0&0&0 \\
		0&0&0 \\
		0&0& {\rm K}
		\end{pmatrix}\end{footnotesize}\, [\nabla_x \Theta(0)]^{-1}\,,$
		
		$ {\rm B}_{y_0}^2 \,=\, [\nabla_x \Theta(0)]^{-T} {\rm III}_{y_0}^\flat[\nabla_x \Theta(0)]^{-1} \,=\, [\nabla_x \Theta(0)] \big( {\rm L}_{y_0}^\flat\big)^2  [\nabla_x \Theta(0)]^{-1}\;; $ 
		
		\item[iii)] ${\rm B}_{y_0}$ satisfies the equation of Cayley-Hamilton type
		$
		{\rm B}_{y_0}^2-2\,{\rm H}\, {\rm B}_{y_0}+{\rm K}\, {\rm A}_{y_0}\,=\,0_3;
		$
		\item[iv)] ${\rm A}_{y_0}{\rm B}_{y_0}\,=\,{\rm B}_{y_0}{\rm A}_{y_0}\,=\,{\rm B}_{y_0}$, \qquad  ${\rm A}_{y_0}^2\,=\,{\rm A}_{y_0}$;
		\item [v)] $
		(u_1|u_2|0)\,[	\nabla_x \Theta(0)]^{-1}\, {\rm A}_{y_0}\,=\,(u_1|u_2|0) \,[	\nabla_x \Theta(0)]^{-1}$ for all $u_1,u_2\in \mathbb{R}^3$;
		\item[vi)] ${\rm C}_{y_0}\in \mathfrak{so}(3)$, ${\rm C}_{y_0}^2\,=\,-{\rm A}_{y_0}$ and it has the simplified form
		$
		{\rm C}_{y_0}:=Q_0(0)\begin{footnotesize}\begin{pmatrix}
		0&1&0 \\
		-1&0&0 \\
		0&0&0
		\end{pmatrix}\end{footnotesize}Q_0^T(0);
		$
		\item[vii)] ${\rm B}_{y_0}\,=\,-{\rm C}_{y_0}\,Q_0(0)\, \Big(\mathrm{axl}(Q_0^T(0)\,\partial_{x_1} Q_0(0))\,|\, \mathrm{axl}(Q_0^T(0)\,\partial_{x_2} Q_0(0))\,|\,0\,\Big) \,[\nabla_x \Theta(0)]^{-1} $.
	\end{itemize}
\end{proposition}
\begin{proof}
	\textit{i)} We deduce
	\begin{align}
	{\rm A}_{y_0}&\,=\,(\nabla y_0|0) \,[\nabla_x \Theta(0)]^{-1}\,=\,(\nabla y_0|n_0)\,\id_2^{\flat }\,[\nabla_x \Theta(0)]^{-1}\,=\,
	[\nabla_x \Theta(0)]^{-T}\,[\nabla_x \Theta(0)]^T\,[\nabla_x \Theta(0)]\,\id_2^{\flat }\,[\nabla_x \Theta(0)]^{-1}
	\vspace{1.5mm}\notag\\
	&\,=\,
	[\nabla_x \Theta(0)]^{-T}\,\widehat{\rm I}_{y_0}\,\id_2^{\flat }\,[\nabla_x \Theta(0)]^{-1}\,=\,
	[\nabla_x \Theta(0)]^{-T}\,{\rm I}_{y_0}^{\flat }\,[\nabla_x \Theta(0)]^{-1}.
	\end{align}
	Therefore, the first identity of i) is proven and it also follows that 
	\begin{align}
	\tr({\rm A}_{y_0})&\,=\, \bigl\langle  [\nabla_x \Theta(0)]^{-T}\,{\rm I}_{y_0}^{\flat }\,[\nabla_x \Theta(0)]^{-1},\id_3\bigr\rangle 
	\,=\, \bigl\langle  {\rm I}_{y_0}^{\flat },[\nabla_x \Theta(0)]^{-1} \,[\nabla_x \Theta(0)]^{-T}\bigr\rangle \,=\, \bigl\langle  {\rm I}_{y_0}^{\flat },({\rm I}_{y_0}^{-1})^{\flat }\bigr\rangle \,=\,2.
	\end{align}
We use that $\nabla_x\Theta(x_3)\in{\rm  GKC}$ and calculate
	\begin{align}
	(\nabla y_0|0)(\nabla y_0|n_0)^{-1}&\,=\,(\nabla y_0|n_0)(\nabla y_0|n_0)^{-1}- (0|0|n_0)(\nabla y_0|n_0)^{-1}\notag \\&\,=\,\id_3- (0|0|n_0)(\nabla y_0|n_0)^{-1}\,=\,\id_3- (0|0|n_0)U_0^{-1}(0)\,Q_0^T(0) \\
	&\,=\,\id_3- (0|0|n_0)\,\begin{footnotesize}\begin{pmatrix}
	* &* &0  \\
	* &* &0 \\
	0 &0 &1
	\end{pmatrix}\end{footnotesize}\,Q_0^T(0)\,=\,\id_3- (0|0|n_0)\,Q_0^T(0)\notag \\
	&\,=\,\id_3- (0|0|n_0)\,(d_1^0(0)|d_2^0(0)|n_0)^T\,=\,\id_3- (0|0|n_0)\,(0|0|n_0)^T\, .\notag
	\end{align}
	The last identity of i) follows directly from \eqref{AB}.
	
	\textit{ii)} In terms of the second fundamental form the tensor ${\rm B}_{y_0}$ has the form
	\begin{align}
	{\rm B}_{y_0}\,=\,-[\nabla_x \Theta(0)]^{-T} (\nabla y_0|n_0)^{T}(\nabla n_0|0)[\nabla_x \Theta(0)]^{-1}\,=\,[\nabla_x \Theta(0)]^{-T} {\rm II}_{y_0}^\flat [\nabla_x \Theta(0)]^{-1}\, ,
	\end{align}
	from which symmetry follows, since $\, {\rm II}_{y_0}^\flat\,$ is symmetric.
	Moreover, we have
	\begin{align}
	[\nabla_x \Theta(0)]^{-1}{\rm B}_{y_0}\,=\,&[\nabla_x \Theta(0)]^{-1}[\nabla_x \Theta(0)]^{-T} {\rm II}_{y_0}^\flat [\nabla_x \Theta(0)]^{-1}  \,=\,\left(
	\begin{array}{ccc}
	& {\rm I}_{y_0}^{-1}& 0 \\
	& &0 \\
	0&0&1
	\end{array}
	\right)\left(
	\begin{array}{ccc}
	& {\rm II}_{y_0}& 0 \\
	& &0 \\
	0&0&0
	\end{array}
	\right)[\nabla_x \Theta(0)]^{-1}
	\,=\,{\rm L}_{y_0}^\flat[\nabla_x \Theta(0)]^{-1}\, .
	\end{align}
	
The identity ${\det}[{\rm B}_{y_0}]=0$ follows directly from \eqref{AB}. 	A direct consequence of the above relation is
	\begin{align}\label{BH}
	{\rm tr}({\rm B}_{y_0})\,=\,&{\rm tr}[(\nabla_x \Theta(0)){\rm L}_{y_0}^\flat[\nabla_x \Theta(0)]^{-1}]\,=\,{\rm tr}({\rm L}_{y_0}^\flat)\,=\,{\rm tr}({\rm L}_{y_0})\,=\,2\,{\rm H}\,.
	\end{align}
To compute $ {\rm Cof}\,{\rm B}_{y_0} $ we use that $ {\rm Cof}\,(XY)= {\rm Cof}\,(X) \,{\rm Cof}\,(Y)  $ for any $ X,Y\in\mathbb{R}^{3\times 3}  $. It follows
	\begin{align}
	{\rm Cof}\,{\rm B}_{y_0} &\,=\, {\rm Cof}\,[\nabla_x \Theta(0)] \;{\rm Cof}\,({\rm L}_{y_0}^\flat) \; {\rm Cof}\,[\nabla_x \Theta(0)]^{-1}
	\notag \\&\,=\,
	{\det}\,[\nabla_x \Theta(0)]\cdot
	[\nabla_x \Theta(0)]^{-T} \;{\rm Cof}\, \begin{footnotesize}\begin{pmatrix}
	&  {\rm L}_{y_0}&0  \\
	& &0 \\
	0 &0 &0
	\end{pmatrix}\end{footnotesize}\,[\nabla_x \Theta(0)]^{T} {\det}\,[\nabla_x \Theta(0)]^{-1}
	 \\
	&\,=\,
	[\nabla_x \Theta(0)]^{-T}  \,\begin{footnotesize}\begin{pmatrix}
	0&0&0 \\
	0&0&0 \\
	0&0& {\rm K}
	\end{pmatrix}\end{footnotesize} \,[\nabla_x \Theta(0)]^{T}
	\,=\,(0|0|n_0) \,\begin{footnotesize}\begin{pmatrix}
	0&0&0 \\
	0&0&0 \\
	0&0& {\rm K}
	\end{pmatrix}\end{footnotesize}\, (0|0|n_0)^T
	\;=\;
	[\nabla_x \Theta(0)]  \,\begin{footnotesize}\begin{pmatrix}
	0&0&0 \\
	0&0&0 \\
	0&0& {\rm K}
	\end{pmatrix}\end{footnotesize} \,[\nabla_x \Theta(0)]^{-1}
	\, .\notag
	\end{align}
	The relation for $ {\rm B}_{y_0}^2 $ can be proved similarly without difficulties.

	\textit{iii)} In the following we prove the Cayley-Hamilton type equation. Using ii), we deduce 
	\begin{align}
	{\rm B}_{y_0}^2-2\,{\rm H}\, {\rm B}_{y_0}+{\rm K}\, {\rm A}_{y_0}&\,=\,\nabla_x \Theta(0) ({\rm L}_{y_0}^\flat)^2 [\nabla_x \Theta(0)]^{-1}-2\,{\rm H}\, \nabla_x \Theta(0) {\rm L}_{y_0}^\flat [\nabla_x \Theta(0)]^{-1}+{\rm K}\,(\nabla y_0|n_0)\,\id_2^\flat \,[\nabla_x \Theta(0)]^{-1}\notag \\
	&\,=\,\nabla_x \Theta(0) ({\rm L}_{y_0}^2 -2\,{\rm H}\,  {\rm L}_{y_0} +{\rm K}\,\id_2)^\flat \,[\nabla_x \Theta(0)]^{-1}\,=\,0_3.
	\end{align}

	\textit{iv)} In order to prove iv), we deduce from i) and ii)
	\begin{align}
	{\rm A}_{y_0}{\rm B}_{y_0}& \,=\,[\nabla_x \Theta(0)]^{-T} {\rm I}_{y_0}^\flat [\nabla_x \Theta(0)]^{-1} \; \nabla_x \Theta(0) {\rm L}_{y_0}^\flat [\nabla_x \Theta(0)]^{-1}
	 \\\nonumber
	& \,=\,[\nabla_x \Theta(0)]^{-T} {\rm I}_{y_0}^\flat  {\rm L}_{y_0}^\flat [\nabla_x \Theta(0)]^{-1}
	\,=\,
	[\nabla_x \Theta(0)]^{-T} {\rm II}_{y_0}^\flat [\nabla_x \Theta(0)]^{-1}\,=\,{\rm B}_{y_0}\, .
	\end{align}
	Moreover, we have
	\begin{align}\label{ba=b}
	{\rm B}_{y_0}{\rm A}_{y_0}  & \,=\, {\rm B}_{y_0} ( \id_3-(0|0|n_0)\,(0|0|n_0)^T )
	\,=\, {\rm B}_{y_0} -    \nabla_x \Theta(0) {\rm L}_{y_0}^\flat [\nabla_x \Theta(0)]^{-1} (0|0|n_0)\,(0|0|n_0)^T \, .
	\end{align}
	We notice that 
	\begin{align}\label{ono2}
	 [\nabla_x \Theta(0)]^{-1} (0|0|n_0) =\, (0|0|e_3)
	 \Rightarrow \qquad & [\nabla_x \Theta(0)]^{-1} (0|0|n_0)\,(0|0|n_0)^T\,=\, (0|0|e_3)\,(0|0|n_0)^T \,=\, (0|0|n_0)^T 
	 \\
	\qquad \Rightarrow \qquad& {\rm L}_{y_0}^\flat [\nabla_x \Theta(0)]^{-1} (0|0|n_0)\,(0|0|n_0)^T \,=\, {\rm L}_{y_0}^\flat (0|0|n_0)^T \,=\, (0|0|0)\,=\,0_3\,.\notag
	\end{align}
	Using \eqref{ono2} in \eqref{ba=b} we obtain $ {\rm B}_{y_0}{\rm A}_{y_0} = {\rm B}_{y_0} $. Similarly, from \eqref{ono2}  we find
	\begin{align}\nonumber
	{\rm A}^2_{y_0}  & \,=\, {\rm A}_{y_0} ( \id_3-(0|0|n_0)\,(0|0|n_0)^T )
	\,=\, {\rm A}_{y_0} -    [\nabla_x \Theta(0)]^{-T} {\rm I}_{y_0}^\flat [\nabla_x \Theta(0)]^{-1} (0|0|n_0)\,(0|0|n_0)^T 
	 \\\nonumber
	&\,=\, {\rm A}_{y_0} -    [\nabla_x \Theta(0)]^{-T} {\rm I}_{y_0}^\flat  (0|0|n_0)^T  \,=\, {\rm A}_{y_0}\, .
	\end{align}
	
	\textit{v)} We consider two vectors $u_1,u_2\in \mathbb{R}^3$ and we compute
	\begin{align}
	(u_1|u_2|0)\,[	\nabla_x \Theta(0)]^{-1}\, {\rm A}_{y_0}&\,=\,(u_1|u_2|0) \,[	\nabla_x \Theta(0)]^{-1} -(u_1|u_2|0) \,[	\nabla_x \Theta(0)]^{-1}(0|0|\nabla_x \Theta(0)\, e_3)\,[	\nabla_x \Theta(0)]^{-1} \\
	&\,=\,(u_1|u_2|0) \,[	\nabla_x \Theta(0)]^{-1} -(u_1|u_2|0) \,(0|0|e_3)\,[	\nabla_x \Theta(0)]^{-1}\,=\,(u_1|u_2|0) \,[	\nabla_x \Theta(0)]^{-1}.\notag
	\end{align}
	
	\textit{vi)} Regarding  item vi),  remark that
	\begin{align}
	{\rm C}_{y_0}&\,=\,{\rm Cof}	(\nabla_x \Theta(0))\,\begin{footnotesize}\begin{pmatrix}
	0&1&0 \\
	-1&0&0 \\
	0&0&0
	\end{pmatrix}\end{footnotesize}\,  [	\nabla_x \Theta(0)]^{-1} \\&\,=\,Q_0(0)\,(\det U_0(0))\,  U_0^{-1}(0)\,\begin{footnotesize}\begin{pmatrix}
	0&1&0 \\
	-1&0&0 \\
	0&0&0
	\end{pmatrix}\end{footnotesize}\, U_0^{-1}(0)\, Q_0^T(0)\,=\,Q_0(0)\begin{footnotesize}\begin{pmatrix}
	0&1&0 \\
	-1&0&0 \\
	0&0&0
	\end{pmatrix}\end{footnotesize}Q_0(0)^T,\notag
	\end{align}
	since $\det Q_0\,=\,1$ and
	\begin{equation}
	(\det U_0(0))\, U_0^{-1}(0)\, \,\begin{footnotesize}\begin{pmatrix}
	0&1&0 \\
	-1&0&0 \\
	0&0&0
	\end{pmatrix}\end{footnotesize}\,  U_0^{-1}(0) \,=\,\begin{footnotesize}\begin{pmatrix}
	0&1&0 \\
	-1&0&0 \\
	0&0&0
	\end{pmatrix}\end{footnotesize}.
	\end{equation}
	The alternator tensor has the 
	representation Proposition \ref{propAB} v),
	which shows that ${\rm C}_{y_0}$ is antisymmetric.
	
	Moreover, we deduce 
	\begin{align}
	{\rm C}_{y_0}^2&\,=\,Q_0(0)\,\begin{footnotesize}\begin{pmatrix}
	0&1&0 \\
	-1&0&0 \\
	0&0&0
	\end{pmatrix}\end{footnotesize}\,\begin{footnotesize}\begin{pmatrix}
	0&1&0 \\
	-1&0&0 \\
	0&0&0
	\end{pmatrix}\end{footnotesize}\,  Q_0^T(0)\,=\,-Q_0(0)\,U_0(0)\begin{footnotesize}\begin{pmatrix}
	1&0&0 \\
	0&1&0 \\
	0&0&0
	\end{pmatrix}\end{footnotesize}\,  U_0^{-1}(0)Q_0^T(0) \\
	&\,=\,-(\nabla y_0|\,n_0)\begin{footnotesize}\begin{pmatrix}
	1&0&0 \\
	0&1&0 \\
	0&0&0
	\end{pmatrix}\end{footnotesize}\,[\nabla_x \Theta(0)]^{-1}\,=\, - (\nabla y_0|0)\,[\nabla_x \Theta(0)]^{-1}\,=\,-{\rm A}_{y_0}.\notag\end{align}
	
	\textit{vii)} Let us compute
	\begin{align}
	Q_0^T(0)\,\partial_{x_\alpha} Q_0(0)&\,=\,\Big(d_1^{0}(0)\,|\, d_2^0(0)\,|\,d_3^0(0)\Big)^T\Big(\partial_{x_\alpha}d_1^{0}(0)\,|\, \partial_{x_\alpha}d_2^0(0)\,|\,\partial_{x_\alpha}d_3^0(0)\Big)\notag \\
	&\,=\,\begin{footnotesize}\begin{pmatrix}
	0& \bigl\langle  d_1^{0}(0), \partial_{x_\alpha}d_{2}^{0}(0)\bigr\rangle  & \bigl\langle  d_1^{0}(0), \partial_{x_\alpha}d_{3}^{0}(0)\bigr\rangle   \\
	 \bigl\langle  d_2^{0}(0), \partial_{x_\alpha}d_{1}^{0}(0)\bigr\rangle &0& \bigl\langle  d_2(0), \partial_{x_\alpha}d_{3}(0)\bigr\rangle  \\
	 \bigl\langle  d_3^{0}(0), \partial_{x_\alpha}d_{1}^{0}(0)\bigr\rangle & \bigl\langle  d_3^{0}(0), \partial_{x_\alpha}d_{2}^{0}(0)\bigr\rangle & 0
	\end{pmatrix}\end{footnotesize},
	 \\
	{\rm axl}(Q_0^T(0)\,\partial_{x_\alpha} Q_0(0))&\,=\,\Big(- \bigl\langle  d_2^{0}(0), \partial_{x_\alpha} d_{3}^{0}(0)\bigr\rangle  \,| \, \bigl\langle  d_1^{0}(0), \partial_{x_\alpha}d_{3}^{0}(0)\bigr\rangle  \,| \,- \bigl\langle  d_1^0(0), \partial_{x_\alpha}d_{2}^{0}(0)\bigr\rangle \Big)^T,\qquad \qquad \alpha\,=\,1,2.\notag
	\end{align}
	Hence, we obtain 
	\begin{align}
	{\rm C}_{y_0}\,Q_0(0)\,& \Big(\mathrm{axl}(Q_0^T(0)\,\partial_{x_1} Q_0(0))\,|\, \mathrm{axl}(Q_0^T(0)\,\partial_{x_2} Q_0(0))\,|\,0\,\Big) \,[\nabla_x \Theta(0)]^{-1}\notag \\
	&\,=\,Q_0(0)\, \begin{footnotesize}\begin{pmatrix}
	0&1&0 \\
	-1&0&0 \\
	0&0&0
	\end{pmatrix}\end{footnotesize}\,
	\begin{footnotesize}\begin{pmatrix}
	- \bigl\langle  d_2^{0}(0), \partial_{x_1} d_{3}^{0}(0)\bigr\rangle  &- \bigl\langle  d_2^{0}(0), \partial_{x_2} d_{3}^{0}(0)\bigr\rangle  &0 \\
	 \bigl\langle  d_1^{0}(0), \partial_{x_1}d_{3}^{0}(0)\bigr\rangle & \bigl\langle  d_1^{0}(0), \partial_{x_2}d_{3}^{0}(0)\bigr\rangle &0 \\
	- \bigl\langle  d_1^0(0), \partial_{x_1}d_{2}^{0}(0)\bigr\rangle &- \bigl\langle  d_1^0(0), \partial_{x_2}d_{2}^{0}(0)\bigr\rangle &0
	\end{pmatrix}\end{footnotesize}\,[\nabla_x \Theta(0)]^{-1}
	\end{align}
	\begin{align}
	&\,=\,Q_0(0)\,
	\begin{footnotesize}\begin{pmatrix}
	 \bigl\langle  d_1^{0}(0), \partial_{x_1}d_{3}^{0}(0)\bigr\rangle & \bigl\langle  d_1^{0}(0), \partial_{x_2}d_{3}^{0}(0)\bigr\rangle &0 \\
	 \bigl\langle  d_2^{0}(0), \partial_{x_1} d_{3}^{0}(0)\bigr\rangle  & \bigl\langle  d_2^{0}(0), \partial_{x_2} d_{3}^{0}(0)\bigr\rangle  &0 \\
	0&0&0
	\end{pmatrix}\end{footnotesize}\,[\nabla_x \Theta(0)]^{-1}\notag
	 \\
	&\,=\,Q_0(0)\,Q_0^T(0)\, \Big(\partial_{x_1}d_{3}^{0}(0)\,|\, \partial_{x_2}d_{3}^{0}(0)\,|\,0\Big)\,[\nabla_x \Theta(0)]^{-1}\,=\,(\nabla n_0|0)\,[\nabla_x \Theta(0)]^{-1}\,=\,-{\rm B}_{y_0}.\qedhere\notag
	\end{align}

\end{proof}

\subsection{Neumann condition  on the transverse boundary: an alternative approach}\label{AppendixNeumann}

One can also assume that  on the transverse boundary (upper and lower face of the fictitious Cartesian configuration $\Omega_h$) the Neumann condition 
\begin{equation}\label{bc01}
S_1\bigg( \nabla_x \varphi \big(x_1,x_2,\pm \frac{h}{2}\big), \overline{R}_s\big(x_1,x_2,\pm \frac{h}{2}\big)\bigg)\,(\pm e_3)\,=\,0
\end{equation}
holds.
Using \eqref{pk0}, this implies 
\begin{equation}\label{bcn1}
\bigl\langle  {T}_{\rm Biot}\bigg( \overline{U}_{e,s} \big(x_1,x_2,\pm \frac{h}{2}\big)\bigg)\, n_0,n_0\bigr\rangle\,=\,0\, .
\end{equation}

Hence, for these boundary conditions, we obtain the following linear algebraic system of  equations for the two unknown functions $\varrho_m^e$ and $\varrho_b^e$
\begin{align}
\dd\varrho_m^e\bigg[&(\lambda+2\mu)\pm \lambda\frac{h}{2} \bigl\langle  \overline{Q}_{e} ^T(\nabla (\,\overline{Q}_{e} \nabla_x\Theta(0)\, e_3)|0),\left[\nabla_x\Theta(\pm \frac{h}{2})\right]^{-T}\bigr\rangle \bigg] \\\nonumber
&\dd+\varrho_b^e\bigg[\pm\frac{h}{2}(\lambda+2\mu)+\lambda \frac{h^2}{8} \bigl\langle  \overline{Q}_{e} ^T(\nabla (\,\overline{Q}_{e} \nabla_x\Theta(0)\, e_3)|0),\left[\nabla_x\Theta(\pm \frac{h}{2})\right]^{-T}\bigr\rangle \bigg] \\\nonumber
&\ \ \ \  \dd\,=\,\lambda+2\mu-\lambda[ \bigl\langle  \overline{Q}_{e}^T (\nabla m|0),\left[\nabla_x\Theta(\pm \frac{h}{2})\right]^{-T}\bigr\rangle -2]\, .
\end{align}

Further,  the dependence of  $\nabla_x\Theta(x_1,x_2, x_3)$ on $x_3$ will be denoted by $\nabla_x\Theta(x_3)$.
The above linear algebraic system is equivalent with the  system
\begin{align}
\dd\varrho_m^e\bigg[&2(\lambda+2\mu)+ \lambda\frac{h}{2} \bigl\langle  \overline{Q}_{e}^T(\nabla (\,\overline{Q}_{e}\nabla_x\Theta(0)\, e_3)|0),\left[\nabla_x\Theta(\frac{h}{2})\right]^{-T}-\left[\nabla_x\Theta(-\frac{h}{2})\right]^{-T}\bigr\rangle \bigg] \nonumber \\\nonumber
&\dd+\varrho_b^e\bigg[\lambda \frac{h^2}{8} \bigl\langle  \overline{Q}_{e}^T(\nabla (\,\overline{Q}_{e}\nabla_x\Theta(0)\, e_3)|0),\left[\nabla_x\Theta(\frac{h}{2})\right]^{-1}+\left[\nabla_x\Theta(-\frac{h}{2})\right]^{-T}\bigr\rangle \bigg] \nonumber \\
&\ \ \ \ \dd\,=\,2(\lambda+2\mu)-\lambda[ \bigl\langle  \overline{Q}_{e}^T(\nabla m|0),\left[\nabla_x\Theta(\frac{h}{2})\right]^{-T}+\left[\nabla_x\Theta(-\frac{h}{2})\right]^{-T}-4] , 
 \\
\dd\varrho_m^e & \lambda\frac{h}{2} \bigl\langle  \overline{Q}_{e}^T(\nabla (\,\overline{Q}_{e}\nabla_x\Theta(0)\, e_3)|0),\left[\nabla_x\Theta(\frac{h}{2})\right]^{-T}+\left[\nabla_x\Theta(-\frac{h}{2})\right]^{-T}\bigr\rangle \bigg] \nonumber \\\nonumber
&\dd+\varrho_b^e\bigg[{h}(\lambda+2\mu)+\lambda \frac{h^2}{8} \bigl\langle  \overline{Q}_{e}^T(\nabla (\,\overline{Q}_{e}\nabla_x\Theta(0)\, e_3)|0),\left[\nabla_x\Theta(\frac{h}{2})\right]^{-T}-\left[\nabla_x\Theta(-\frac{h}{2})\right]^{-T}\bigr\rangle \bigg] \nonumber \\
&\ \ \ \ \dd\,=\,-\lambda \bigl\langle  \overline{Q}_{e}^T(\nabla m|0),\left[\nabla_x\Theta(\frac{h}{2})\right]^{-T}-\left[\nabla_x\Theta(-\frac{h}{2})\right]^{-T}\bigr\rangle \, .\notag
\end{align}

The determinant of this system is
$
\dd\delta_1(h)\,=\,h\,\delta_2(h),
$ 
where
\begin{align}
\delta_2(h)\,=&\,
2\,(\lambda+2\mu)^2+\frac{3h}{4}\lambda(\lambda+2\mu) \bigl\langle  \overline{Q}_{e}^T(\nabla (\,\overline{Q}_{e}\nabla_x\Theta(0)\, e_3)|0),\left[\nabla_x\Theta(\frac{h}{2})\right]^{-T}-\left[\nabla_x\Theta(-\frac{h}{2})\right]^{-T}\bigr\rangle   \\\nonumber
&\dd-\frac{h^2\lambda^2}{4} \bigl\langle  \overline{Q}_{e}^T (\nabla (\,\overline{Q}_{e}\nabla_x\Theta(0)\, e_3)|0),\left[\nabla_x\Theta(\frac{h}{2})\right]^{-T}\bigr\rangle  \bigl\langle  \overline{Q}_{e}^T(\nabla (\,\overline{Q}_{e}\nabla_x\Theta(0)\, e_3)|0),\left[\nabla_x\Theta(-\frac{h}{2})\right]^{-T}\bigr\rangle \,>0\ \  \text{for}\ \ h\ll 1 .
\end{align}
We define the  quantities
\begin{align}
M^{h,+}&\,=\,\overline{Q}_{e}^T(\nabla (\,\overline{Q}_{e}\nabla_x\Theta(0)\, e_3)|0)\left\{\left[\nabla_x\Theta(\frac{h}{2})\right]^{-1}+\left[\nabla_x\Theta(-\frac{h}{2})\right]^{-1}\right\},\notag \\ M^{h,-}&\,=\,\overline{Q}_{e}^T(\nabla (\,\overline{Q}_{e}\nabla_x\Theta(0)\, e_3)|0)\left\{\left[\nabla_x\Theta(\frac{h}{2})\right]^{-1}-\left[\nabla_x\Theta(-\frac{h}{2})\right]^{-1}\right\}, \\
N^{h,+}&\,=\,\overline{Q}_{e}^T(\nabla m|0)\left\{\left[\nabla_x\Theta(\frac{h}{2})\right]^{-1}+\left[\nabla_x\Theta(-\frac{h}{2})\right]^{-1}\right\},\qquad  N^{h,-}\,=\,\overline{Q}_{e}^T(\nabla m|0)\left\{\left[\nabla_x\Theta(\frac{h}{2})\right]^{-1}-\left[\nabla_x\Theta(-\frac{h}{2})\right]^{-1}\right\}\,. \notag
\end{align}

Then, the exact solution is given by
\begin{equation}\hspace{-1cm}
\begin{array}{crl}
&&\dd\left(\!
\begin{array}{l}
\varrho_m^e\vspace{2mm}\\
\varrho_b^e
\end{array}\!\!\right)\,=\,\delta_1(h)^{-1}\left(\!\!\begin{array}{ll}
(\lambda+2\mu)+\lambda \frac{h}{8} \bigl\langle  M^{h,-},\id_3\bigr\rangle  
&
-\lambda \frac{h}{8} \bigl\langle  M^{h,+},\id_3\bigr\rangle \vspace{2mm}\nonumber \\
-\frac{\lambda}{2} \bigl\langle  M^{h,+},\id_3\bigr\rangle  
&  (\lambda+2\mu)\frac{2}{h}+ \frac{\lambda}{2} \bigl\langle  M^{h,-},\id_3\bigr\rangle 
\end{array}\!\!\!\right)\vspace{3mm}
\left(\begin{array}{l}
2(\lambda+2\mu)-\lambda[ \bigl\langle   N^{h,+},\id_3\bigr\rangle -4]\vspace{2mm}\nonumber \\
-\lambda \bigl\langle  N^{h,-},\id_3\bigr\rangle 
\end{array}\!\!\right)\, .
\end{array}
\end{equation}
Consider the quantities $
C_1 \,=\,(\nabla y_0|n_0), C_2 \,=\,(\nabla n_0|0)\, .
$ 
Then
$
\nabla_x \Theta(x_3)\,=\,C_1+C_2 \,x_3\, .
$
Moreover
\begin{equation}\label{x3det}
[\nabla_x \Theta(x_3)]^{-1}\,=\,[\id_3 +x_3C_1^{-1}C_2]^{-1}C_1^{-1}\,=\,[\id_3-x_3C_1^{-1}C_2+x_3^2(C_1^{-1}C_2)^2-x_3^3(C_1^{-1}C_2)^3+...]C_1^{-1}\, .
\end{equation}
Due to the fact that $C_1$ and $C_2$ do not depend on $x_3$, we have
\begin{align}\label{exdet}
\left[\nabla_x \Theta(\dd\frac{h}{2})\right]^{-1}-\left[\nabla_x \Theta(-\frac{h}{2})\right]^{-1}&\,=\,-2\,
\left[\dd\frac{h}{2}C_1^{-1}C_2+\bigg(\frac{h}{2}\bigg)^3(C_1^{-1}C_2)^3+...\right]C_1^{-1}, \\
\left[\nabla_x \Theta(\dd\frac{h}{2})\right]^{-1}+\left[\nabla_x \Theta(-\frac{h}{2})\right]^{-1}&\,=\,2\,
\left[\id_3+\dd\bigg(\frac{h}{2}\bigg)^2(C_1^{-1}C_2)^2+\bigg(\frac{h}{2}\bigg)^4(C_1^{-1}C_2)^4+...\right]C_1^{-1}.\notag
\end{align}
Hence, we deduce
\begin{align}
M^{h,+}&\,=\,2\,\overline{Q}_{e}^T(\nabla (\,\overline{Q}_{e}\nabla_x\Theta(0)\, e_3)|0)[\nabla_x\Theta(0)]^{-1}+o(h^2),\notag \\ 
N^{h,+}&\,=\,2\,\overline{Q}_{e}^T(\nabla m|0)[\nabla_x\Theta(0)]^{-1}+o(h^2), \\ M^{h,-}&\,=\,o(h),\quad N^{h,-}\,=\,-h\,\overline{Q}_{e}^T(\nabla m|0)  [\nabla_x\Theta(0)]^{-1} \,(\nabla n_0|0)\,[\nabla_x\Theta(0)]^{-1} +o(h^2),\notag \\
\delta_1&\,=\,
2(\lambda+2\mu)^2+o(h).\notag
\end{align}
and 
\begin{align}\label{exprho}
\varrho_m^e\,=\,\frac{1}{\delta_1(h)}\Big\{ &[(\lambda+2\mu)+\lambda \frac{h}{8} \bigl\langle  M^{h,-},\id_3\bigr\rangle  
][2(\lambda+2\mu)-\lambda[ \bigl\langle   N^{h,+},\id_3\bigr\rangle -4]+[\lambda \frac{h}{8} \bigl\langle  M^{h,+},\id_3\bigr\rangle ][\lambda \bigl\langle  N^{h,-},\id_3\bigr\rangle ]  \Big\}, \\
\varrho_b^e\,=\,\frac{1}{\delta_1(h)}\Big\{ &[-\frac{\lambda}{2} \bigl\langle  M^{h,+},\id_3\bigr\rangle  ][2(\lambda+2\mu)-\lambda[ \bigl\langle   N^{h,+},\id_3\bigr\rangle -4]-[ (\lambda+2\mu)\frac{2}{h}+ \frac{\lambda}{2} \bigl\langle  M^{h,-},\id_3\bigr\rangle ][\lambda \bigl\langle  N^{h,-},\id_3\bigr\rangle ]  \Big\}.\notag
\end{align}

If we  take $h\rightarrow 0$ (as appropriate for a very thin shell) in the  expressions (\ref{exprho}) of $\varrho_m^e$ and $\varrho_b^e$, then we obtain

\begin{align}\label{exprho1}
\varrho_m^e\,=\,&\frac{1}{
	2(\lambda+2\mu)^2}\Big\{ (\lambda+2\mu)
[2(\lambda+2\mu)-\lambda[2\, \bigl\langle   \overline{Q}_{e}^T(\nabla m|0)[\nabla_x\Theta(0)]^{-1},\id_3\bigr\rangle -4] \Big\},\notag \\
\,=\,& 1-\frac{\lambda}{\lambda+2\mu}[\, \bigl\langle   \overline{Q}_{e}^T(\nabla m|0)[\nabla_x\Theta(0)]^{-1},\id_3\bigr\rangle -2] ,\notag \\
\varrho_b^e\,=\,&\frac{1}{
	2(\lambda+2\mu)^2}\Big\{ -\frac{\lambda}{2} \bigl\langle  2\,\overline{Q}_{e}^T(\nabla (\,\overline{Q}_{e}\nabla_x\Theta(0)\, e_3)|0)[\nabla_x\Theta(0)]^{-1},\id_3\bigr\rangle  \notag \\&\qquad\qquad\qquad\qquad\times \Big[2(\lambda+2\mu)-\lambda[ \bigl\langle  2\,\overline{Q}_{e}^T(\nabla m|0)[\nabla_x\Theta(0)]^{-1},\id_3\bigr\rangle -4]\Big]\notag \\&\qquad\qquad\qquad+ 2\,(\lambda+2\mu) \lambda \bigl\langle \overline{Q}_{e}^T(\nabla m|0)  [\nabla_x\Theta(0)]^{-1} \,(\nabla n_0|0)\,[\nabla_x\Theta(0)]^{-1},\id_3\bigr\rangle ]  \Big\},\notag \\\,=\,&
-\frac{\lambda}{
	\lambda+2\mu} \bigl\langle  \overline{Q}_{e}^T(\nabla (\,\overline{Q}_{e}\nabla_x\Theta(0)\, e_3)|0)[\nabla_x\Theta(0)]^{-1},\id_3\bigr\rangle  
 \\&+ \frac{\lambda}{
	\lambda+2\mu}  \bigl\langle \overline{Q}_{e}^T(\nabla m|0)  [\nabla_x\Theta(0)]^{-1} \,(\nabla n_0|0)\,[\nabla_x\Theta(0)]^{-1},\id_3\bigr\rangle ] \notag \\
&+\frac{\lambda^2}{
	(\lambda+2\mu)^2} \bigl\langle  \overline{Q}_{e}^T(\nabla (\,\overline{Q}_{e}\nabla_x\Theta(0)\, e_3)|0)[\nabla_x\Theta(0)]^{-1},\id_3\bigr\rangle  [ \bigl\langle  \,\overline{Q}_{e}^T(\nabla m|0)[\nabla_x\Theta(0)]^{-1},\id_3\bigr\rangle -2]\Big].\notag
\end{align}
Ignoring the 
term $\frac{\lambda^2}{(\lambda+2\mu)^2} \bigl\langle  \overline{Q}_{e}^T (\nabla (\,\overline{Q}_{e}\nabla_x\Theta(0)\, e_3)|0)[\nabla_x\Theta(0)]^{-1},\id_3\bigr\rangle [ \bigl\langle  \overline{Q}_{e}^T(\nabla m|0)[\nabla_x\Theta(0)]^{-1},\id_3\bigr\rangle -2]$ we obtain the same expressions for  $\varrho_m^e$ and $\varrho_b^e$ as in \eqref{final_rho}, but here considering the Neumann condition \eqref{bc01} instead of the approximated boundary conditions \eqref{conditii}. Therefore, our result in \eqref{final_rho} is correct.

\end{footnotesize}
\end{document}